%% file: bipGraphDrawing+.tex
\documentclass[a4paper,12pt]{article}

\usepackage{graphicx,xspace,amsmath,enumitem,verbatim}
\usepackage{wrapfig,subfigure}
\usepackage{amsfonts,amssymb,bm,amsthm,fullpage,charter}
\usepackage{caption}

\input{psfig.sty}

\usepackage[boldsans]{concmath} %
\usepackage{euler} %
\usepackage{rsfso}

\newtheorem{lemma}[equation]{Lemma}
\newtheorem{corollary}[equation]{Corollary}
\newtheorem{proposition}[equation]{Proposition}
\newtheorem{theorem}[equation]{Theorem}

\newtheorem{observation}[equation]{Observation}

\newtheorem{problem}{Problem}

\def\Case#1{\medskip\ni{\bf Case\ }#1{\bf.}}
\def\Fact{\medskip\ni{\bf Fact.\ }}

\def\ITEMMACRO #1 ??? #2 ???{\smallskip\par\noindent
\hangindent=#2em\setbox0\hbox{#1 \kern5pt}%
\ifdim\wd0<\hangindent\setbox0\hbox to\hangindent{\hss#1\quad}\fi%
\box0\ignorespaces}

\def\Item#1{\ITEMMACRO {\rm #1} ??? 1.8 ???}
\def\FreeItem#1{\ITEMMACRO {#1} ??? 1.8 ???}

\let\Bitem=\bItem
\def\BrackItem[#1]{\ITEMMACRO [#1] ??? 1.8 ???}

\def\ni{\noindent}
\def\qedclaim{\hfill$\triangle$\smallskip}

\usepackage[centertableaux,mathmode,notabloids]{ytableau}

\newcommand{\triple}[6]{
\hbox{\small
\begin{ytableau}
\none[#1] & #4        & #5 \\
\none     & \none[#2] & #6 \\
\none     & \none     & \none[#3]
\end{ytableau}
}}

\newcommand{\nitriple}[3]{
\hbox{\small
\begin{ytableau}
    #1  & #2 \\
\none   & #3 
\end{ytableau}
}}

\DeclareMathOperator{\type}{type}
\DeclareMathOperator{\inv}{inv}
\DeclareMathOperator{\id}{id}
\def\ovl{\overline}

\def\ered{e_{red}}
\def\egreen{e_{green}}
\def\eblue{e_{blue}}

\graphicspath{{./Figures/}}
\title{Topological Drawings of Complete Bipartite Graphs}
\author{Jean Cardinal~\thanks{Universit\'e libre de Bruxelles (ULB), Brussels, Belgium. {\tt jcardin@ulb.ac.be}.}
  \and Stefan Felsner~\thanks{Technische Universit\"at Berlin, Germany. {\tt felsner@math.tu-berlin.de}.}}

\date{February 2017}

\begin{document}
\maketitle

\ni\hrulefill

\begin{abstract}
  Topological drawings are natural representations of graphs in the
  plane, where vertices are represented by points, and edges by curves
  connecting the points.
  Topological drawings of complete graphs and of complete bipartite
  graphs have been studied extensively in the context of crossing number problems.
  We consider a natural class of simple topological drawings of {\em complete bipartite} graphs,
  in which we require that one side of the vertex set bipartition lies on the outer boundary of the drawing.

  We investigate the combinatorics of such drawings.  For this
  purpose, we define combinatorial encodings of the drawings by
  enumerating the distinct drawings of subgraphs isomorphic to
  $K_{2,2}$ and $K_{3,2}$, and investigate the constraints they must
  satisfy. We prove that a drawing of $K_{k,n}$ exists if and only
  if some simple local conditions are satisfied by the encodings.
  This directly yields a polynomial-time algorithm for deciding
  the existence of such a drawing given the encoding. We show the encoding
  is equivalent to specifying which pairs of edges cross, yielding
  a similar polynomial-time algorithm for the realizability of abstract
  topological graphs.

  We also completely characterize and enumerate such drawings of $K_{k,n}$
  in which the order of the edges around each vertex is the same for vertices
  on the same side of the bipartition.
  Finally, we investigate drawings of $K_{k,n}$ using straight lines and pseudolines,
  and consider the complexity of the corresponding realizability problems.
  \end{abstract}

\tableofcontents

%%%%%%%%%%%%%%%%%%%%%%%%%%%%%%%%%%%%%%%%%%%%%%%%%%%%%%%%%%%%%%%%%%%%%%%%%%%
\section{Introduction}
%%%%%%%%%%%%%%%%%%%%%%%%%%%%%%%%%%%%%%%%%%%%%%%%%%%%%%%%%%%%%%%%%%%%%%%%%%%

We consider {\em topological graph drawings}, which are drawings of
simple undirected graphs where vertices are represented by points in
the plane, and edges are represented by simple curves that connect the
corresponding points. We typically restrict those drawings to satisfy
some natural nondegeneracy conditions. In particular, we consider {\em
  simple} drawings, in which every pair of edges intersect at most
once. A common vertex counts as an intersection. 

While being perhaps the most natural and the most used representations
of graphs, simple drawings are far from being understood from the
combinatorial point of view. A prominent illustration is the problem of
identifying the minimum number of edge crossings
in a simple topological drawing of $K_n$~\cite{HH63,BK64,AAFRS13} or
of $K_{k,n}$~\cite{Z54,CRS13}, for which there are long standing conjectures.

In order to cope with the inherent complexity of the drawings, it is
useful to consider combinatorial abstractions. Those abstractions are
discrete structures encoding some features of a drawing.  One such
abstraction, introduced by Kratochv\'il, Lubiw, and Ne\v{s}et\v{r}il,
is called {\em abstract topological graphs} (AT-graph)~\cite{KLN91}.
An AT-graph consists of a graph $(V,E)$ together with a set ${\cal
  X}\subseteq \binom{E}{2}$.  A topological drawing is said to {\em
  realize} an AT-graph if the pairs of edges that cross are exactly
those in $\cal X$.  Another abstraction of a topological drawing is
called the {\em rotation system}. The rotation system associates a
circular permutation with every vertex $v$, which in a realization
must correspond to the order in which the neighbors of $v$ are
connected to $v$.  Natural realizability problems are:
given an AT-graph or a rotation system, is it realizable as a
topological drawing?  The realizability problem for AT-graphs
is known to be NP-complete~\cite{KM89}.

For simple topological drawings of complete graphs, the two
abstractions are actually equivalent~\cite{PT06}. It is possible
to reconstruct the set of crossing pairs of edges by looking at the
rotation system, and vice-versa (up to reversal of all permutations).
Kyn\v cl recently proved the
remarkable result that a complete AT-graph (an AT-graph for which the
underlying graph is complete) can be realized as a simple topological
drawing of $K_n$ if and only if all the AT-subgraphs on at most 6
vertices are realizable~\cite{K11,K15}.  This directly yields a
polynomial-time algorithm for the realizability problem.  While this
provides a key insight on topological drawings of complete graphs,
similar realizability problems already appear much more difficult when
they involve complete {\em bipartite} graphs.  In that case, knowing
the rotation system is not sufficient for recontructing the
intersecting pairs of edges.

We propose a fine-grained analysis of simple topological drawings of
complete bipartite graphs. In order to make the analysis more
tractable, we introduce a natural restriction on the drawings, by
requiring that one side of the vertex set bipartition lies on a circle
at infinity. This gives rise to meaningful, yet complex enough,
combinatorial structures.

%%%%%%%%%%%%%%%%%%%%%%%%%%%%%%%%%%%%%%%%%%%%%%%%%%%%%%%%%%%%%%%%%%%%%%%%%%%
\paragraph{\bf Definitions.} 

We wish to draw the complete bipartite graph $K_{k,n}$ in the plane in
such a way that:
\begin{enumerate}
\item vertices are represented by points,
\item edges are continuous curves that connect those points, and do
  not contain any other vertices than their two endpoints
\item no more than two edges intersect in one point,
\item edges pairwise intersect at most once; in particular, edges
  incident to the same vertex intersect only at this vertex,
\item the $k$ vertices of one side of the bipartition lie on the outer
  boundary of the drawing.
\end{enumerate}

Properties 1--4 are the usual requirements for {\em simple topological
  drawings} also known as {\em good drawings}. As we will see,
property~5 leads to drawings with interesting combinatorial
structures. We will refer to drawings satisfying
properties 1--5 as {\em outer drawings}. Since this is the only type
of drawings we consider, we will use the single term {\em drawing}
instead when the context is clear.

The set of vertices of a bipartite graph $K_{k,n}$ will be denoted by
$P\cup V$, where $P$ and~$V$ are the two sides of the bipartition,
with $|P|=k$ and $|V|=n$. When we consider a given drawing, we will
use the word ``vertex'' and ``edge'' to denote both the vertex or edge
of the graph, and their representation as points and curves.  Without
loss of generality, we can assume that the $k$ outer vertices
$p_1,\ldots,p_k$ lie in clockwise order on the boundary of a disk that
contains all the edges, or on the line at infinity. The vertices of
$V$ are labeled $1,\ldots,n$. An example of such a drawing is given in
Figure~\ref{fig:example-1}.  \medskip

%%%%%%%%%%%%%%%%%%%%%%%%%%%%%%%%%%%%%%%%%%%%%%%%%%%
%%
% in einem figure environment mit caption
   \calc_figscale{14}
    \begin{figure}[htb]
    \centerline{\input{\path/example-1.pstex_t}}
    \caption{\label{fig:example-1}}
    \end{figure}
    VC
{Two outer drawings of $K_{3,5}$. In both drawings the rotation system is $(12345,21435,13254)$.}
%%
%%%%%%%%%%%%%%%%%%%%%%%%%%%%%%%%%%%%%%%%%%%%%%%%%%%

The {\em rotation system} of the drawing is a sequence of $k$
permutations on $n$ elements associated with the vertices of $P$ in
clockwise order. For each vertex of $P$, its permutation encodes the
(say) counterclockwise order in which the $n$ vertices of $V$ are
connected to it. Due to our last constraint on the drawings, the
rotations of the $k$ vertices of~$P$ around each vertex of $V$ are
fixed and identical, they reflect the clockwise order of
$p_1,\ldots,p_k$ on the boundary.

Unlike for complete graphs, the rotation system of an outer drawing of a
complete bipartite graph does not completely determine which pairs of
edges are intersecting. This is exemplified with the two drawings in
Figure~\ref{fig:example-1}.

%%%%%%%%%%%%%%%%%%%%%%%%%%%%%%%%%%%%%%%%%%%%%%%%%%%%%%%%%%%%%%%%%%%%%%%%%%%
\paragraph{\bf Results.}

The paper is organized as follows. In Section~\ref{sec:uniform}, we
consider outer drawings with a {\em uniform} rotation system, in which the
$k$ permutations of the vertices of~$P$ are all equal to the
identity. In this case, we can state a general structure theorem that
allows us to completely characterize and count outer drawings of arbitrary
bipartite graphs $K_{k,n}$.

In Section~\ref{sec:k=2}, we consider outer drawings of $K_{2,n}$ with
arbitrary rotation systems. We consider a natural combinatorial
encoding of such drawings, and state two necessary consistency
conditions involving triples and quadruples of points in~$V$.  We show
that these conditions are also sufficient, yielding a polynomial-time
algorithm for checking consistency of a drawing.

We also observe that our encoding is equivalent to specifying which
pairs of edges must intersect in the drawing, hence exactly encodes the
corresponding AT-graph. Therefore we show as a corollary that we
can decide the realizability of a given AT-graph with underlying graph
isomorphic to $K_{2,n}$ in polynomial time.

In Section~\ref{sec:k=3} and \ref{sec:arbk}, we extend these results,
first to outer drawings of $K_{3,n}$, then to outer drawings of $K_{k,n}$.
We prove that simple consistency conditions on triples and quadruples are
sufficient for drawings of $K_{k,n}$, yielding again a polynomial-time
algorithm for consistency checking.

In Section~\ref{sec:extendable}, we consider outer drawings with the additional
property that the edges can be extended into a pseudoline arrangement,
which we refer to as {\em extendable} drawings.
We give a necessary and sufficient condition for the existence
of an extendable outer drawing of a complete bipartite graph given the
rotation system. We also touch upon the even more restricted problem
of finding straight-line outer drawings with prescribed rotation systems.

%%%%%%%%%%%%%%%%%%%%%%%%%%%%%%%%%%%%%%%%%%%%%%%%%%%%%%%%%%%%%%%%%%%%%%%%%%%%%%%%
\section{Outer Drawings with uniform rotation system}\label{sec:uniform}
%%%%%%%%%%%%%%%%%%%%%%%%%%%%%%%%%%%%%%%%%%%%%%%%%%%%%%%%%%%%%%%%%%%%%%%%%%%%%%%%%

We first consider the case where $k$ is arbitrary but the rotation
system is uniform, that is, the permutation around each of the $k$
vertices $p_i$ is the same. Without loss of generality we assume
that this permutation is the identity permutation on $[n]$. 

In a given outer drawing, each of the $n$ vertices of $V$ splits the plane
into $k$ regions $Q_1, Q_2, \ldots ,Q_k$, where each $Q_i$ is bounded
by the edges from $v$ to $p_i$ and $p_{i+1}$, with the understanding
that $p_{k+1}=p_1$.  We denote by
$Q_i(v)$ the $i$th region defined by vertex~$v$ and further on call
these regions {\em quadrants}. We let $\type (a,b)
= i$, for $a,b\in V$ and $i\in [k]$, whenever $a\in Q_i(b)$. 
This implies that $b\in Q_i(a)$, see Figure~\ref{fig:quadrants}.
Indeed if $a<b$ and $j\neq i+1$, then edge $p_{i+1}b$ has to intersect
all the edges $p_ja$, while edge $p_jb$ has to
avoid $p_{i+1}b$ until they meet in $b$. It follows that
none of the edges $p_jb$ can intersect~$p_{i+1}a$.
This shows that $a \in Q_i(b)$.

%%%%%%%%%%%%%%%%%%%%%%%%%%%%%%%%%%%%%%%%%%%%%%%%%%%
%%
% in einem figure environment mit caption
   \calc_figscale{30}
    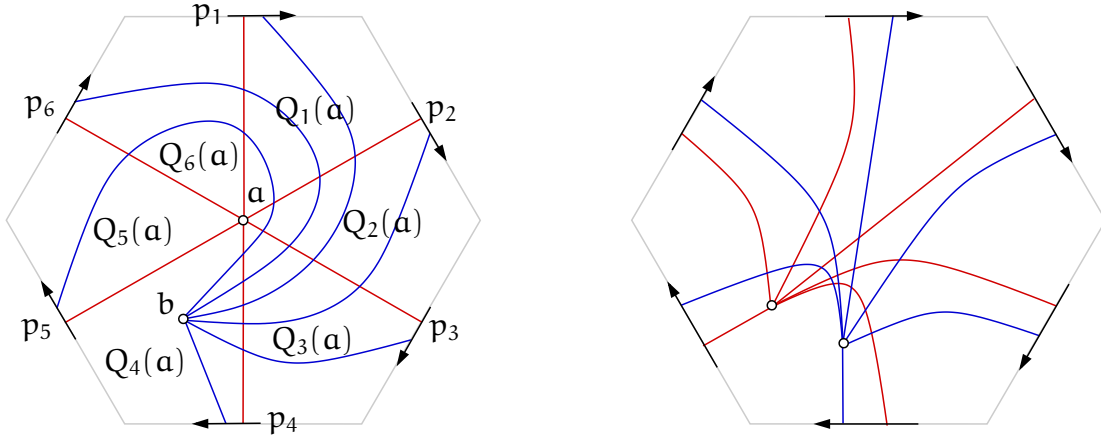
\begin{figure}[htb]
    \centerline{\input{\path/quadrants.pstex_t}}
    \caption{Having placed $b$ in $Q_4(a)$ the crossing
  pairs of edges and the order of crossings on each edge is prescribed. 
  In particular $a\in Q_4(b)$. On the right a symmetric outer drawing of the
pair.\label{fig:quadrants}}
    \end{figure}
    
%%
%%%%%%%%%%%%%%%%%%%%%%%%%%%%%%%%%%%%%%%%%%%%%%%%%%%

\begin{observation}[Symmetry]\quad\\
For all $a$, $b$ in uniform rotation systems:\quad $\type (a,b) = \type (b,a).$
\end{observation}

For the case $k=2$, we have exactly two types of pairs, that we will
denote by $A$ and $B$.  The two types are illustrated on
Figure~\ref{fig:pseudolines-1}. Note that the two types can be
distinguished by specifying which are the pairs of intersecting edges.
The outer drawings of $K_{2,n}$ with uniform rotations can be 
viewed as {\em colored pseudoline
  arrangements}, where:
\begin{itemize}
\item each pseudoline is split into two segments of distinct colors,
\item no crossing is monochromatic.
\end{itemize}
%%%%%%%%%%%%%%%%%%%%%%%%%%%%%%%%%%%%%%%%%%%%%%%%%%%
%%
% in einem figure environment mit caption
   \calc_figscale{30}
    \begin{figure}[htb]
    \centerline{\input{\path/pseudolines-1.pstex_t}}
    \caption{The two types of pairs for outer drawings of
  $K_{2,n}$ with uniform rotation systems.\label{fig:pseudolines-1}}
    \end{figure}
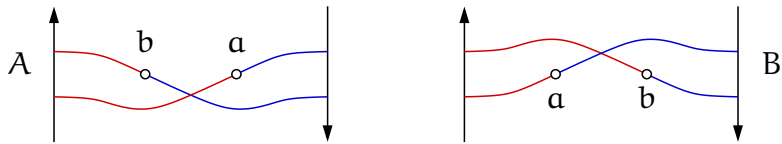
    
%%
%%%%%%%%%%%%%%%%%%%%%%%%%%%%%%%%%%%%%%%%%%%%%%%%%%%

This is illustrated on Figure~\ref{fig:pseudolines-2}.  The pseudoline
of a vertex $v\in V$ is denoted by~$\ell(v)$. The left (red) and right
(blue) parts of this pseudoline are denote by $\ell_L (v)$
and~$\ell_R(v)$.  Now having $\type (a,b) = \type (b,a) = A$ means
that $b$ lies {\bf a}bove $\ell(a)$ and $a$ lies {\bf
  a}bove~$\ell(b)$. While having $\type (a,b) = \type (b,a) = B$ means
that $b$ lies {\bf b}elow $\ell(a)$ and $a$ lies {\bf
  b}elow~$\ell(b)$.

%%%%%%%%%%%%%%%%%%%%%%%%%%%%%%%%%%%%%%%%%%%%%%%%%%%
%%
\def\QTableOne{
\begin{ytableau}
\none[1]  &  B        & B        & B \\
\none     & \none[2]  & B        & A \\
\none     & \none     & \none[3] & A \\
\none     & \none     &\none     & \none[4]
\end{ytableau}
}
% in einem figure environment mit caption
   \calc_figscale{30}
    \begin{figure}[htb]
    \centerline{\input{\path/pseudolines-2.pstex_t}}
    \caption{Drawing $K_{2,4}$ as a colored pseudoline
  arrangement. The type of each pair is given in the table on the right.\label{fig:pseudolines-2}}
    \end{figure}
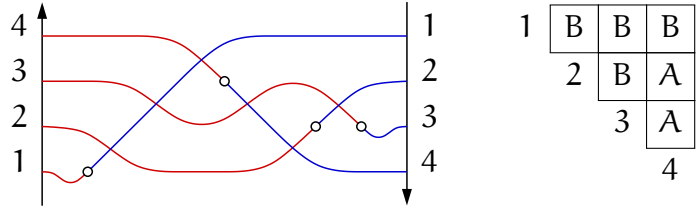

%%%%%%%%%%%%%%%%%%%%%%%%%%%%%%%%%%%%%%%%%%%%%%%%%%%%%%%%%%%%%%%%%%%%%%%%%%%%%%%%
\subsection{The triple rule}

\begin{lemma}[Triple rule]
     \label{lem:triple}\quad\\
For uniform rotation systems and three vertices $a,b,c\in V$ with $a<b<c$\\[1mm]
\centerline{$ \type (a,c)\in \{ \type (a,b), type (b,c) \}.$}
\end{lemma}

\begin{proof}
{\bf Case $k=2$.}
If $\type(a,b) \neq \type(b,c)$ there is nothing to show since there are only two types.
Without loss of generality, suppose that $\type (a,b)=\type (b,c)=B$. 
This situation is illustrated in the left part of Figure~\ref{fig:triple-both}. The pseudoline 
$\ell (c)$ must cross $\ell(b)$ on $\ell_R(b)$, otherwise we would
have $\type (b,c)=A$. Hence the point $c$ is on the right of this intersection.
Pseudoline 
$\ell (a)$ must cross $\ell(b)$ on $\ell_L(b)$, and $a$ is left of
this intersection. It follows that  $\ell (a)$ and  $\ell (c)$ cross
on $\ell_R(a)$ and $\ell_L(c)$, i.e.,~$\type (a,c)=B$.
%%%%%%%%%%%%%%%%%%%%%%%%%%%%%%%%%%%%%%%%%%%%%%%%%%%
%%
% in einem figure environment mit caption
   \calc_figscale{30}
    \begin{figure}[htb]
    \centerline{\input{\path/triple-both.pstex_t}}
    \caption{Illustrations for the $k=2$ case of
Lemma~\ref{lem:triple} (left), and the $k>2$ case of
Lemma~\ref{lem:triple} (right).\label{fig:triple-both}}
    \end{figure}
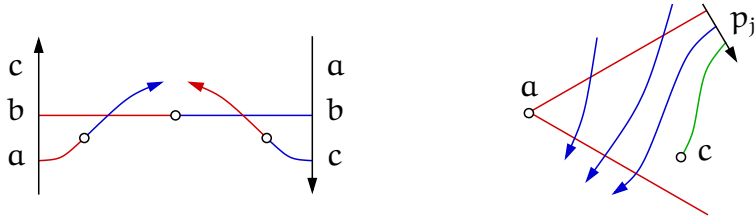
    
%%
%%%%%%%%%%%%%%%%%%%%%%%%%%%%%%%%%%%%%%%%%%%%%%%%%%%

{\bf Case $k>2$.}  For the general case assume that $\type(a,b)= i$
and $\type(a,c) = j$. If $i=j$ there is nothing to show. Now suppose
$i\neq j$. From $c\in Q_j(a)$ it follows that $p_{j+1}a$ and $p_jc$
are disjoint. Edges $p_jb$ and $p_jc$ only share the endpoint $p_j$,
hence $c$ has to be in the region delimited by $p_jb$ and $p_{j+1}a$,
see the right part of Figure~\ref{fig:triple-both}. This region is
contained in $Q_j(b)$, whence $\type(b,c) = j$.
\end{proof}

%%%%%%%%%%%%%%%%%%%%%%%%%%%%%%%%%%%%%%%%%%%%%%%%%%%%%%%%%%%%%%%%%%%%%%%%%%
\subsection{The quadruple rule}

\begin{lemma}[Quadruple rule]\quad\\\label{lem:qrule}
For four vertices $a,b,c,d\in V$ with $a<b<c<d$:\\
if $\;\type (a,c) = \type (b,c) = \type (b,d) = X$ then $\type (a,d) = X$.
\end{lemma}
\begin{proof}
  {\bf Case $k=2$.}  Suppose, without loss of generality, that
  $X=B$. Consider the pseudolines representing $b$ and $c$ with their
crossing at $\ell_R(b)\cup \ell_L(c)$. Coming from the left the edge
  $\ell_L(d)$ has to avoid $\ell_L(c)$ and therefore intersects
  $\ell_R(b)$. On $\ell_R(b)$ the crossing with $\ell_L(c)$ is left of
  the crossing with $\ell_L(d)$, see
  Figure~\ref{fig:quadruple}. Symmetrically
from the right the
edge $\ell_R(a)$ has to intersect $\ell_L(c)$ and this intersection is
left of  $\ell_R(b)\cup \ell_L(c)$. To reach the crossings with
 $\ell_L(c)$ and $\ell_R(b)$ edges $\ell_R(a)$ and $\ell_L(d)$
have to intersect, hence, $\type (a,d)=B$.

%%%%%%%%%%%%%%%%%%%%%%%%%%%%%%%%%%%%%%%%%%%%%%%%%%%
%%
% in einem figure environment mit caption
   \calc_figscale{30}
    \begin{figure}[htb]
    \centerline{\input{\path/quadruple.pstex_t}}
    \caption{Illustration for the $k>2$ case of Lemma~\ref{lem:triple}.\label{fig:quadruple}}
    \end{figure}
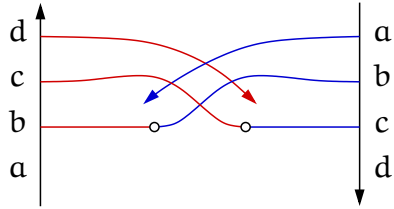
    
%%
%%%%%%%%%%%%%%%%%%%%%%%%%%%%%%%%%%%%%%%%%%%%%%%%%%%

{\bf Case $k>2$.}  In the general case, we let $X=i$, and consider the
pseudoline arrangement defined by the two successive vertices $p_i$
and $p_{i+1}$ of $P$ defining the quadrants $Q_i$.  Proving that
$\type (a,d)=i$, that is, that $a\in Q_i(d)$, can be done as above for
$k=2$ on the drawing of $K_{2,n}$ induced by $\{p_i, p_{i+1}\}$ and
$V$.
\end{proof}

%%%%%%%%%%%%%%%%%%%%%%%%%%%%%%%%%%%%%%%%%%%%%%%%%%%%%%%%%%%%%%%%%%%%%%%%%%%%%%%%
\subsection{Decomposability and Counting}

We can now state a general structure theorem for all outer drawings of $K_{k,n}$ with 
uniform rotation systems.

\begin{theorem}\label{thm:unif->dec}
  Consider the complete bipartite graph $G$ with vertex bipartition $(P,V)$ such that
  $|P|=k$ and $|V|=n$.  
Given a type in $[k]$ for each pair of vertices in $V$, there exists an outer drawing
of $G$ realizing those types with a uniform rotation system if and only if:
\begin{enumerate}
\item there exists $s\in \{2,\ldots,n\}$ and $X\in [k]$ such that $\type (a,b)
  = X$ for all pairs $a,b$ with $a< s$ and $b \geq s$, (in the table
  this corresponds to maximal rectangle whose cells have all the same entry)
\item the same holds recursively when the interval $[1,n]$ is replaced
  by any of the two intervals~$[1,s-1]$ and $[s,n]$.
\end{enumerate}
\end{theorem}

\begin{proof} 
{\bf $(\Rightarrow)$} Let us first show that if there exists a
drawing, then the types must satisfy the above structure.  We proceed
by induction on $n$. Pick the smallest $s \in \{2,\ldots,n\}$ such
that $\type (1,b) = \type(1,s)$ for all $b\geq s$. Set $X := \type
(1,s)$.  We claim that $\type(a,b)=X$ for all $a,b$ such that $1\leq
a<s\leq b\leq n$. For $a=1$ this is just the condition on $s$. Now let
$1<a$.

First suppose that $\type(1,a)\not= X$. We can apply the triple rule
on the indices $1,a,b$. Since $\type(1,b)\in\{\type(1,a),\type(a,b)\}$,
we must have that $\type(a,b)=X$.

Now suppose that $\type(1,a)=X$.  We have $\type(1,s-1) = Y \not= X$
by definition. As in the previous case we obtain $\type (s-1, b)=X$
from the triple rule for $1,s-1,b$. Applying the triple rule on
$1,a,s-1$ yields that $\type(a,s-1)=Y$.

Now apply the quadruple rule on $1,a,s-1,b$. We know that
$\type(1,s-1)=\type(a,s-1)=Y$, and by definition $\type(1,b)=X$.
Hence we must have that $\type(a,b)\not= Y$.

Finally, apply the triple rule on $a,s-1,b$. We know that
$\type(a,s-1)=Y$, $\type(s-1,b)=X$. Since $\type (a,b)\not= Y$, we
must have $\type(a,b)=X$. This yields the claim.

{\bf $(\Leftarrow)$} Now given the recursive structure, it is not
difficult to construct a drawing. Consider the two subintervals as a
single vertex, then recursively blow up these two vertices. (See
Figure~\ref{fig:exuni} for an illustration).
\end{proof}

%%%%%%%%%%%%%%%%%%%%%%%%%%%%%%%%%%%%%%%%%%%%%%%%%%%
%%
\begin{figure}[h!]
\begin{center}
\includegraphics[scale=.9,page=2]{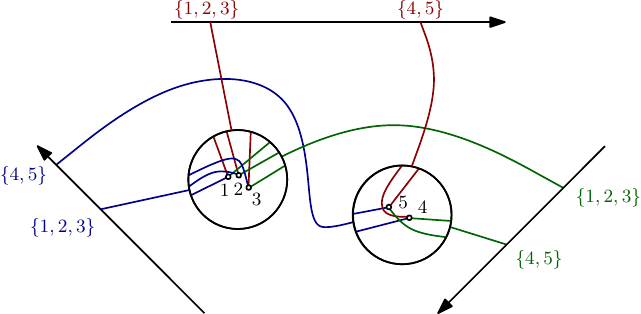}
\end{center}
\caption{\label{fig:exuni}Illustration of the recursive structure of
the outer drawings in the uniform case.}
\end{figure}
%%
%%%%%%%%%%%%%%%%%%%%%%%%%%%%%%%%%%%%%%%%%%%%%%%%%%%

The recursive structure yields a corollary on the number of distinct drawings.

\begin{corollary}[Counting outer drawings with uniform rotation systems]
  \label{cor:countunif}
  For every pair of integers $k,n>0$ denote by $T(k,n)$ the number of
  outer drawings of the complete bipartite graph isomorphic to
  $K_{k,n}$ with uniform rotation systems. Then
  $$
  T(n+1, k+1) = \sum_{j=0}^{n} \tbinom{n+j}{2j}\,C_j\,k^j
  $$
  where $C_j$ is the $j$th Catalan number.
\end{corollary}
\begin{proof}
The recursive structure can be modeled in a labeled binary tree. The root
corresponds to $[1,n]$, the subtrees correspond to the
intervals~$[1,s-1]$ and $[s,n]$,
and the label of the root is $\type(a,b)$ for $a<s\leq b$.
The definition implies that the label of the left child of a node is
different from the label of the node. Leaves have no label.

For the number $T(k,n)$ of labeled binary trees we therefore get a
Catalan-like recursion 
$ 
T(k,n) = k \sum_{i=1}^{n-1} \left(\frac{k-1}{k}\kern3pt T(k, i)\right)
\cdot T(k, n-i) + T(k, n-1).
$
%% Maple:
%% T := proc(k,n)
%% option remember;
%% if n=1 then 1 else (k-1)*sum( 'T(k,i)*T(k,n-i)' , 'i'=1..n-1 ) + T(k,n-1) fi;
%% end;
  The factor $k$ preceeding the sum accounts for the choice of the
  label for the root. Using symmetry on the labels we find that a
  $\frac{k-1}{k}$ fraction of the candidates for the left subtree
  comply with the condition on the labels. The case where the left
  subtree only consists of a single leaf node is exceptional, in this
  case there is no label and we have one choice for this subtree, not
  just $(1-1/k)$. This explains the additional summand. The recursion
  $$
  T(k,n) = T(k,n-1) + (k-1) {\textstyle \sum_{i=1}^{n-1}} T(k, i)
  \cdot T(k,n-i)
  $$
  together with the initial condition $T(k,1) = 1$ yields an array of
  numbers which is is listed as entry A103209 in the encyclopedia of
  integer sequences\footnote{www.oeis.org} (OEIS). The stated
  explicite expression for $T(k,n)$ can be found there. It can be
  verified by induction.
\end{proof}

Note that in the case $k=2$, Corollary~\ref{cor:countunif}
provides a bijection between outer drawings with uniform rotation systems and
combinatorial structures counted by Schr\"oder
numbers, such as separable permutations and guillotine partitions.

%%%%%%%%%%%%%%%%%%%%%%%%%%%%%%%%%%%%%%%%%%%%%%%%%%%%%%%%%%%%%%%%%%%%%%%%%%%%%%%%%%
\section{Outer Drawings with $k=2$}\label{sec:k=2}
%%%%%%%%%%%%%%%%%%%%%%%%%%%%%%%%%%%%%%%%%%%%%%%%%%%%%%%%%%%%%%%%%%%%%%%%%%%%%%%%%%

In this section we deal with outer drawings with $k=2$ and arbitrary rotation system.
We now have three types of pairs, that we call $N$, $A$, and $B$, as
illustrated on Figure~\ref{fig:threetypes-2}.  The type $N$ (for
noncrossing) is new, and is forced whenever the pair corresponds to an
inversion in the two permutations. Note again that the three types exactly
encode which are the pairs of crossing edges.

%%%%%%%%%%%%%%%%%%%%%%%%%%%%%%%%%%%%%%%%%%%%%%%%%%%
%%
% in einem figure environment mit caption
   \calc_figscale{30}
    \begin{figure}[htb]
    \centerline{\input{\path/threetypes-2.pstex_t}}
    \caption{The three types of pairs for outer drawings of
  $K_{2,n}$ with arbitrary rotation systems.\label{fig:threetypes-2}}
    \end{figure}
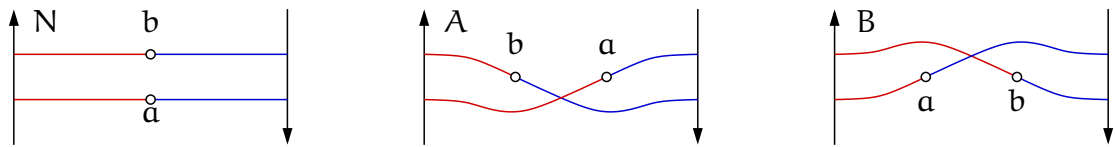
    
%%
%%%%%%%%%%%%%%%%%%%%%%%%%%%%%%%%%%%%%%%%%%%%%%%%%%%

Recall that an outer drawing of $K_{2,n}$, in which no pair is of type $N$, can
be seen as a colored pseudoline arrangement as defined
previously. Similarly, an outer drawing of $K_{2,n}$ in which some pairs are
of type $N$ can be seen as an arrangement of colored monotone curves
crossing pairwise {\em at most} once. We will refer to arrangement of
monotone curves that cross {\em at most} once as {\em quasi-pseudoline
  arrangements}.  The pairs of type $N$ correspond to parallel
pseudolines. Without loss of generality, we can suppose that the first
permutation in the rotation system, that is, the order of the
pseudolines on the left side, is the identity. We denote by $\pi$ the
permutation on the right side.

The first question is whether every permutation $\pi$ is feasible
in the sense that there is a drawing of $K_{2,n}$ such that the
rotations are $(\id,\pi)$. The answer is yes, two easy constructions
are exemplified in Figure~\ref{fig:feasible}

%%%%%%%%%%%%%%%%%%%%%%%%%%%%%%%%%%%%%%%%%%%%%%%%%%%
%%
\def\QTableAB{{\small
\begin{ytableau}
\none[1]  &  N        &          &          & N \\ 
\none     & \none[2]  &          &          &   \\
\none     & \none     & \none[3] &  N       & N \\
\none     & \none     &\none     & \none[4] & N \\
\none     & \none     &\none     &\none     & \none[5]
\end{ytableau}}
}
% in einem figure environment mit caption
   \calc_figscale{20}
    \begin{figure}[htb]
    \centerline{\input{\path/feasible.pstex_t}}
    \caption{Two outer drawings with rotations
  $(\id_5,[3,4,1,5,2])$.
  On the left all non-$N$-types are $B$ on the right they are $A$.\label{fig:feasible}}
    \end{figure}
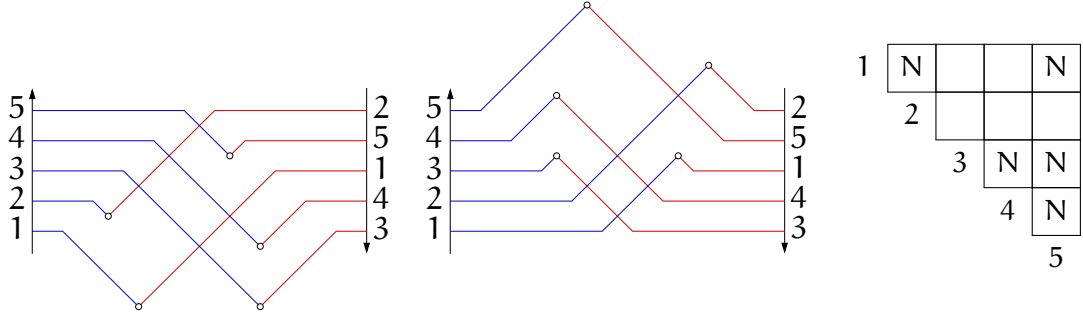
    
%%
%%%%%%%%%%%%%%%%%%%%%%%%%%%%%%%%%%%%%%%%%%%%%%%%%%%

%%%%%%%%%%%%%%%%%%%%%%%%%%%%%%%%%%%%%%%%%%%%%%%%%%%%%%%%%%%%%%%%%%%%%%%%%%%%
\subsection{Triples}

For $a,b,c\in V$, with $a<b<c$, we are interested in the triples of
types $(\type(a,b)$, $\type(a,c)$, $\type(b,c))$ that are possible in
an outer drawing of $K_{2,n}$, such triples are called
\emph{legal}. We like to display triples in little tables, e.g.,
the triple $\type(a,b)=X$, $\type(a,c)=Y$, and $\type(b,c)=Z$ is
represented as
$$\triple{a}{b}{c}{X}{Y}{Z}.$$

\begin{lemma}[decomposable triples]
    \quad\\ \label{lem:dec tri}
A triple with $Y \in \{X,Z\}$ is always legal. There are~15 triples of
this kind.
\end{lemma}

\begin{proof}
If $Y=X$ take a drawing of $K_{2,2}$ of type $X$ with vertices $a,v$
and $a < v$. In this drawing double the pseudoline corresponding to
$v$ and cover vertex $v$ by a small circle. Then plug a drawing of
$K_{2,2}$ of type $Z$ with vertices $b,c$ in this circle. This results
in a drawing of $K_{2,3}$ with the prescribed types. The construction
is very much as in Figure~\ref{fig:exuni}.

There are 3 triples\triple{a}{b}{c}{X}{X}{X} and in addition
for each of the 6 pairs $(X,Z)$ with $X\neq Z$ a triple of each the 2
types\triple{a}{b}{c}{X}{X}{Z}
and\triple{a}{b}{c}{X}{Z}{Z}.
\end{proof}

Note that the triples of the latter lemma are
decomposable in the sense of Theorem~\ref{thm:unif->dec}.

\begin{lemma}\label{lem:two-non-decomp}\quad\\[-4.1mm]
There are exactly two non-decomposable legal triples:
\raisebox{1mm}{$\triple{a}{b}{c}{N}{A}{B}$} and 
\raisebox{1mm}{$\triple{a}{b}{c}{A}{B}{N}$}. 
\end{lemma}

\begin{proof}
  From Lemma~\ref{lem:triple} we know that triples where all entries
  are $A$ or $B$ are decomposable. If $\type(a,b) = N$, then $(a,b)$
  is a non-inversion of $\pi$ while pairs $(a,b)$ with $\type(a,b) \in
  \{A,B\}$ are inversions of $\pi$. Both, the set of inversion pairs
  and the set of non-inversion pairs are transitive. Hence, triples of
  type\nitriple{{}}{N}{{}} and\nitriple{N}{}{N} where empty
  cells represent inversion pairs are impossible. It remains to
  consider the cases where exactly one of $X$ and $Z$ is $N$ and the
  other two symbols in the triple are $A$ and~$B$. Only the two
  triples shown in the statement of the lemma remain.
\end{proof}

With the two lemmas we have classified all 17 legal triples, i.e., all
outer drawings of $K_{2,3}$. 

\begin{observation}[Triple rule]\label{obs:triple-rule}
  Any three vertices of $V$ in an outer drawing of $K_{2,n}$ must induce
  one of the 17 legal triples of types.
\end{observation}

%%%%%%%%%%%%%%%%%%%%%%%%%%%%%%%%%%%%%%%%%%%%%%%%%%%%%%%%%%%%%%%%%%%%%%%%%%%%
\subsection{Quadruples}

We aim at a characterization of collections of types that correspond
to outer drawings. Already in the case of uniform rotations we had to add
Lemma~\ref{lem:qrule}, a condition for quadruples. In the general
case the situation is more complex than in the uniform case,
see Figure~\ref{fig:qruleFail}.

%%%%%%%%%%%%%%%%%%%%%%%%%%%%%%%%%%%%%%%%%%%%%%%%%%%
%%
\def\QTableTwoN{
{\small\begin{ytableau}
\none[_1]  &  A        & N        & B \\
\none     & \none[_2]  & N        & N \\
\none     & \none     & \none[_3] & B \\
\none     & \none     &\none     & \none[_4]
\end{ytableau}}
}
\def\QTableTwoB{
{\small\begin{ytableau}
\none[_1]  &  N        & B        & A \\
\none     & \none[_2]  & B        & B \\
\none     & \none     & \none[_3] & A \\
\none     & \none     &\none     & \none[_4]
\end{ytableau}}
}
% in einem figure environment mit caption
   \calc_figscale{26}
    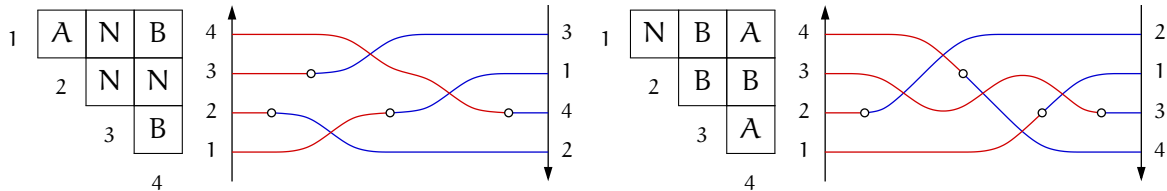
\begin{figure}[htb]
    \centerline{\input{\path/qruleFail.pstex_t}}
    \caption{The quadruple rule from
Lemma~\ref{lem:qrule} does not hold in the presence of $N$ types.\label{fig:qruleFail}}
    \end{figure}
    
%%
%%%%%%%%%%%%%%%%%%%%%%%%%%%%%%%%%%%%%%%%%%%%%%%%%%%

Reviewing the proof of Lemma~\ref{lem:qrule} we see that in the case
discussed there, where given $B$ types are intended to enforce
$\type(a,d)=B$, we need that in~$\pi$ element $a$ is before $b$.
This is equivalent to $\type(a,b)\neq N$. Symmetrically, three
$A$ types enforce $\type(a,d)=A$ when $d$ is the last in $\pi$,
i.e., if $\type(c,d)\neq N$.

\begin{lemma}\label{lem:qrule-2}
  Consider four vertices $a,b,c,d\in V$ such that $a<b<c<d$.\\
  If $\type(a,b)\neq N$ and $\type (a,c) = \type (b,c) = 
  \type (b,d) = B$ then $\type (a,d) = B$.\\
  If $\type(c,d)\neq N$ and $\type (a,c) = \type (b,c) = 
  \type (b,d) = A$ then $\type (a,d) = A$.
\end{lemma}

%%%%%%%%%%%%%%%%%%%%%%%%%%%%%%%%%%%%%%%%%%%%%%%%%%%%%%%%%%%%%%%%%%%%%%%%%%%%%%%%%%%%%
\subsection{Consistency}

With the next theorem we show that consistency on triples and
quadruples is sufficient to grant the existence of an outer drawing.

\begin{theorem}[Consistency of outer drawings for $k=2$]
  \label{thm:consist-2}
    Consider the complete bipartite graph $G$ with vertex bipartition $(P,V)$ such that
  $|P|=2$ and $|V|=n$.  
  Given a type in $\{A,B,N\}$ for each pair of vertices in $V$, there exists an outer
  drawing of $G$ realizing those types if and only if all triples are legal
  and the quadruple rule (Lemma~\ref{lem:qrule-2}) is satisfied.
\end{theorem}

The proof of this result uses the following known result on {\em local
  sequences} in pseudoline arrangements. Given an arrangement of $n$
pseudolines, the local sequences are the permutations $\alpha_i$ of
$[n]\setminus \{ i\}$, $i\in [n]$, representing the order in which the $i$th
pseudoline intersects the $n-1$ others.

\begin{lemma}[Thm. 6.17 in \cite{f-gga-04}]\label{lem:localseq}
  The set $\{ \alpha_i\}_{i\in [n]}$ is the set of local sequences of
  an arrangement of $n$ pseudolines if and only if
  $$
  ij \in \inv(\alpha_k) \Leftrightarrow ik \in \inv(\alpha_j) 
      \Leftrightarrow jk \in \inv(\alpha_i),
  $$
  for all triples $i,j,k$, where $\inv(\alpha)$ is the set of
  inversions of the permutation $\alpha$.
\end{lemma}

\begin{proof}[Proof of Theorem~\ref{thm:consist-2}]
The necessity of the condition was already 
stated in Observation~\ref{obs:triple-rule}

We proceed by giving an algorithm for constructing an appropriate
drawing. First recall from the proof of Lemma~\ref{lem:two-non-decomp}
that having legal triples implies that 
the sets of inversion pairs and its complement, the set of 
non-inversion pairs, are both transitive. Hence, there is a 
well defined permutation $\pi$ representing the rotation at $p_2$.

We aim at defining the local sequences $\alpha_i$ that allow an
application of Lemma~\ref{lem:localseq}. This will yield a pseudoline
arrangement. A drawing of $K_{2,n}$, however, will only correspond to a
quasi-pseudoline arrangement.  Therefore, we first construct a
quasi-pseudoline arrangement $T$ for the pair $(\ovl{\pi},\id)$,
i.e., only the quasi-pseudolines corresponding to $i$ and $j$ with
$\type (i,j)=N$ cross in $T$. The idea is that appending $T$ on the
right side of the quasi-pseudoline arrangement of the drawing yields a 
full pseudoline arrangement.

Now fix $i\in [n]$. Depending on $i$ we partition the set $[n]\setminus i$ into
five parts. For a type~$X$ let $X_<(i) = \{ j : j < i \text{ and }
\type(j,i) = X \}$ and $X_>(i) = \{ j : j > i \text{ and } \type(i,j)
= X \}$, the five relevant parts are $A_<(i)$, $A_>(i)$, $B_<(i)$,
$B_>(i)$, and $N(i) = N_<(i) \cup N_>(i)$.  The pseudoline $\ell_i$
has three parts.  The edge incident to $p_1$ (the red edge) is crossed
by pseudolines $\ell_j$ with $j \in A_>(i) \cup B_<(i)$.  The edge
incident to $p_2$ (the blue edge) is crossed by pseudolines $\ell_j$
with $j \in A_<(i) \cup B_>(i)$.  The part of $\ell_i$ belonging to
$T$ is crossed by pseudolines $\ell_j$ with $j \in N(i)$.  The order
of the crossings in the third part, i.e., the order of crossings with
pseudolines $\ell_j$ with $j\in N(i)$,  is prescribed by $T$.  

Regarding the order of the crossings on the second part we know that
the lines for $j\in A_<(i)$ have to cross $\ell_i$ from left to right in order
of decreasing indices and the lines for $j\in B_>(i)$ have to cross $\ell_i$
from left to right in order of increasing indices, see
Figure~\ref{fig:local}.  If $j\in A_<(i)$ and
$j'\in B_>(i)$, then consistency of triples implies that $\type(j,j')
\in \{A,B\}$. If $\type(j,j')=A$, then on $\ell_i$ the crossing of $j'$
has to be left of the crossing of $j$. If $\type(j,j')=B$, then on
$\ell_i$ the crossing of $j$ has to be left of the crossing of $j'$.

%%%%%%%%%%%%%%%%%%%%%%%%%%%%%%%%%%%%%%%%%%%%%%%%%%%
%%
% in einem figure environment mit caption
   \calc_figscale{38}
    \begin{figure}[htb]
    \centerline{\input{\path/local.pstex_t}}
    \caption{Crossings on the edge $i\,p_2$.\label{fig:local}}
    \end{figure}
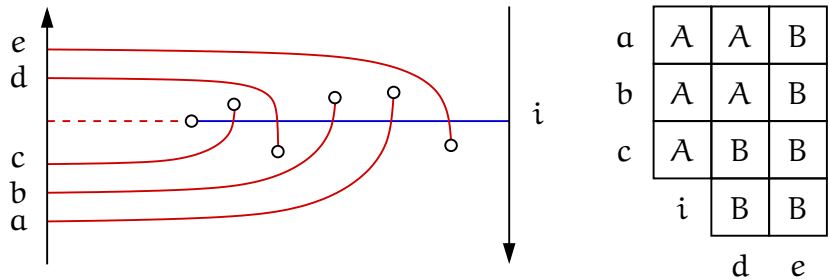
    
%%
%%%%%%%%%%%%%%%%%%%%%%%%%%%%%%%%%%%%%%%%%%%%%%%%%%%

The described conditions yield a {\it ``left--to--right''} relation
$\to_i$ such that for all $x,y\in A_<(i) \cup B_>(i)$ one of $x\to_i
y$ and $y\to_i x$ holds.  We have to show that $\to_i$ is
acyclic. Since $\to_i$ is a tournament it is enough to show that
$\to_i$ is transitive.

Suppose there is a
cycle $x\to_i y\to_i z\to_i x$. If $x,y < i$ and $z > i$, then
$\type(x,i)=\type(y,i) = A$, moreover, from
$x\to_i y$ we get $y < x$ and from $y\to_i z\to_i x$
we get $\type(x,z)= A$, and $\type(y,z)=B$. Since 
$\type(i,z) = B \neq N$ this is a violation of the 
second quadruple rule of Lemma~\ref{lem:qrule-2}.

If $x <i$ and $y,z >i$, then we have $\type(i,y)=\type(i,z) = B$.
From this together with $y\to_i z$ we obtain $y < z$, and
$ z\to_i x \to_i y$ yields
$\type(x,y)= B$, and $\type(x,z)=A$. 
This is a violation of the first quadruple rule of Lemma~\ref{lem:qrule-2}.

Adding the corresponding arguments for the order of crossings on the 
first part of line $\ell_i$ we conclude that the permutation
$\alpha_i$ is uniquely determined by the given types and the choice of
$T$.

The consistency condition on triples of local sequences needed for
the application of Lemma~\ref{lem:localseq}
is trivially satisfied because legal triples of types correspond to
drawings of $K_{2,3}$ and each such drawing together with $T$ 
consists of three pairwise crossing pseudolines.
\end{proof}

Since the condition only involves triples and quadruples of vertices in $V$,
this directly yields a polynomial-time algorithm for consistency checking.
By observing that the information given by the types is equivalent to specifying which
pairs of edges cross, we directly get the analogue statement for AT-graphs.

\begin{corollary}[AT-graph realizability]
  There exists an $O(n^4)$ algorithm for deciding the existence of an outer drawing of an AT-graph whose
  underlying graph is of the form $K_{2,n}$.
\end{corollary}
\begin{proof}
  We can check that the three types of pairs in Figure~\ref{fig:threetypes-2} exactly prescribe which pairs
  of edges cross. Furthermore, given the set of crossing pairs, we can reconstruct the type assignment.
  We can then check that every triple is legal and that the quadruple rule is satisfied in time proportional to the number of triples and quadruples, hence $O(n^4)$.
  \end{proof}

%%%%%%%%%%%%%%%%%%%%%%%%%%%%%%%%%%%%%%%%%%%%%%%%%%%%%%%%%%%%%%%%%%%%%%%%%%%%%%%%
\section{Outer Drawings with $k=3$}\label{sec:k=3}
%%%%%%%%%%%%%%%%%%%%%%%%%%%%%%%%%%%%%%%%%%%%%%%%%%%%%%%%%%%%%%%%%%%%%%%%%%%%%%%%
\def\U#1{B_{#1}}
\def\S#1{W_{#1}}

At the beginning of the previous section we have seen that any pair of
rotations is feasible for outer drawings of $K_{2,n}$. This is not true in
the case of $k> 2$. For $k=4$ the system of rotations
$([1,2],[2,1],[1,2],[2,1])$ is easily seen to be infeasible.  In the
case $k=3$ it is less obvious that infeasible systems of rotations
exist. We also had an efficient characterization 
of consistent assignments of types for $k=2$. We generalize this to $k=3$.

We again start by looking at the types for pairs, i.e., at all
possible outer drawings of~$K_{3,2}$. We already know that if the rotation
system is uniform $(\id_2,\id_2,\id_2)$, then there are three types of
drawings. The other three options $(\id_2,\ovl{\id_2},\ovl{\id_2})$,
$(\id_2,\ovl{\id_2},{\id_2})$, and $(\id_2,{\id_2},\ovl{\id_2})$, each
have a unique drawing. Figure~\ref{fig:sixtypes} shows the six
possible types and associates them to the symbols $\U{\alpha}$, and
$\S{\alpha}$, for $\alpha=1,2,3$.

%%%%%%%%%%%%%%%%%%%%%%%%%%%%%%%%%%%%%%%%%%%%%%%%%%%
%%
% in einem figure environment mit caption
   \calc_figscale{33}
    \begin{figure}[htb]
    \centerline{\input{\path/sixtypes.pstex_t}}
    \caption{The six types of outer drawings of $K_{3,2}$.\label{fig:sixtypes}}
    \end{figure}
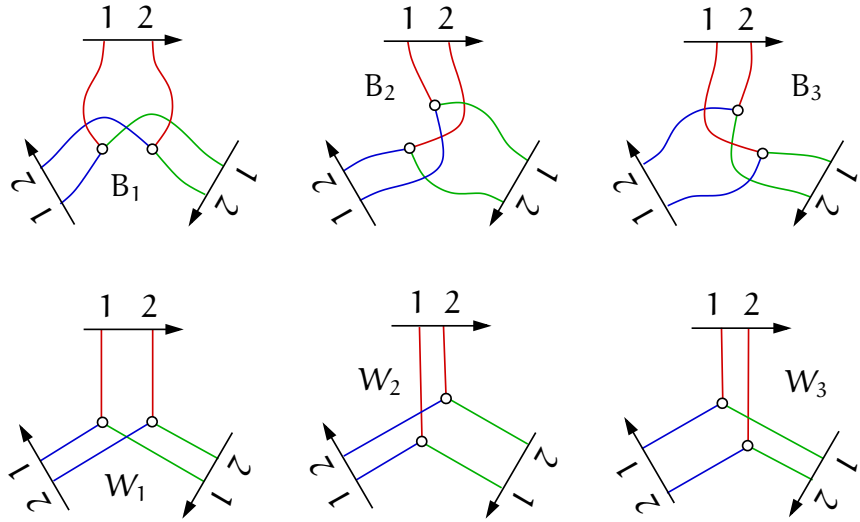
    
%%
%%%%%%%%%%%%%%%%%%%%%%%%%%%%%%%%%%%%%%%%%%%%%%%%%%%

\begin{comment} % do we need this?
The three edges emanating from a vertex $i\in [n]$ partition the drawing area
into three regions. Let $R_\alpha(i)$ be the region bounded by the two edges
$i\,p_{\alpha+1}$ and $i\,p_{\alpha-1}$ not containing $p_\alpha$ 
($R_\alpha(i)$ is the same as the quadrant $Q_{\alpha+1}(i)$ from
Section~\ref{sec:uniform}). When the types have been prescribed for all pairs
of vertices we know which vertices are located in which region of $i$.  The
set of vertices in $R_\alpha(i)$ is denoted $S_\alpha(i)$.  The conditions for
$j\in S_\alpha(i)$ can be read from Figure~\ref{fig:sixtypes}. For example
$j\in S_2(i)$ if either $i,j$ is of type $\U2$ or $j<i$ and the type is $\S1$
or $\S3$.
\end{comment}

This classification allows us to reason about the following simple example.
\begin{proposition}\label{prop:infeasible}
The system $(\id_4,[4,2,1,3],[2,4,3,1])$ is an infeasible set of rotations. 
\end{proposition}
\begin{proof}
The table of types for the given permutations must be of the following form:
\smallskip
\centerline{$
\begin{ytableau}
\none[1]  &  \S1        & \S3        & \S1 \\
\none     & \none[2]    & \U\alpha   & \S2 \\
\none     & \none       & \none[3]   & \S1 \\
\none     & \none       &\none       & \none[4]
\end{ytableau}
$}
Looking at the subtable \nitriple{\S1}{\S3}{\U\alpha} corresponding to
$\{1,2,3\}$ one can realize that there is only one choice for $\alpha$,
namely $\alpha=2$. The same figure shows that the
subtable~\nitriple{\U\alpha}{\S2}{\S1} of $\{2,3,4\}$ again only
allows a unique choice of $\alpha$, namely $\alpha=3$.  This proves
that there is no drawing for this set of rotations.
\end{proof}

Let~$T_1$ be the assignment of types in $\{A,B,N\}$
corresponding to the outer drawing of $K_{2,n}$ with outer vertices $(p_3,p_2)$.
Similarly, let $T_2$ be the table
corresponding to $(p_1,p_3)$ and $T_3$ the table
corresponding to $(p_2,p_1)$.
Note that these assignments determine the rotation system $(\pi_1,\pi_2,\pi_3)$.
An assignment $T$ of types in $\{B_1,B_2,B_3,W_1,W_2,W_3\}$ to every pair of vertices in $V$
translates to the following assignments $T_1,T_2,T_3$:
\begin{table}[h]
\centering 
 \begin{tabular}{l|cccccc}
  \qquad & $B_1$ & $B_2$ & $B_3$ & $W_1$ & $W_2$ & $W_3$ \\[0.5mm]
  \hline\noalign{\smallskip}
   $T_1$ & $B$   & $A$   & $A$   & $A$   & $N$   & $N$   \\ 
   $T_2$ & $A$   & $B$   & $A$   & $N$   & $A$   & $N$   \\ 
   $T_3$ & $A$   & $A$   & $B$   & $N$   & $N$   & $A$   \\ 
 \end{tabular}
 \caption{The projections of the six types of outer drawings of $K_{3,2}$.}
 \label{tab:legal-projections}
\end{table}

This table reveals a dependency between the types in $\{A,B,N\}$ of
every pair of vertices in $V$ for different pairs of vertices in $P$.
This dependency is instrumental in the upcoming consistency theorems.
For an assignment $T$ of types in $\{B_1,B_2,B_3,W_1,W_2,W_3\}$ to pairs of vertices
in $V$, we will refer to the induced assignments $T_1,T_2,T_3$ as the {\em projections}
of $T$.

We are now ready to state our consistency theorem on outer drawings for $k=3$.

%%%%%%%%%%%%%%%%%%%%%%%%%%%%%%%%%%%%%%%%%%%%%%%%%%%
\subsection{The consistency theorem for $k=3$}

\begin{theorem}[Consistency of outer drawings for $k=3$]
  \label{thm:consist-3}
    Consider the complete bipartite graph $G$ with vertex bipartition $(P,V)$ such that
  $|P|=3$ and $|V|=n$. Given the assignments $T_1, T_2, T_3$ of
    types in $\{A,B,N\}$ for the pairs $(p_3,p_2)$, $(p_1,p_3)$,
    and $(p_2,p_1)$ of vertices in $P$, respectively,
    there exists an outer drawing of $G$ realizing those types if and only if
    \begin{itemize}
      \item all three assignments obey the
        triple and quadruple consistency rules in the sense of Theorem~\ref{thm:consist-2},
      \item for any pair of vertices in $V$, the assignments correspond to an
        outer drawing of $K_{3,2}$ in the sense of Table~\ref{tab:legal-projections}.
    \end{itemize}
\end{theorem}

This directly yields the following corollary.

\begin{corollary}[AT-graph realizability]
  There exists an $O(n^4)$ algorithm for deciding the existence of an outer drawing of an AT-graph whose
  underlying graph is of the form $K_{3,n}$.
\end{corollary}
\begin{proof}
  Again, one can check that the three types of pairs in Figure~\ref{fig:sixtypes} exactly prescribe which pairs of edges cross, and given the set of crossing pairs, we can reconstruct the type assignment.
  We can then check that every triple is legal and that the quadruple rule is satisfied in time proportional to the number of triples and quadruples, hence $O(n^4)$.
  \end{proof}

\begin{proof}[Proof of Theorem~\ref{thm:consist-3}]
Let us first note that one direction of the Theorem is easy: if there exists an outer drawing, then the assignments must be consistent. 

We now show that consistency of the type assignments is sufficient for the existence of an outer drawing.
Let $T$ be the assignment of types in $\{B_1,B_2,B_3,W_1,W_2,W_3\}$ to the pairs of $V$ given by Table~\ref{tab:legal-projections}, and let $(\pi_1,\pi_2,\pi_3)$ be the corresponding rotations. 
The consistency of the tables $T_2$ and $T_3$ in the sense of Theorem~\ref{thm:consist-2}
implies that there are drawings~$D_2$ and~$D_3$ of~$K_{2,n}$ realizing the type assignments $T_2$ and $T_3$.
The vertex $p_1$ (the outer vertex with rotation~$\pi_1$) and its edges form a non-crossing star in both drawings.

Let the drawing $D_2$ live on plane $Z_2$ and $D_3$ live on plane $Z_3$
and consider a fixed homeomorphism $\phi$ between the planes.  There
is a homeomorphism $\psi: Z_2 \to Z_2$ such that mapping $D_2$ via
$\phi\circ\psi$ to $Z_3$ yields a superposition of the two drawings
with the following properties.  
{
\Bitem Corresponding vertices are
mapped onto each other.  
\Bitem The stars of $p_1$ are mapped onto
each other, i.e., the edges at $p_1$ of the two
drawings are represented by the same curves.
\Bitem At each vertex $v\in V$ the rotation is correct, i.e., 
we see the edges to $p_1,p_2,p_3$  in clockwise order.
\Bitem The drawing has no touching edges, i.e., when two edges 
meet they properly cross in a single point. 

}\par

\smallskip
\ni
The drawing $D$ obtained by superposing  $D_2$ and $D_3$ is an outer drawing
of $K_{3,n}$.  We color the edges of $D$ as in our figures,
for example the edges incident to $p_2$ are the green edges. 
In $D$ each color class of edges is a non-crossing star. For the blue and the green this
is true because the edges come from only one of $D_2$ and $D_3$. For the
red star it is true due to construction. Moreover, all the red-blue and
red-green crossings are as prescribed by the original table $T$.
The problem we face is that there is little control on blue-green
crossings.
Let $\ered(v),\egreen(v),\eblue(v)$ be the red, green, and blue edge of $v$.

\medskip \ni {\bf Claim:} For all $v$, $w$ the parity of the number of
crossings between $\egreen(v)$ and $\eblue(w)$ in~$D$ is prescribed
by $T$, i.e., if $T$ requires a crossing between two edges, then they
have an odd number of crossings and an even number of crossings otherwise.
\medskip

Consider the curves $\egreen(v)\cup \ered(v)$ and 
$\eblue(w)\cup \ered(w)$. The rotations prescribe whether the number of
intersections of the two curves is odd or even. Hence, the parity is
respected by $D$. The crossings
of the pairs $(\egreen(v),\ered(w))$, $(\ered(v),\eblue(w))$
in~$D$ are as prescribed by $T(v,w)$. Hence, the
parity of the number of crossings of $\egreen(v)$ and $\eblue(w)$
in $D$ is also prescribed.
\qedclaim

Because the rotation at $v$ in $D$ is correct we also note: If
$\egreen(v)$ and $\eblue(v)$ cross, then the number of crossings is
even.  Hence, if $D$ has no pair of a green and a blue edge crossing
more than once, then $D$ is an outer drawing realizing the
types given by $T$.

Now consider a pair of edges $\egreen(v)$ and $\eblue(w)$ crossing
more than once, $v=w$ is allowed.  Use
a homeomorphism of the plane to make $\egreen(v)$ a horizontal
straight line segment, see Figure~\ref{fig:meander}. (In the
literature intersection patterns of two simple curves in the plane are
often called \emph{meanders}. They are of interest in enumerative and
algebraic combinatorics.)

The intersections with $\egreen(v)$ subdivide $\eblue(w)$ into a family of
\emph{arcs} and two extremal pieces. 
A blue arc defines an interval on the green edge, this is the
\emph{interval of the arc}. An arc together with its interval 
enclose a bounded region, this is the \emph{region of the arc}.

\Fact Apart from intersecting intervals any two regions of arcs over a
fixed green edge $\egreen(v)$ are either disjoint or nested. Because
blue edges are pairwise disjoint this also holds if the arcs are
defined by different blue edges.  
\smallskip

%%%%%%%%%%%%%%%%%%%%%%%%%%%%%%%%%%%%%%%%%%%%%%%%%%%
%%
% in einem figure environment mit caption
   \calc_figscale{26}
    \begin{figure}[htb]
    \centerline{\input{\path/meander.pstex_t}}
    \caption{A meander with 7 empty lenses.\label{fig:meander}}
    \end{figure}
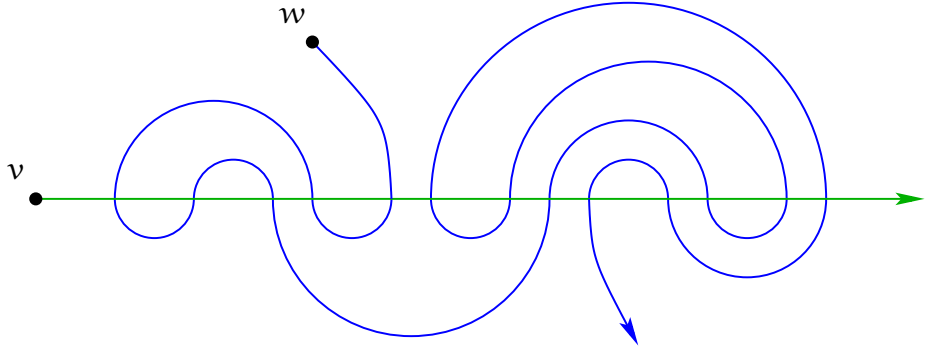
    
%%
%%%%%%%%%%%%%%%%%%%%%%%%%%%%%%%%%%%%%%%%%%%%%%%%%%%

An inclusionwise minimal region of an arc is a \emph{lens}.
Since $D$ has a finite number of crossings and hence a finite number
of regions we have:

\Fact Every region of an arc contains a lens.
\smallskip

Consider a lens $L$ formed by pair of a green and a blue edge in $D$.
Suppose that~$L$ is an {\em empty lens}, i.e., there is no vertex
$u\in V$ inside of $L$. It follows that the boundary of $L$ is only
intersected by red edges, moreover, if a red edge $e$ intersects the
boundary of $L$ on the green side, then $e$ also intersects $L$ on the
blue side. Therefore, we can make~$\egreen(v)$ and~$\eblue(w)$ switch
sides at $L$ and with small deformations at the two crossings get rid
of them. It is important to note that the \emph{switch at a lens} does
not change the types. In particular after the switch the drawing still
represents assignment~$T$.

Apply switching operations until the drawing $D'$ obtained by
switching has the property that every lens in $D'$ contains a vertex.  

In the following we show that~$D'$ has no lens. For the proof we use
that the table $T_1$ corresponding to
$(\pi_3,\pi_2)$ also has to be consistent. Since~$D'$ has no lens we
conclude that it is an outer drawing that realizes the types
given by $T$.  The existence of such a drawing was the statement of
the theorem.
\end{proof}
%%%%%%%%%%%%%%%%%%%%%%%%%%%%%%%%%%%%%%%%%%%%%%%%%%%
\def\sigv{\sigma_{green}(v)}

Let $D'$ be a drawing with the property that every lens contains a
vertex. Before going into the proof that there is no lens in $D'$ we
fix some additional notation. For a given green edge $\egreen(v)$ the
regions defined by blue arcs over $\egreen(v)$ can be classified as
\emph{above}, \emph{below}, and \emph{wrapping}.
A wrapping region is a region with one contact between  $\egreen(v)$  
and $\eblue(w)$ from above and one contact from below.

\begin{lemma}\label{lem:wrapping}
There is no wrapping region.
\end{lemma}

\begin{proof}
  Let $R$ be the wrapping region formed by a blue edge wrapping around
  $\egreen(v)$ and note that $v\in R$. If the wrapping blue edge is
  $\eblue(v)$, then $\ered(v)$ has to intersect one of $\egreen(v)$ and
  $\eblue(v)$ to leave the region. This is not allowed. If the wrapping blue edge is
  $\eblue(u)$ with $u\neq v$, then $\eblue(v)$ has to intersect one of $\egreen(v)$ and
  $\eblue(u)$ to leave the region. Again, this is not allowed. 
\end{proof}

With a similar proof we get:
\begin{lemma}\label{lem:selftrap}
Vertex $w$ is not contained in the region defined by an arc
of $\eblue(w)$ over $\egreen(v)$.
\end{lemma}

\begin{proof}
The green edge of $w$ would be trapped in such a region.
\end{proof} 

For a region defined by an arc on $\eblue(w)$ we speak of a
\emph{forward} or \emph{backward} region depending on the direction of
the arc above $\egreen(v)$. Formally: label the~$t$~crossings of
$\egreen(v)$ and $\eblue(w)$ as $1,..,t$ according to the order on
$\eblue(w)$.  Arcs correspond to consecutive crossings. An arc
$[i,i+1]$ is \emph{forward} if crossing $i$ is to the left of crossing
$i+1$ on $\egreen(v)$, otherwise it is \emph{backward}.  The
permutation of $1,..,t$ obtained by reading the crossings from $v$ to
$p_2$ along $\egreen(v)$ is called the \emph{meander permutation} and
denoted $\sigv$.

A region is a \emph{relative lens} for $\egreen(v)$ and $\eblue(w)$
if it is above or below $\egreen(v)$ and minimal in the nesting
order of regions defined by $\egreen(v)$ and $\eblue(w)$.
In the sequel we sometimes abuse notation by talking of lenses when we
mean relative lenses.

\begin{proposition}\label{prop:v-lens}
For all $v\in V$, the green and blue edges of $v$ do not cross in $D'$.
\end{proposition}

\begin{proof}
If $\egreen(v)$ and $\eblue(v)$ cross, then there is at least one blue arcs
on $\eblue(w)$ over $\egreen(v)$ and, hence, there are regions. 

From the order of the three outer vertices and the fact that $\ered(v)$
has no crossing with the two other edges of $v$ we conclude that
at the last crossing of $\eblue(v)$ and $\egreen(v)$ the blue edge is
crossing $\egreen(v)$ downwards.

Since at the first crossing the blue edge is crossing upwards there is
an arc and consequently also a lens above $\egreen(v)$.

Suppose there is a forward lens above $\egreen(v)$. Let $u$ be a
vertex in the lens. Vertex~$u$ is below the line of $v$ in the
green-red arrangement.  Hence, either $T_3(u,v)=B$ or $T_3(u,v)=N$ and
$v<_1 u$. Vertex~$u$ is below the forward arc of $\eblue(v)$ forming
the lens. Hence, it is below the line of $v$ in the red-blue
arrangement, and either $T_2(u,v)=B$ or $T_2(u,v)=N$ and $u<_1 v$.
From Table~\ref{tab:legal-projections} we infer that the only legal
assignment is $T(u,v)=W_1$ and the projections are $T_2(u,v)=N$ and
$T_3(u,v)=N$. However, there is no consistent choice for the order of
$u$ and $v$ in $<_1$. This shows that there is no forward lens above
$\egreen(v)$.

Now let $[i,i+1]$ be a backward lens above $\egreen(v)$. Let $u$ be a
vertex in the lens.  We distinguish whether the last crossing on the
blue edge is to the right/left of the lens on $\egreen(v)$, see
Figure~\ref{fig:Fall2+3}.

%%%%%%%%%%%%%%%%%%%%%%%%%%%%%%%%%%%%%%%%%%%%%%%%%%%
%%
% in einem figure environment mit caption
   \calc_figscale{22}
    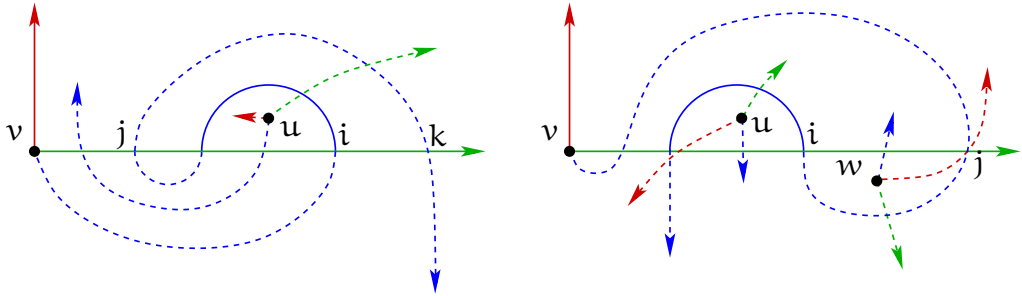
\begin{figure}[htb]
    \centerline{\input{\path/Fall2+3.pstex_t}}
    \caption{The two cases for a backward lens above
$\egreen(v)$. Dashed curves indicate the order of some crossings along
the corresponding edges.\label{fig:Fall2+3}}
    \end{figure}
    
%%
%%%%%%%%%%%%%%%%%%%%%%%%%%%%%%%%%%%%%%%%%%%%%%%%%%%

Suppose the last crossing $k$ on the blue edge is to the right of the
lens, Figure~\ref{fig:Fall2+3} (left).  To get from the arc $[i,i+1]$ to
$k$ the edge $\eblue(v)$ has to cross $\egreen(v)$ upwards at some
crossing $j$ left of $i+1$.  The edge
$\eblue(u)$ has to stay disjoint from $\eblue(v)$, therefore, after
crossing $\egreen(v)$ to leave the lens it has a second crossing left
of $j$.  The edge $\egreen(u)$ has a crossing with $\eblue(v)$ in the
arc $[i,i+1]$ and another one between $j$ and~$k$.  Now consider the
edge $\ered(u)$. If it leaves the lens through $\egreen(v)$, then it
has entered a region whose boundary consists of a piece of $\eblue(u)$
and a piece of $\egreen(v)$. Edge $\ered(u)$ is disjoint from
$\eblue(u)$ and it has already used its unique crossing with
$\egreen(v)$. Hence, $\ered(u)$ has to leave the lens through the blue
arc on $\eblue(v)$. In this case, however, $\ered(u)$ enters a region
whose boundary consists of a piece of $\eblue(v)$ and a piece of
$\egreen(u)$.  Again $\ered(u)$ is trapped. Hence, this configuration
is impossible.

Suppose the last crossing on the blue edge is to the left of the lens,
Figure~\ref{fig:Fall2+3} (right). To get from $v$ to the backward arc
$[i,i+1]$ edge $\eblue(v)$ has to cross $\egreen(v)$ downwards at some
crossing $j$ right of $i$. Let $w$ be a vertex in the region of some
blue arc $[k,k+1]$ below $\egreen(v)$ with the property that $j\leq k < i$.
The edges $\ered(u)$ and $\ered(w)$ both need at least two
crossings with the union of $\egreen(v)$ and $\eblue(v)$. Hence, they
both cross $\egreen(v)$ and $\eblue(v)$. The order of the red edges at
$p_1$ implies $u <_1 v <_1 w$. Matching these data with the drawings
of Figure~\ref{fig:sixtypes} we find that $T(u,v) = B_3$ and $T(v,w) =
B_2$. These types together with the already known permutation $\pi_1$
imply that $\pi_2$ and $\pi_3$ also equal $\pi_1$. Hence, the triple is
a uniform system, and we must have $T(u,w) \in \{B_2,B_3\}$.
Since edge $\eblue(w)$ has to follow the
`tunnel' prescribed by $\eblue(v)$ there is a crossing of edges
$\egreen(v)$ and $\eblue(w)$. There is also a crossing of $\egreen(v)$
and $\ered(w)$.  However, neither in $B_2$ nor in $B_3$ we see
crossings of $\egreen(1)$ with $\eblue(2)$ and $\ered(2)$.  Hence,
again the configuration is impossible.
\end{proof}

We now come to the discussion of the general case.

\begin{proposition}\label{prop:v,w-lens}
  For all $v,w\in V$, there is no lens formed by $\egreen(v)$ and
  $\eblue(w)$ in~$D'$.
\end{proposition}

Our proof of this proposition unfortunately depends on a lengthy
case analysis, an outline of which is given in the following subsection.

%%%%%%%%%%%%%%%%%%%%%%%%%%%%%%%%%%%%%%%%%%%%%%%%%%%%%%%%%%%%%%%%%%%%%%%%%%%%%%
\subsection{Outline of the Proof of Proposition~\ref{prop:v,w-lens}}

Suppose that there is a lens $L$ formed by $\egreen(v)$ and
$\eblue(w)$ in~$D'$. From Proposition~\ref{prop:v-lens} we know that
$v\neq w$. In fact the proposition implies that the star of every 
vertex is non-intersecting. For emphasis we collect this and the other 
restrictions on crossings in $D'$ in a list.

\Item{(1)}  Edges of the same color do not cross.
\Item{(2)}  Edges of different color that belong to the same
        vertex do not cross.
\Item{(3)}  Red edges have at most one intersection with any other edge.
\smallskip

These properties will be crucial throughout the argument. We also know
that there are no wrapping regions (Lemma~\ref{lem:wrapping}) and the
vertex in a region or lens is always different from the vertices of
the edges defining the region (Lemma~\ref{lem:selftrap}).
\medskip

In section~\ref{ssec:8configs}, we discuss eight configurations that may appear
in a meander of $\eblue(w)$ over $\egreen(v)$. The eight
configurations shown in Figure~\ref{fig:eight-configs} correspond to
the simplest meanders for the following three binary decisions:

\Bitem Vertex $w$ is \emph{above}/\emph{below} edge $\egreen(v)$.
\Bitem The last crossing is to the \emph{left}/\emph{right} of the
first crossing.
\Bitem At the last crossing $\eblue(w)$ is cutting $\egreen(v)$ 
\emph{upward}/\emph{downward}.
\medskip

\ni
For example the meander labeled VII corresponds to below/left/up. 
In the case discussion we impose additional conditions on
vertices $u$ that are contained in the regions defined by arcs of
$\eblue(w)$ over $\egreen(v)$. These conditions either ask for $u\in
S_1(w)$ or for $u\not\in S_1(w)$. The conditions are so that they
are satisfied for free if the respective regions are 
relative lenses of $\eblue(w)$
over $\egreen(v)$. In each of the cases we show that the 
$T_1$ projections yield an illegal table. Hence the configurations 
do not appear in the drawing $D'$. 

%%%%%%%%%%%%%%%%%%%%%%%%%%%%%%%%%%%%%%%%%%%%%%%%%%%
%%
% in einem figure environment mit caption
   \calc_figscale{14}
    \begin{figure}[htb]
    \centerline{\input{\path/eight-configs.pstex_t}}
    \caption{The eight configurations for the 
first phase of the case analysis.\label{fig:eight-configs}}
    \end{figure}
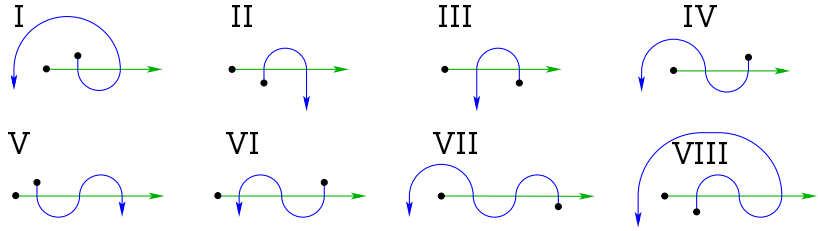
    
%%
%%%%%%%%%%%%%%%%%%%%%%%%%%%%%%%%%%%%%%%%%%%%%%%%%%%

In the second part we use the results from the first part to show that
meanders of the drawing $D'$ have no \emph{inflections}, i.e., in the meander
permutation we never see $i-1$ and $i+1$ on the same side of $i$. 
This is split in the following two lemmas. The first, easy one, is
the following.

\begin{lemma}
If a meander permutation $\sigv$ has an inflection, then it has one of
the four \emph{turn-back patterns} shown in Figure~\ref{fig:back-configs}.
\end{lemma}

\begin{proof}
Suppose that a meander permutation $\sigma = \sigv$ has an inflection
at $i$. We discuss the case where $i+1 <_\sigma i-1 <_\sigma i$ and
the blue arc $[i+1,i]$ of $\eblue(w)$ is above $\egreen(v)$. All the
other cases can be treated with symmetrical arguments.  First note
that $i-1 > 1$, otherwise, $w$ is in the region defined by the
arc $[i+1,i]$. This is impossible (Lemma~\ref{lem:selftrap}).

Now suppose that $i-1 <_\sigma i-2 <_\sigma i$. Let $x$ be a vertex in
the region defined by the arc $[i-1,i-2]$. Confined by $\eblue(w)$ the
edge $\eblue(x)$ has to cross $\egreen(v)$ at least three times and it
contains an arc whose region contains $x$. This is impossible
(Lemma~\ref{lem:selftrap}). Hence, $i+1 <_\sigma i-2 <_\sigma i-1
<_\sigma i$ and the meander contains the second of the turn-back
patterns shown in Figure~\ref{fig:back-configs}. Other cases of
inflections are related to the other cases of turn-back patterns.
\end{proof}

The proof of the second lemma involves a delicate case
analysis and is deferred to Subsection~\ref{ssec:no-turnback}.

\begin{lemma}\label{lem:turn-back}
A meander permutation has no turn-back.
\end{lemma}

%%%%%%%%%%%%%%%%%%%%%%%%%%%%%%%%%%%%%%%%%%%%%%%%%%%
%%
% in einem figure environment mit caption
   \calc_figscale{14}
    \begin{figure}[htb]
    \centerline{\input{\path/back-configs.pstex_t}}
    \caption{The four turn-back pattern.\label{fig:back-configs}}
    \end{figure}
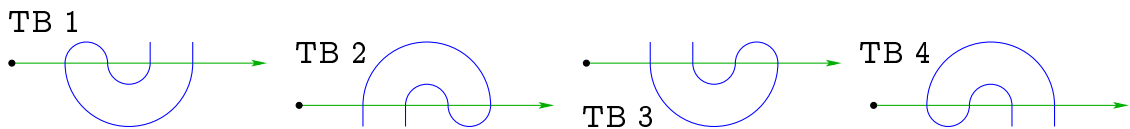
    
%%
%%%%%%%%%%%%%%%%%%%%%%%%%%%%%%%%%%%%%%%%%%%%%%%%%%%

The consequence of Lemma~\ref{lem:turn-back} is that the meanders in
$D'$ are very simple, their meander permutations are either the
identity permutation or the reverse of the identity. In particular
every arc defines a relative lens.  It then follows that each meander
in $D'$ can be classified according to the three binary decisions
mentioned above. Hence, it falls into one of the configurations that
have been discussed in Subsection~\ref{ssec:8configs}. The additional
conditions are satisfied because the arcs all define relative
lenses. Hence, there are no nontrivial meanders, i.e., every pair of a
green and a blue edge crosses at most once in~$D'$. This concludes the
proof of Proposition~\ref{prop:v,w-lens}.

\subsection{Case analysis for the proof of Proposition~\ref{prop:v,w-lens}}
%%%%%%%%%%%%%%%%%%%%%%%%%%%%%%%%%%%%%%%%%%%%%%%%%%%%%%%%%%%%%%%%%%%
\subsubsection{The eight basic configurations}\label{ssec:8configs}

We now deal with special instances of the
configurations from Figure~\ref{fig:eight-configs}.

%%==============================================================
\begin{wrapfigure}[10]{r}{0.30\textwidth}
\vskip-4mm
\centering
\includegraphics[width=0.25\textwidth]{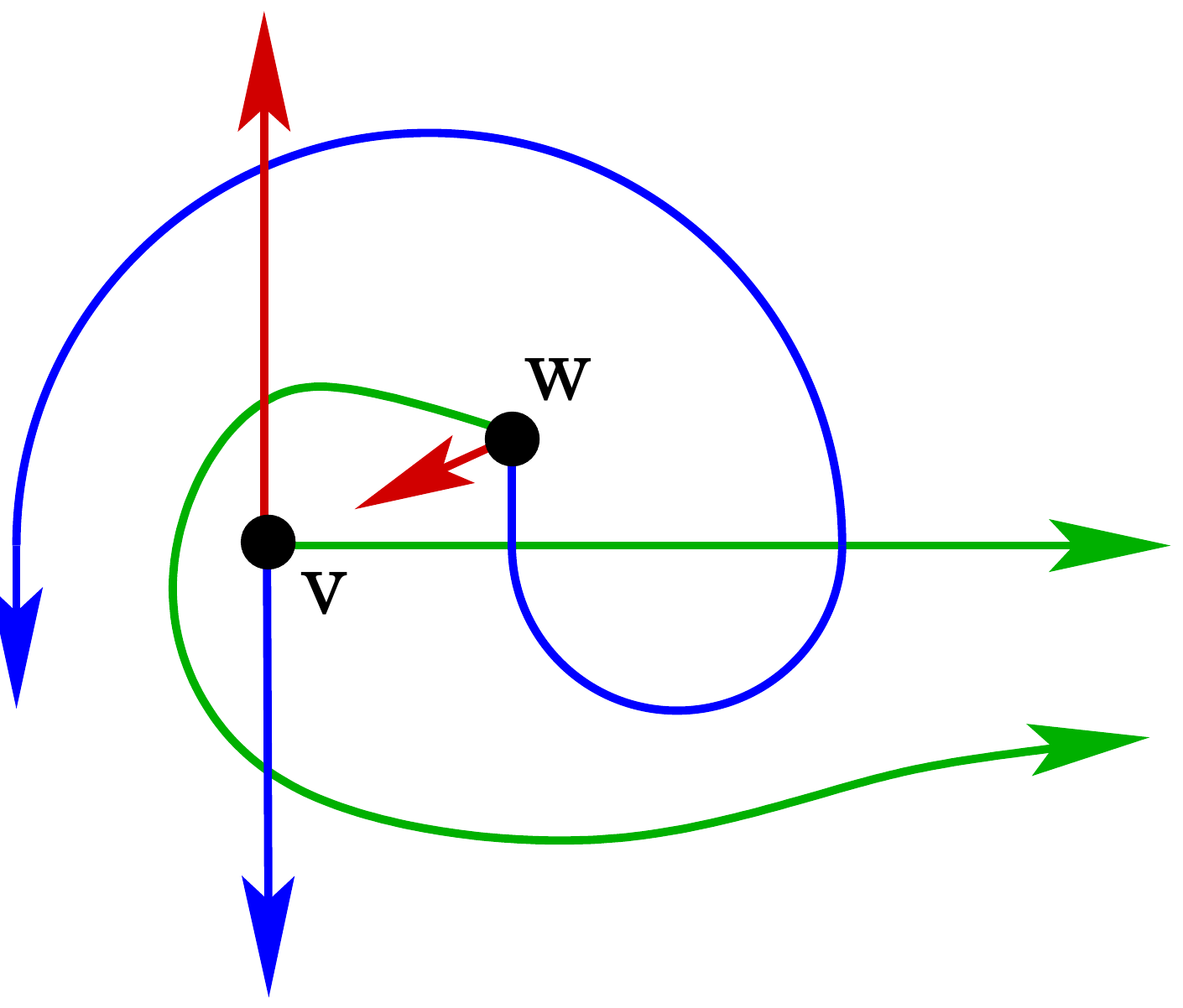}
\caption{\label{fig:case1}
\hbox to 0.15\textwidth{\hfil} Illustration for Case I.}
\end{wrapfigure}
\Case{I} \emph{The first intersection along $\eblue(w)$ is downwards
  and the last intersection is upwards and to the right of the first
  intersection along $\egreen(v)$.}
\medskip

Its last intersection with $\egreen(v)$ being upward forces
$\eblue(w)$ to intersect $\ered(v)$ on its final piece. 
Therefore, edge $\egreen(w)$ has to turn around $v$ as shown in
Figure~\ref{fig:case1}. Now consider $\ered(w)$, its first crossing
(among those relevant to the argument) has to be with $\egreen(v)$.
With this crossing, however, $\ered(w)$ gets separated from its
destination $p_1$ by the union of $\eblue(w)$ and $\egreen(v)$.
Therefore, this case is impossible.

%%==============================================================
%
\begin{wrapfigure}[14]{r}{0.30\textwidth}
%\vskip-4mm
\centering
\includegraphics[width=0.25\textwidth]{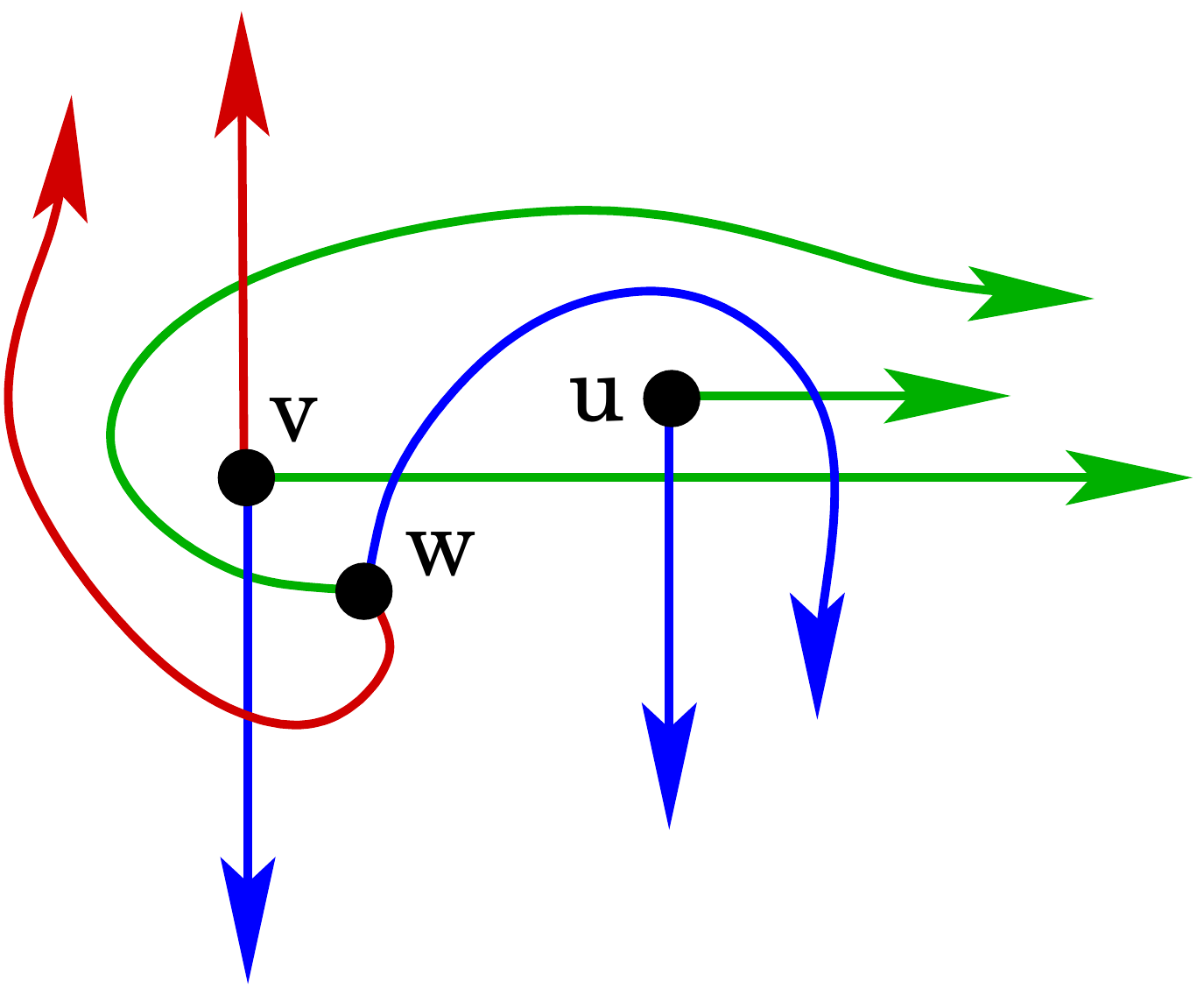}
\caption{\label{fig:case2}
\hbox to 0.15\textwidth{\hfil} Illustration for Case II.\\[3mm]
\hbox{\qquad$\protect\triple{w}{u}{v}{A}{B}{A}$}.}
\end{wrapfigure}
\Case{II} \emph{The first intersection along $\eblue(w)$ is upwards
  and the last intersection is downwards and to the right of the first
  along $\egreen(v)$ and there is a forward arc of $\eblue(w)$ above
  $\egreen(v)$ whose region contains a vertex $u\not\in S_1(w)$.}
\medskip

Edge $\egreen(w)$ has to turn around $v$ and also $\ered(w)$ has to
cross $\eblue(v)$, see Figure~\ref{fig:case3}. Now consider a
vertex~$u$ in a forward region of $\eblue(w)$ above
$\egreen(v)$. Vertex $u$ is not under an arc of $\eblue(u)$
(Lemma~\ref{lem:selftrap}) and there is at most one crossing of
$\eblue(u)$ and $\ered(w)$. Therefore, $\eblue(u)$ behaves as shown in
the sketch, i.e., $w <_3 u <_3 v$.  Also  $w <_2 u <_2 v$. From the
intersections of green and blue edges we conclude that the projections
$T_1$ are as given by the table on the right.  This table is not
legal. Hence, this case is impossible.

%%==============================================================
\begin{wrapfigure}[14]{r}{0.30\textwidth}
\vskip-0mm
\centering
\includegraphics[width=0.20\textwidth]{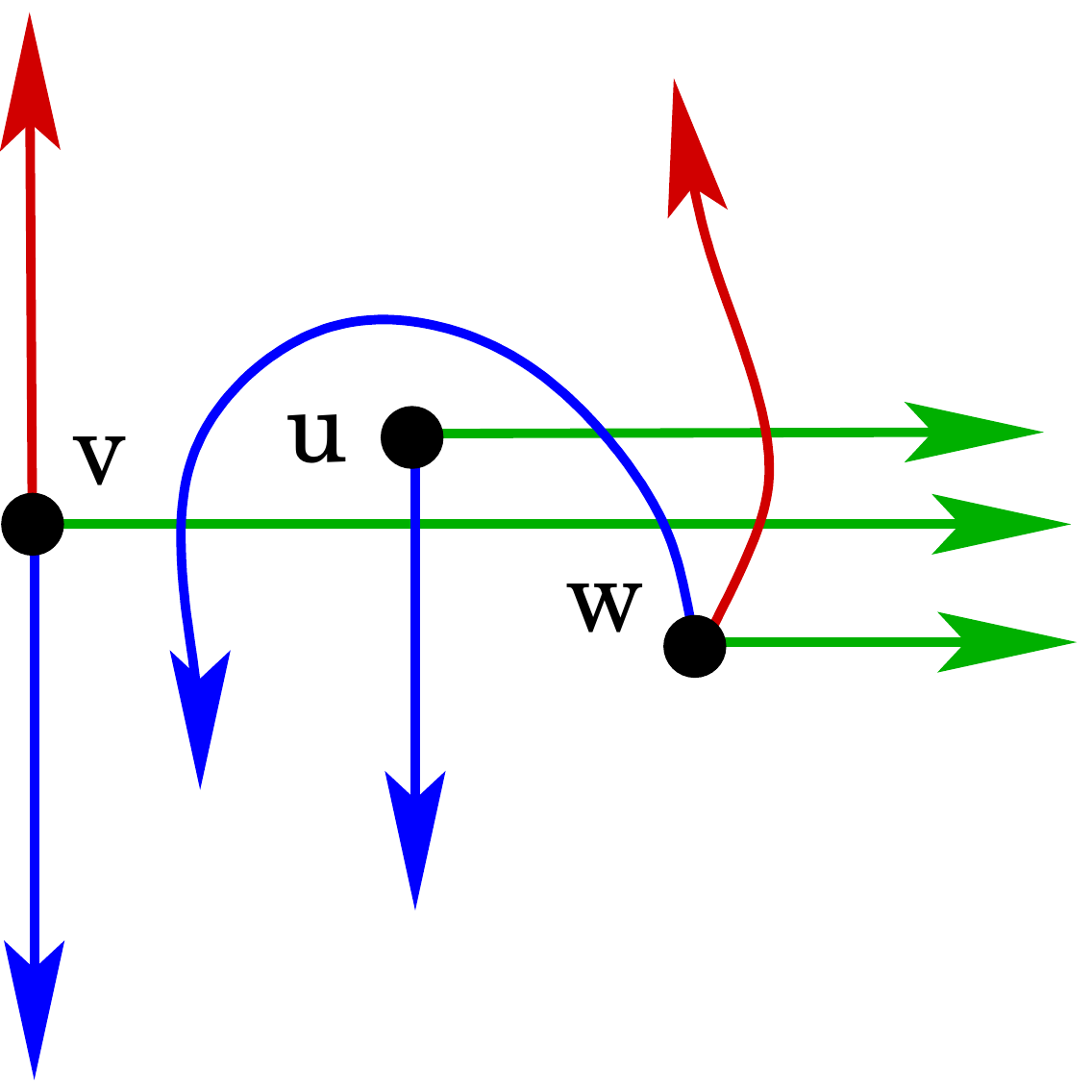}
\caption{\label{fig:case3}
\hbox to 0.15\textwidth{\hfil} Illustration for Case III.\\[3mm]
\hbox{\qquad$\protect\triple{u}{w}{v}{B}{A}{N}$}.}
\end{wrapfigure}
\Case{III} \emph{The first intersection along $\eblue(w)$ is upwards
  and the last intersection is downwards and to the left of the first
  along $\egreen(v)$ and there is a backward arc of $\eblue(w)$ above
  $\egreen(v)$ whose region contains a vertex $u\in S_1(w)$.}
\medskip

Edge $\egreen(w)$ has to stay below $\egreen(v)$.  Now consider a
vertex~$u$ in a backward region of $\eblue(w)$ above $\egreen(v)$.
Vertex $u$ is not under an arc of $\eblue(u)$ and 
we exclude the configuration of Case~I. Therefore,
$\eblue(u)$ behaves as shown in Figure~\ref{fig:case2},
i.e., $u <_3 w <_3 v$. It follows that $\egreen(u)$
has to stay above  $\egreen(v)$ and $u <_2 v <_2 w$.
From the intersections of green and blue edges we
conclude that the projections $T_1$ are as given by the table
on the right.  This tables is not legal. Hence, this case
is impossible.

%%==============================================================
%%==============================================================
\newpage

\begin{wrapfigure}[13]{r}{0.30\textwidth}
\vskip-0mm
\centering
\includegraphics[width=0.28\textwidth]{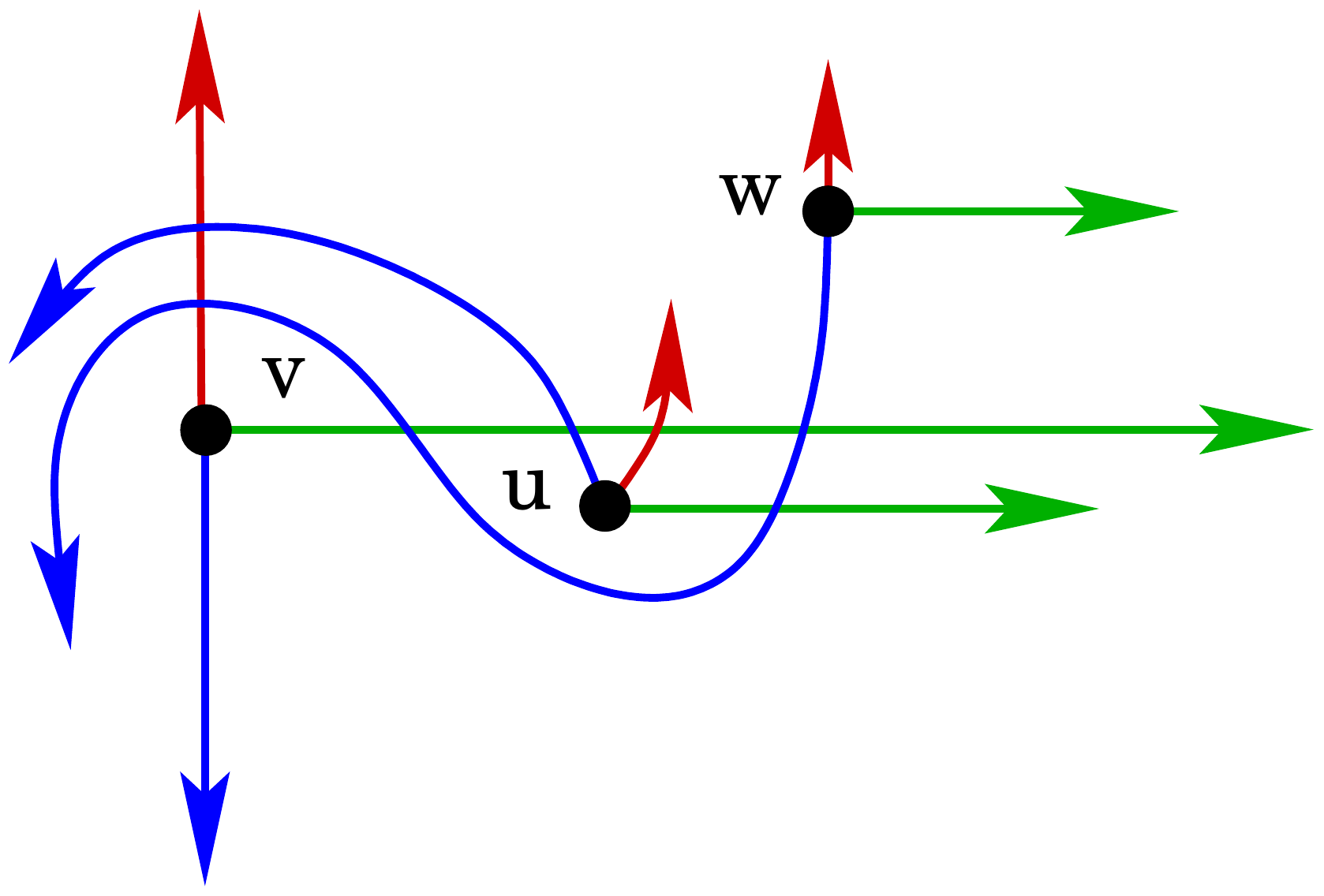}
\caption{\label{fig:case4}
\hbox to 0.15\textwidth{\hfil} Illustration for Case IV.\\[3mm]
\hbox{\qquad$\protect\triple{v}{w}{u}{N}{B}{A}$}.}
\end{wrapfigure}
\Case{IV} \emph{The first intersection along $\eblue(w)$ is downwards
  and the last intersection is upwards and to the left of the first
  along $\egreen(v)$ and there is a backward arc of $\eblue(w)$ below
  $\egreen(v)$ whose region contains a vertex $u\not\in S_1(w)$.}
\medskip

Edge  $\eblue(w)$ has
to cross $\ered(v)$ and the other edges of $w$ have no intersections
with edges of $v$. Edge $\eblue(u)$ either turn around $w$ or it
crosses $\ered(v)$. In the first case $\egreen(u)$ has turn around $v$.
Now, $\ered(u)$ has to cross $\eblue(w)$ and gets separated from its
destination $p_1$ by the union of$\eblue(w)$ and $\egreen(u)$.
Hence, there is no possible routing for edge $\ered(u)$ respecting the
restrictions on crossings in~$D'$. In the second case $\eblue(u)$
crosses $\ered(v)$ and $v <_3 w <_3 u$.  From the green edges we
obtain $w <_2 v <_2 u$. The intersections of green and blue
edges, see Figure~\ref{fig:case4} imply that the projections
$T_1$ are as given by the table below the figure.  This table is not
legal. Hence, this case is impossible.

%%==============================================================
%%
\begin{wrapfigure}[13]{r}{0.30\textwidth}
\vskip-0mm
\centering
\includegraphics[width=0.28\textwidth]{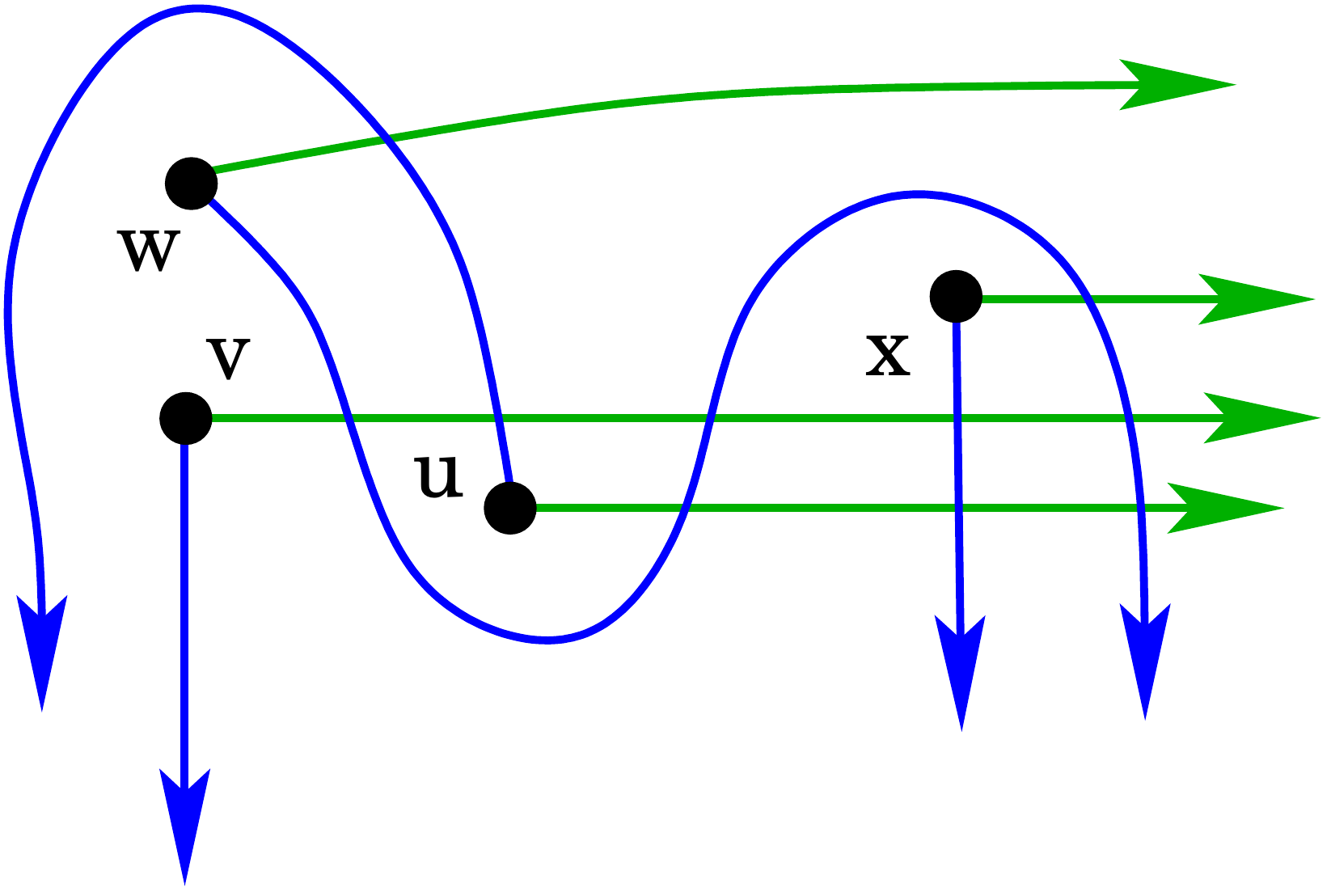}
\caption{\label{fig:case5}
\hbox to 0.15\textwidth{\hfil} Illustration for Case V.\\[3mm]
\hbox{\qquad$\protect\triple{w}{x}{u}{A}{B}{A}$}.}
\end{wrapfigure}
\Case{V} \emph{The first intersection along $\eblue(w)$ is downwards
  and the last intersection is downward and to the right of the first.
  Moreover, along $\egreen(v)$ there is a forward arc of
  $\eblue(w)$ below $\egreen(v)$ whose region contains a vertex
  $u\in S_1(w)$ and further to the right a forward arc of
  $\eblue(w)$ above $\egreen(v)$ containing a vertex $x\not\in S_1(w)$.}
\medskip

If $u <_3 w$, then edge $\eblue(w)$ makes sure that $x\not\in S_1(u)$.
Therefore, we have Case~II with $v$, $u$, and $x$.

If $w <_3 u <_3 v$, then  $w\in S_1(u)$.
Therefore, we have a Case~III.

Now let $w <_3 v <_3 u$. Since $\eblue(x)$ is crossing $\egreen(v)$
downwards we have $\egreen(x)$ above $\egreen(v)$ and a crossing
of $\eblue(w)$ with $\egreen(x)$. Therefore $w <_3 x$.

\begin{wrapfigure}[7]{r}{0.30\textwidth}
\vskip-7mm
\centering
\includegraphics[width=0.28\textwidth]{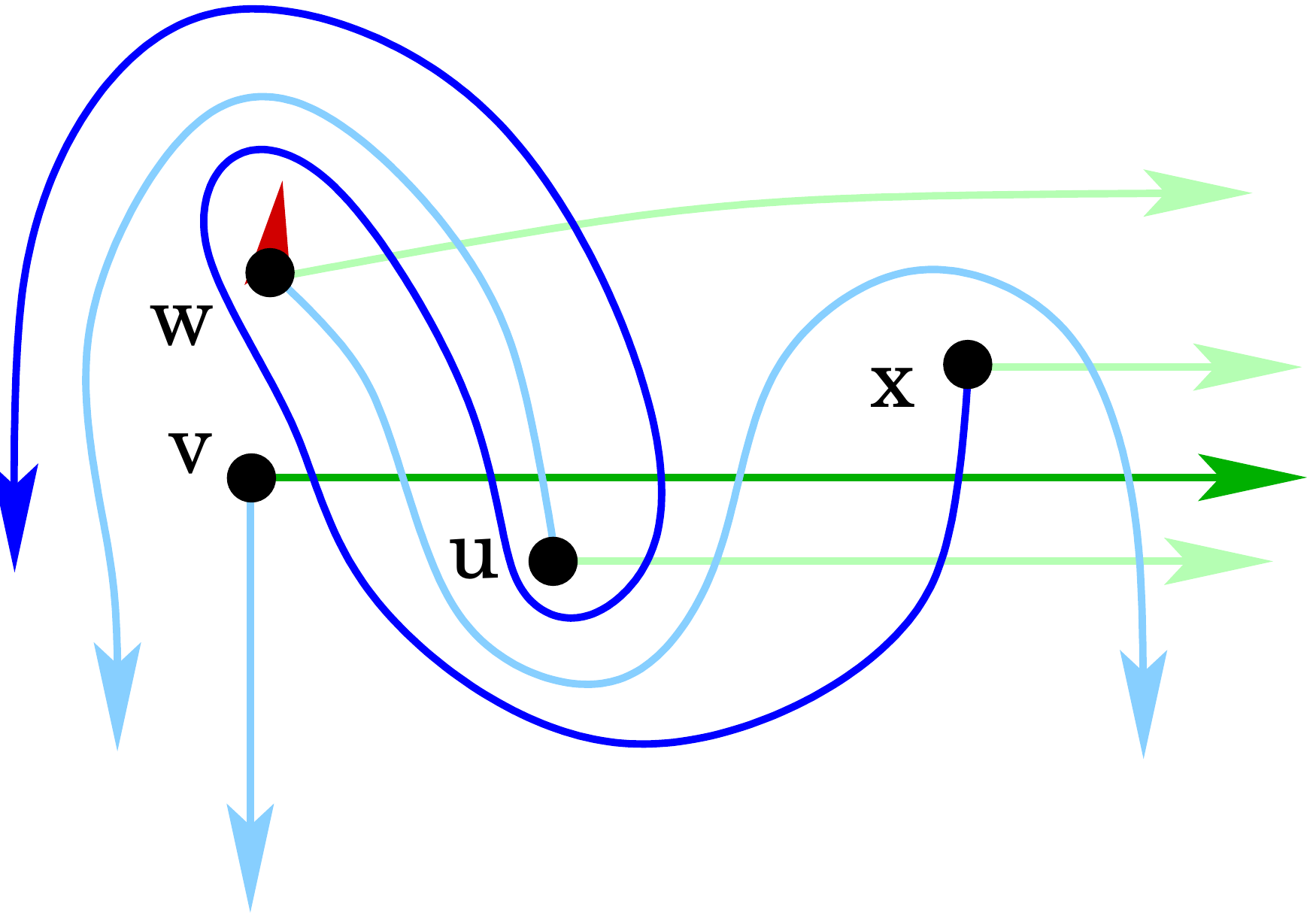}
\caption{\label{fig:case5extra}
\hbox to 0.15\textwidth{\hfil} Case V, 2nd illustration.}
\end{wrapfigure}
If $w <_3 x <_3 v$, then the intersections of green and blue
edges areas shown in Figure~\ref{fig:case5}.  The projections $T_1$ are
given by the table below the figure.  This table is not
legal. Hence, this case is impossible.
If $v <_3 x <_3 u$, then we have Case~IV with $x$ and
$u$. Finally, if  $u <_3 x$, then the green and blue
edges behave as shown in Figure~\ref{fig:case5extra}. 
Now, $\ered(w)$ would have to cross one of $\eblue(x)$ and $\egreen(v)$
twice to get to its destination $p_1$.

\newpage
%%==============================================================
%%
\def\QTableCaseVI{
{\small\begin{ytableau}
\none[u]  & A        & B \\
\none[w]  & A        & A \\
\none     & \none[v] & B \\
\none     &\none     & \none[x]
\end{ytableau}}
}\begin{wrapfigure}[12]{r}{0.30\textwidth}
\vskip-3mm
\centering
\includegraphics[width=0.28\textwidth]{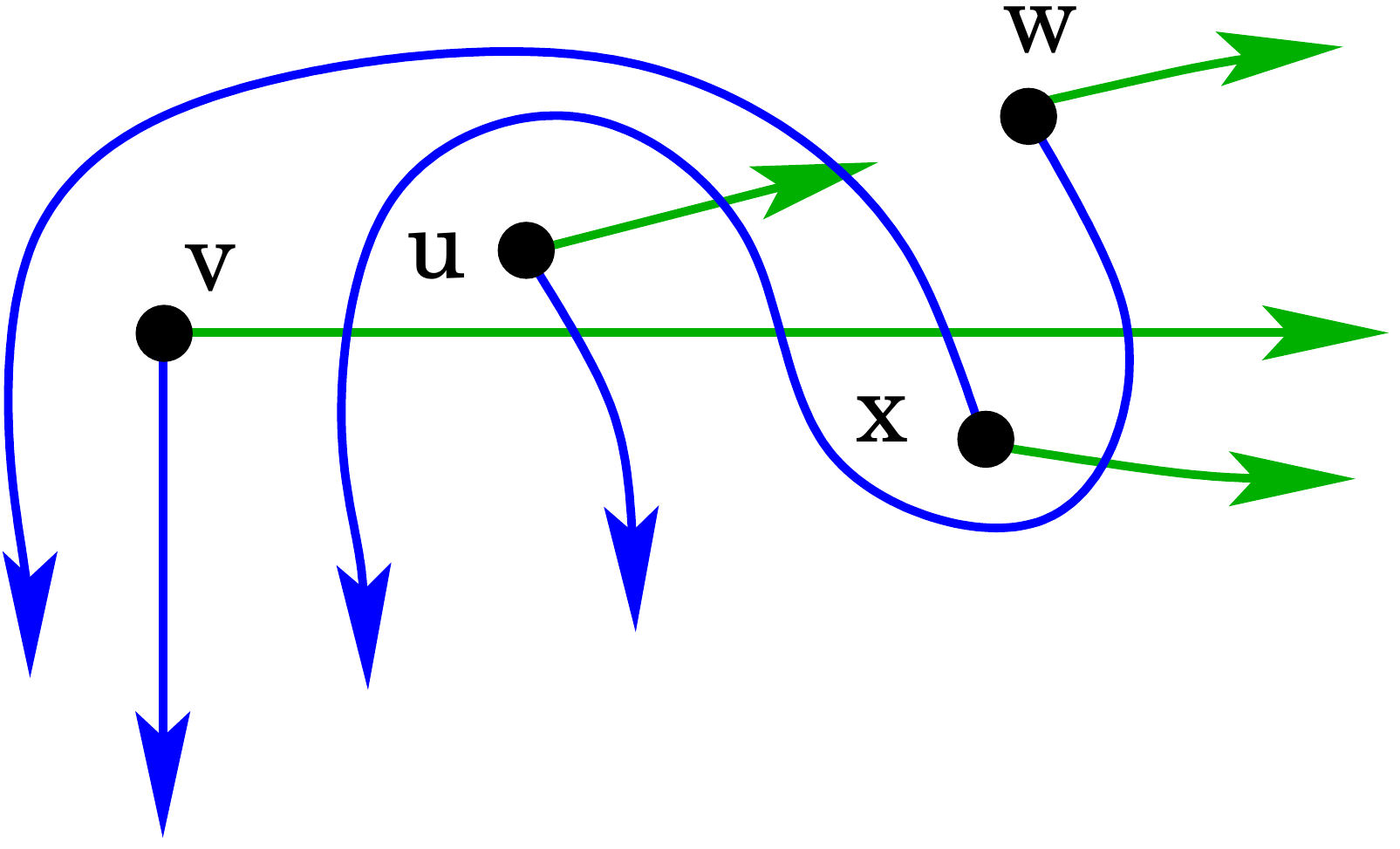}
\caption{\label{fig:case6}
\hbox to 0.15\textwidth{\hfil} Illustration for Case VI.\\[3mm]
\hbox{\qquad\protect\QTableCaseVI}.}
\end{wrapfigure}
\Case{VI} \emph{The first intersection along $\eblue(w)$ is downwards
  and the last intersection is downward and to the left of the first.
  Moreover, along $\egreen(v)$ there is a backward arc of
  $\eblue(w)$ above $\egreen(v)$ whose region contains a vertex
  $u\in S_1(w)$ and further to the right a backward arc of
  $\eblue(w)$ below $\egreen(v)$ containing a vertex $x\not\in S_1(w)$.}
\medskip

If $x <_3 v$, then edge $\eblue(x)$ either yields a Case~II with $w\not\in
S_1(x)$ or a Case~III with $u\in S_1(x)$. Therefore, $v <_3 x$.

If $w <_3 u$, then edge $\eblue(u)$ has an arc above $u$ or it forms a
Case~I. Therefore, $u <_3 w$ and the complete order is $u <_3 w <_3 v <_3 x$.

From green edges we get $w <_2 v <_2 x$ and $u <_2 v$, see
Figure~\ref{fig:case6}.  Independent of the entry $T_1(u,w)$ the
resulting table as shown below Figure~\ref{fig:case6}
violates the quadruple rule, Lemma~\ref{lem:qrule-2}.
Hence, this case is impossible.

%%==============================================================
%%
\def\QTableCaseVII{
{\small\begin{ytableau}
\none[x]  &  A        & B        & A \\
\none     & \none[v]  & B        & B \\
\none     & \none     & \none[w] & \none[u]
\end{ytableau}}
}\begin{wrapfigure}[12]{r}{0.30\textwidth}
\vskip-3mm
\centering
\includegraphics[width=0.28\textwidth]{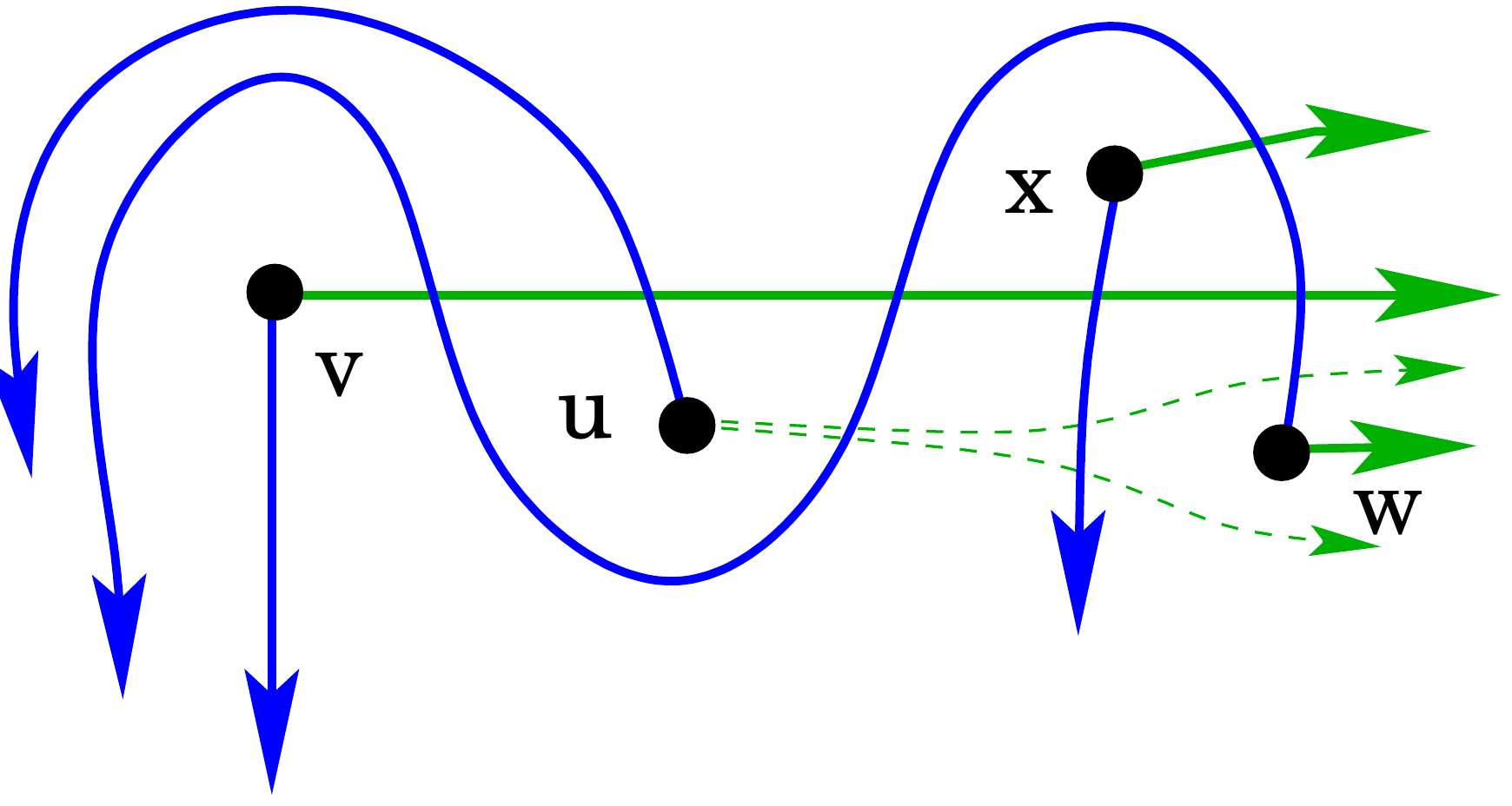}
\caption{\label{fig:case7}
\hbox to 0.15\textwidth{\hfil} Illustration for Case VII.\\[3mm]
\hbox{\qquad\protect\QTableCaseVII}.}
\end{wrapfigure}
\Case{VII} \emph{The first intersection along $\eblue(w)$ is upwards
  and the last intersection is upward and to the left of the first.
  Along $\egreen(v)$ there is a backward arc of
  $\eblue(w)$ below $\egreen(v)$ whose region contains a vertex
  $u\not\in S_1(w)$ and further to the right a backward arc of
  $\eblue(w)$ above $\egreen(v)$ containing a vertex $x\in S_1(w)$.}
\medskip

If $v <_3 x$, then edge $\eblue(x)$ either yields a Case~I or a
Case~IV with $u\not\in S_1(x)$. Therefore, $x <_3 v$.

If $u <_3 x$, then vertices $v$, $u$ and $x$ are in the configuration of
Case~II. If $x <_3 u <_3 w$, then $u <_3 v$ since otherwise $u$ would
be in a region of $\eblue(u)$, a violation of
Lemma~\ref{lem:selftrap}. Now $u <_3 v <_3 w$ and $T_1(u,v)=B$,
$T_1(u,w)=A$ and $T_1(v,w)=B$ is an illegal table.

If  $v <_3 u$, then $x <_3 v <_3 w <_3 u$ and the situation is as
sketched in Figure~\ref{fig:case7}. Irrespective of $u <_2 w$ or $w <_2 u$ 
the $T_1$ projections of the remaining 5 pairs 
yield the partial table shown below the figure. 
This violates the quadruple rule, Lemma~\ref{lem:qrule-2}.
Hence, this case is impossible. 

%%==============================================================
%%
\def\QTableCaseVIII{
{\small\begin{ytableau}
\none[u]  &  A        & 
    \hbox{\footnotesize N\kern-1pt\big|\kern-1ptB}         
                                 & A \\
\none     & \none[v]  & B        & B \\
\none     & \none     & \none[x] & \none[w]
\end{ytableau}}
}\begin{wrapfigure}[12]{r}{0.30\textwidth}
\vskip-3mm
\centering
\includegraphics[width=0.28\textwidth]{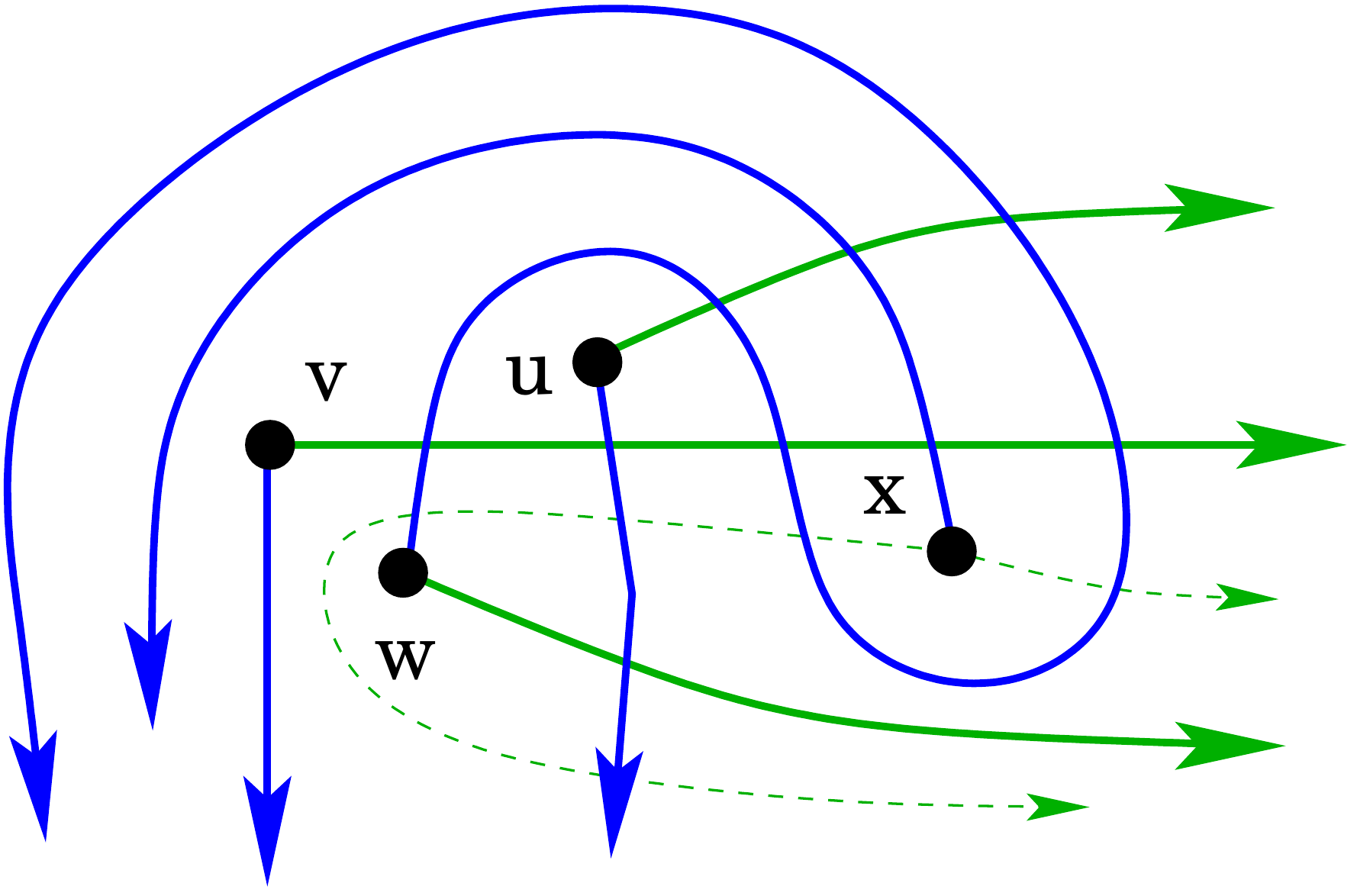}
\caption{\label{fig:case8}
\hbox to 0.15\textwidth{\hfil} Illustration for Case VIII.\\[3mm]
\hbox{\qquad\protect\QTableCaseVIII}.}
\end{wrapfigure}
\Case{VIII} \emph{The first intersection along $\eblue(w)$ is upwards
  and the last intersection is upward and to the right of the first.
  Along $\egreen(v)$ there is a forward arc of
  $\eblue(w)$ above $\egreen(v)$ whose region contains a vertex
  $u\not\in S_1(w)$ and further to the right a forward arc of
  $\eblue(w)$ below $\egreen(v)$ containing a vertex $x\in S_1(w)$.}
\medskip

If $v <_3 u$, then edge $\eblue(u)$ either yields a Case~I or a
Case~IV with $w\not\in S_1(x)$. Therefore, $u <_3 v$.  

If $u <_3 x <_3 v$,
then edge $\eblue(x)$ yields a Case~III with $u\in S_1(x)$.
If $x <_3 u$, then $x <_3 u <_3 w$ and $T_1(x,u)=N$,
$T_1(x,w)=B$ and $T_1(u,w)=A$ is an illegal table.

Therefore, $v <_3 x$ and in total $u <_3 v <_3 x <_3 w$, see
Figure~\ref{fig:case8}. Also $u <_2 v <_2 w$ and $v <_2 x$ only the
order of $x$ and $w$ in $<_2$ is not determined.  If $w <_2 x$ then
$T_1(x,w) = N$ otherwise $T_1(x,w) = B$. The resulting table is shown
below Figure~\ref{fig:case8}. In the first case the triple $u,v,x$
yields an illegal subtable. In the second case there is a violation of
the quadruple rule, Lemma~\ref{lem:qrule-2}.  Hence, this case is
impossible.

%%%%%%%%%%%%%%%%%%%%%%%%%%%%%%%%%%%%%%%%%%%%%%%%%%%%%%%%%%%%%%%%%%%
\subsubsection{Proof of Lemma~\ref{lem:turn-back}}\label{ssec:no-turnback}

Given $v$ suppose that some $\eblue(w)$ has a meander over
$\egreen(v)$ with a turn-back. Among all the turn-backs over
$\egreen(v)$ choose the one with the turn-back point $i$ as far as
possible, where $i$ is the crossing of the turn-back whose two adjacent
crossings are on the same side.
 
It will be seen in the proof that this specific choice of turn-back 
substantially simplifies the analysis of two of the cases.

%%==============================================================
%% Turn-back
\Case{TB 1} \emph{The turn back corresponds to 
a substring $i,i-1,i-2,i+1$ or to 
a substring $i,i+1,i+2,i-1$ in $\sigv$ and the 
turn-back arc is above $\egreen(v)$.}
\smallskip

Let us assume that the substring is $i,i-1,i-2,i+1$ so that
the turn-back arc is $[i,i-1]$. Let $u$ be a vertex in the region
of the turn-back arc and let $x$ be a vertex in the region of the arc
$[i-1,i-2]$. 

The edge $\eblue(u)$ is confined by $\eblue(w)$ and has a forward arc
below $x$. If $\eblue(u)$ has a crossing with $\ered(v)$ then we have
a Case~I. From Lemma~\ref{lem:selftrap} it follows that the last
crossing of $\eblue(u)$ and $\egreen(v)$ is downwards and to the right
of $i+1$.

If $x\not\in S_1(u)$, then $\eblue(u)$ has a backward arc $a$ in the
interval $[i-1,i-2]$ such that $x$ is in the region of arc $a$.  The
left part of Figure~\ref{fig:caseTB1} shows a schematic view of the
situation.  Let $y$ be a vertex in the region of an arc of $\eblue(u)$
above $\egreen(v)$ and preceeding $a$ along $\eblue(u)$ such that $y
\not\in S_1(x)$. Edges $\eblue(u)$ and Lemma~\ref{lem:selftrap} force
$\eblue(x)$ to have its last intersection with $\egreen(v)$ downwards
and to the right of $i+1$. Hence $\eblue(x)$ and $y$ are in the
configuration of Case~II.

We know that $x\in S_1(u)$. Let $y$ be a vertex below the last forward
arc of $\eblue(u)$ above $\egreen(v)$. If $y\not\in S_1(u)$, then
$\eblue(u)$ with $x$ and $y$ are in the configuration of
Case~V. Otherwise, if $y\in S_1(u)$, then there is a backward arc $a$ of 
$\eblue(u)$ above $y$. Choose a vertex $z$ from a backward arc below 
$\egreen(v)$ and to the right of $a$. The
right part of Figure~\ref{fig:caseTB1} shows a schematic view of the
situation. Edge $\eblue(u)$ is forcing $\eblue(z)$
to have an arc over $y$. 

%%%%%%%%%%%%%%%%%%%%%%%%%%%%%%%%%%%%%%%%%%%%%%%%%%%
%%
% in einem figure environment mit caption
   \calc_figscale{26}
    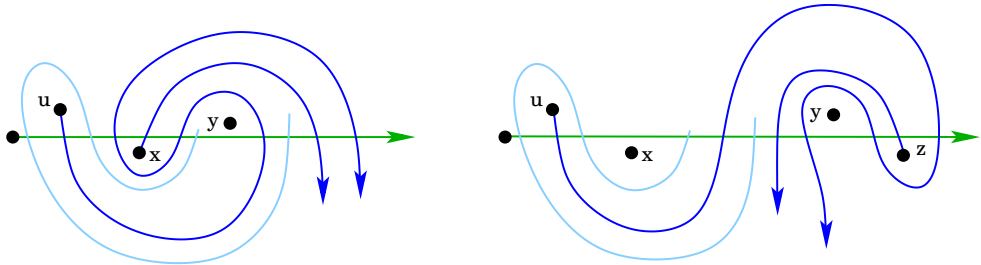
\begin{figure}[htb]
    \centerline{\input{\path/caseTB1.pstex_t}}
    \caption{Illustrations for case TB 1.\label{fig:caseTB1}}
    \end{figure}
    
%%
%%%%%%%%%%%%%%%%%%%%%%%%%%%%%%%%%%%%%%%%%%%%%%%%%%%

If the last intersection of $\eblue(z)$ with
$\egreen(v)$ is downwards and left of
the last intersection of $\eblue(u)$ and $\egreen(v)$, then 
$v$ with $z$ and~$y$ are in the configuration of
Case~III. 

If the last intersection of $\eblue(z)$ with
$\egreen(v)$ is downwards and right of the last intersection of
$\eblue(u)$ with $\egreen(v)$, then $\ered(z)$ has to cross
$\eblue(u)$ twice to reach its destination $p_1$ without intersecting 
$\eblue(z)$ and $\egreen(z)$. The
left part of Figure~\ref{fig:caseTB1second} shows a schematic view of the
situation, the light-blue arc is spanned by two arcs of $\eblue(z)$.

If the last intersection of $\eblue(z)$ with
$\egreen(v)$ is upwards and $\eblue(z)$ crosses $\ered(v)$,
then there are two options.
The first is that $x\not\in S_i(z)$ and there is a configuration of Case~VII with
$v$, $z$, $x$ and $y$. The second is $x\in S_i(z)$, this requires that
$\eblue(z)$ has a turn-back interior of the original turn-back of
$\eblue(w)$ before crossing $\ered(v)$, in this situation $\eblue(x)$
shows a Case~I or together with $u$ a Case~III.
The right part of Figure~\ref{fig:caseTB1second} illustrates this case.

%%%%%%%%%%%%%%%%%%%%%%%%%%%%%%%%%%%%%%%%%%%%%%%%%%%
%%
% in einem figure environment mit caption
   \calc_figscale{26}
    \begin{figure}[htb]
    \centerline{\input{\path/caseTB1second.pstex_t}}
    \caption{More illustrations for case TB 1.\label{fig:caseTB1second}}
    \end{figure}
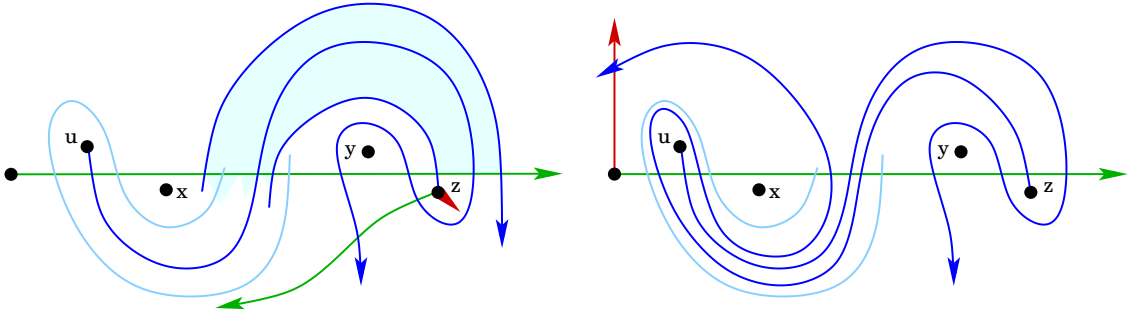
    
%%
%%%%%%%%%%%%%%%%%%%%%%%%%%%%%%%%%%%%%%%%%%%%%%%%%%%

%%==============================================================
%% Turn-back
\Case{TB 2} \emph{The turn back corresponds to 
a substring $i+1,i-2,i-1,i$ or to 
a substring $i-1,i+2,i+1,i$ and the 
turn-back arc is below $\egreen(v)$.}
\smallskip

Let us assume that the substring is $i+1,i-2,i-1,i$ so that
the turn-back is the arc $[i,i-1]$. Let $u$ be a vertex in the region
of the turn-back arc and let $x$ be a vertex in the region of the arc
$[i-2,i-1]$. 

The first intersection of $\eblue(u)$ with $\egreen(v)$ is left
of the turn back point $i$ of the meander $\eblue(w)$. The choice of
the turn-back as the leftmost implies that $\eblue(u)$ has no
turn-back. Hence, $\sigma_{green}(u)$ is the reverse of the identity.

If the last intersection of $\eblue(u)$ with $\egreen(v)$ is
downwards, then this is a Case~III with $x$. Otherwise the last
intersection is upwards and $\eblue(u)$ intersects $\ered(v)$.
Then there is a backward arc of $\eblue(u)$ below $\egreen(v)$ 
and a vertex $y$ from such an arc together with $x$ shows that 
the situation of Case~VII. See Figure~\ref{fig:caseTB2}.

%%%%%%%%%%%%%%%%%%%%%%%%%%%%%%%%%%%%%%%%%%%%%%%%%%%
%%
% in einem figure environment mit caption
   \calc_figscale{26}
    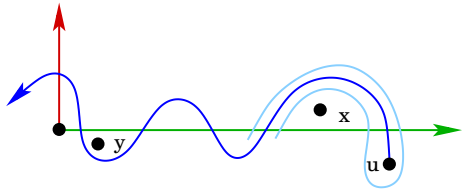
\begin{figure}[htb]
    \centerline{\input{\path/caseTB2.pstex_t}}
    \caption{Illustrations for case TB 2.\label{fig:caseTB2}}
    \end{figure}
    
%%
%%%%%%%%%%%%%%%%%%%%%%%%%%%%%%%%%%%%%%%%%%%%%%%%%%%

%%==============================================================
%% Turn-back
\Case{TB 3} \emph{The turn back corresponds to 
a substring $i+1,i-2,i-1,i$ or to 
a substring $i-1,i+2,i+1,i$ and the 
turn-back arc is above $\egreen(v)$.}
\smallskip

Let us assume that the substring is $i+1,i-2,i-1,i$ so that
the turn-back is the arc $[i,i-1]$. Let $u$ be a vertex in the region
of the turn-back arc and let $x$ be a vertex in the region of the arc
$[i-2,i-1]$. 

The first intersection of $\eblue(u)$ with $\egreen(v)$ is left
of the turn back point $i$ of the meander $\eblue(w)$. The choice of
the turn-back as the leftmost implies that $\eblue(u)$ has no
turn-back. Hence, $\sigma_{green}(u)$ is the reverse of the identity.

If $\eblue(u)$ intersects $\ered(v)$, then with $x$ this yields a
Case~IV. Otherwise there is a backward arc of $\eblue(u)$
above $\egreen(v)$. Let $y$ be a vertex from such an arc.
Now $\eblue(u)$ together with $x$ and $y$ are a Case~VI. 

%%==============================================================
%% Turn-back
\Case{TB 4} \emph{The turn back corresponds to 
a substring $i,i-1,i-2,i+1$ or to 
a substring $i,i+1,i+2,i-1$ and the 
turn-back arc is below $\egreen(v)$.}
\smallskip

Let us assume that the substring is $i,i-1,i-2,i+1$ so that
the turn-back is the arc $[i,i-1]$. Let $u$ be a vertex in the region
of the turn-back arc and let $x$ be a vertex in the region of the arc
$[i-2,i-1]$. 

If $\eblue(x)$ intersects $\ered(v)$ this is a Case~I. 
It follows that that $x <_2 v$. 
If $u <_3 x$ then $u$ and $x$ are in the configuration of Case~II. 

If $x <_3 u$, then there is a forward arc of $\eblue(u)$ below
$\egreen(v)$ and to the right of $i-2$. Let $y$ be a vertex from the
region of such an arc. 

We first consider the case $y\in S_1(u)$. 
If $\eblue(u)$ intersects $\ered(v)$ then $u$ with $y$ yields a
Case~VIII. If $x <_3 y$, then it follows with Lemma~\ref{lem:selftrap}
that $x\in S_1(y)$ so that they form a Case~III.
For $y <_3 x$ edge $\eblue(y)$ has to turn back as shown in the
left part of Figure~\ref{fig:caseTB4}. In this configuration
$\ered(u)$ has to cross one of $\egreen(v)$ and $\eblue(y)$ twice to
get to its destination $p_1$.

Now we are in the case $x <_3 u$ and $y\in S_1(u)$. 
This requires a backward arc of $\eblue(u)$ whose region contains $y$.
To the right of this arc there has to be a backward arc above
$\egreen(v)$. Let $z$ be a vertex from the region of this arc.
The situation is shown in the right part of Figure~\ref{fig:caseTB4}.  
If $y <_3 x$, then $z\not\in S_1(y)$ so that $y$ and $z$ form a
Case~II. If $x <_3 y$, then either $y$ is a Case~I or it forms a 
Case~III with $x$.

%%%%%%%%%%%%%%%%%%%%%%%%%%%%%%%%%%%%%%%%%%%%%%%%%%%
%%
% in einem figure environment mit caption
   \calc_figscale{26}
    \begin{figure}[htb]
    \centerline{\input{\path/caseTB4.pstex_t}}
    \caption{Illustrations for case TB 4.\label{fig:caseTB4}}
    \end{figure}
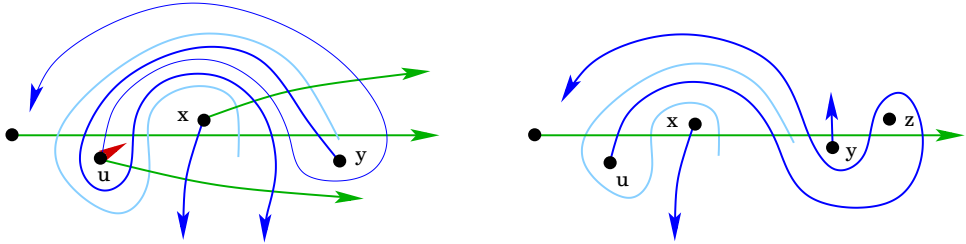
    
%%
%%%%%%%%%%%%%%%%%%%%%%%%%%%%%%%%%%%%%%%%%%%%%%%%%%%

%%%%%%%%%%%%%%%%%%%%%%%%%%%%%%%%%%%%%%%%%%%%%%%%%%%
\section{Outer Drawings with arbitrary $k$}\label{sec:arbk}

We now generalize our result on consistency of outer drawings for $k=2,3$ to any $k\geq 2$.

In the previous section, we defined types for each pair of vertices in $V$ in an outer drawing of $K_{3,n}$.
Fortunately, we do not need to go beyond the enumeration of the outer drawings of $K_{3,2}$ to be able to
generalize our result to arbitrary values of $k$.

\begin{theorem}[Consistency of outer drawings for $k\geq 2$]
  Consider the complete bipartite graph $G$ with vertex bipartition
    $(P,V)$ such that $|P|=k$ and $|V|=n$. Given assignments $T_{ij}$ of
    types in $\{A,B,N\}$ for each pair of vertices $(p_i, p_j)$ in $P$ and each pair of vertices
    in $V$, there exists an outer drawing of $G$ realizing those types if and only if
    \begin{itemize}
      \item all assignments $T_{ij}$ for $1\leq i<j\leq k$ obey the
        triple and quadruple consistency rules in the sense of Theorem~\ref{thm:consist-2},
      \item for any triple of vertices from $P$ and and each pair of vertices
    in $V$ the assignments correspond to an outer drawing of $K_{3,2}$
    in the sense of Table~\ref{tab:legal-projections}.
    \end{itemize}
\end{theorem}
\begin{proof}
  We proceed by induction, suppose that the result holds for some $k$ and prove that the result holds for $k+1$.
  The base cases for $k=2,3$ were proved in the previous sections.

  From the induction hypothesis, there exists a drawing $D_{k-1}$ of $K_{k-1,n}$ realizing the types involving vertices
  $p_1,p_2,\ldots ,p_{k-1}$ in $P$. Similarly, there exists a drawing $D_k$ of $K_{2,n}$ realizing the
  types involving the two vertices $p_1$ and $p_k$.

  Similarly to what we did in the proof of Theorem~\ref{thm:consist-3}, we can superpose the two drawings in such a way
  that the drawings of the two stars $K_{1,n}$ corresponding to $p_1$ match.
  More precisely, let the drawing $D_{k-1}$ live on plane $Z_{k-1}$ and $D_k$ live on plane $Z_k$.
  Consider a fixed homeomorphism $\phi$ between the planes. There is a homeomorphism $\psi: Z_{k-1} \to Z_{k-1}$
  such that mapping $D_{k-1}$ via $\phi\circ\psi$ to $Z_k$ yields a superposition of the two drawings with the following properties:
{
\Bitem Corresponding vertices are mapped onto each other.  
\Bitem The stars of $p_1$ are mapped onto each other, i.e., the edges at $p_1$ of the two drawings are represented by the same curves.
\Bitem At each vertex $v\in V$ the rotation is correct, i.e., we see the edges to $p_1,p_2,\ldots ,p_k$ in clockwise order.
\Bitem The drawing has no touching edges, i.e., when two edges meet they properly cross in a single point. 
}\par

\smallskip
\ni
The drawing $D$ obtained by superposing $D_{k-1}$ and $D_k$ is a drawing of $K_{k,n}$, possibly with some multiple crossings.
We color the edges of $D$ according to the vertex $p_i\in P$ it is incident to, using colors $1,2,\ldots ,k$.
Each color class of edges is a non-crossing star, and the parity of the number of crossings between two edges of distinct colors is prescribed by the type assignments.
Furthermore, since $D_{k-1}$ and $D_k$ are proper outer drawings, all pairs of edges crossing more than once have respective colors $i$ and $k$ for some $i\in \{2,3,\ldots ,k-1\}$.

% now perform uncrossings iteratively to get D'
Consider all lenses formed by an edge of color $k$ and another edge of color $i\in \{2,3,\ldots ,k-1\}$.
Such a lens is said to be empty whenever it does not contain any vertex from $V$. It is inclusionwise minimal whenever it does not fully contain any other lens.
Consider an empty inclusionwise minimal lens formed by an edge $e_k(v)$ of color $k$ and another edge $e_i(w)$ of color $i$.

If an edge of color $i'\in \{1,2,3,\ldots ,k-1\}\setminus \{i\}$ intersects the lens, it can only cross $e_i(w)$ once.
Similarly, it cannot cross the boundary of the lens on $e_k(v)$ more than once, either because $i'=1$, or because otherwise the lens would not be minimal.
Another edge of color $k$ cannot intersect the lens, because it could only do so by intersecting $e_i(w)$ at least twice, contradicting the minimality of the lens.

Hence all edges intersecting the lens intersect its boundary exactly once on each edge.
Therefore we can safely get rid of the two crossings by making $e_k(v)$ and $e_i(w)$ switch sides.
The {\em switch at a lens} does not change the type assignment. We can perform this iteratively until no empty lens is remaining.
We call the final drawing $D'$. This drawing has the property that every lens contains a vertex.

% suppose D' has a lens - this concerns three edges 1,i,k
We now wish to show that $D'$ has no lens. For the sake of contradiction, suppose that there is a lens formed by $e_i(v)$ and $e_k(w)$ in $D'$.
We can apply our previous analysis for the case $k=3$ on the restriction $D''$ of $D'$ formed by the edges of color $1,i$, and $k$, letting red$=1$, green$=i$, and blue$=k$.
In particular Proposition~\ref{prop:v,w-lens} holds, and $D''$ has no lens. Therefore, neither does $D'$, which is a suitable outer drawing of $K_{k,n}$. 
\end{proof}

This yields the following corollary for AT-graphs.

\begin{corollary}[AT-graph realizability]
  There exists an $O(k^3n^2 + k^2n^4)$ algorithm for deciding the existence of an outer drawing of an AT-graph whose underlying graph is of the form $K_{k,n}$.
\end{corollary}
\begin{proof}
  We know that the three types of pairs $A,B,N$ (see Figure~\ref{fig:sixtypes}) exactly prescribe which pairs of edges cross, hence given the set of crossing pairs, we can reconstruct the type assignments $T_{ij}$. From these data, we can consider every 5-tuple $(p_a,p_b,p_c,u,v)$, with $p_a,p_b,p_c\in P$ and $u,v\in V$, and check that the corresponding drawing of the induced $K_{3,2}$ graph is one of the drawings given in Table~\ref{tab:legal-projections}. This can be done in time $O(k^3n^2)$. If this is correct, we can then check that every triple is legal and that the quadruple rule is satisfied for each $T_{ij}$. This can be done in time $O(k^2n^4)$, hence the result.
\end{proof}

In fact, this implies that checking consistency on all 6-tuples of vertices is sufficient. This matches the result of Kyn{\v{c}}l for complete graphs~\cite{K13}.

%%%%%%%%%%%%%%%%%%%%%%%%%%%%%%%%%%%%%%%%%%%%%%%%%%%%%%%%%%%%%%%%%%%%%
\section{Extendable and Straight-line Outer Drawings}\label{sec:extendable}

A simple topological drawing of a graph is called {\em extendable} if its edges can be extended into a pseudoline arrangement.
We consider the problem of the existence of an extendable outer drawing for a given rotation system. We also further restrict to straight-line drawings.
We exploit a connection between topological drawings and generalized configuration of points which was also used by Kyn{\v{c}}l~\cite{K13}.
The main result uses the notion of {\em suballowable sequence}, as defined by Asinowski~\cite{A08}.

\subsection{Extendable Outer Drawings}

We will use the notion of {\em allowable sequence}. These sequences were introduced by Goodman and Pollack in a series of papers~\cite{GP80,GP84a,GP91} as an abstraction
of the sequences of successive permutations obtained by projecting a set of $n$ points in $\mathbb{R}^2$ on a rotating line. A good summary of this and related notions can be
found in Goodman's article in the Handbook of Discrete and Computational Geometry~\cite{handbook}.
We restrict to simple allowable sequences, corresponding to point sets with no three points on a line and such that no two pair of points determine the same slope.

A sequence $\pi_1, \pi_2, \ldots, \pi_{m+1}, \ldots, \pi_{2m}$ of permutations of an $n$-element set, with $m={n\choose 2}$, is an {\em allowable sequence}
if for all $i\in [m]$, $\pi_i$ and $\pi_{m+i}$ are reverse of each other, each consecutive pair of permutations differ only by a single adjacent transposition, and every pair of element is
reversed exactly twice in the sequence.

We will also need the following additional definition:
  A  {\em Generalized configuration} is a pair $(C,A)$ such that $C$ is a collection of $n$ points in the projective plane $\mathbb{P}^2$ and $A$
  is an arrangement of pseudolines, one of which is the {\em pseudoline at infinity} $L_{\infty}$, satisfying the following conditions:
  \begin{itemize}
  \item every pair of points of $C$ lie on a pseudoline of $A$,
  \item each pseudoline in $A$ contains exactly two points of $C$, except $L_{\infty}$,
    \item $L_{\infty}$ does not contain any point of $C$.
  \end{itemize}

We say that an allowable sequence on an $n$-element set is {\em realized} by a generalized configuration on $n$ points when the
successive permutations of the sequence correspond to successive segments of the pseudoline $L_{\infty}$, and the
transposition between a pair of permutations corresponds to the intersection of $L_{\infty}$ with the pseudoline containing the two corresponding points.
In particular, we observe that if a permutation $\pi$ is realized at some point $p\in L_{\infty}$, then $p$ can be connected by additional pseudolines
to the $n$ other points in the radial order prescribed by $\pi$.

The following result is known. 
\begin{theorem}[Goodman and Pollack~\cite{GP84a}, Thm. 4.4]
  \label{thm:gcas}
Every allowable sequence can be realized by a generalized configuration.
\end{theorem}
This directly yields that allowable sequences exactly encode which are the achievable rotation systems for our drawings.
We first define {\em suballowable sequences}, see Asinowski~\cite{A08}:
A sequence  $\pi_1, \pi_2,\ldots ,\pi_k$ of permutations of an $n$-element set is said to be a {\em suballowable sequence} if there exists an allowable sequence admitting  $\pi_1, \pi_2,\ldots ,\pi_k$ as a subsequence.

\begin{theorem}
Given a sequence $\pi_1, \pi_2,\ldots ,\pi_k$ of permutations of $n$ elements, there exists an extendable outer drawing of $K_{k,n}$ with this rotation system if and only if $\pi_1, \pi_2,\ldots ,\pi_k$ is a suballowabe sequence.
\end{theorem}
\begin{proof}
  First, suppose that the rotation system is a suballowable sequence. Then by definition there exists an allowable sequence admitting it as a subsequence. Then from Theorem~\ref{thm:gcas} there exists a generalized configuration on $n$ points realizing this sequence. The $n$ points form the set $V$. For each permutation $\pi_i$, there exists a point on $L_{\infty}$ that can be connected by additional pseudolines to the $n$ points of $V$ in the order prescribed by $\pi_i$. These $k$ points form the set $P$. By keeping only the inner segments of the pseudolines connecting every point of $P$ with all points of $V$, we obtain an extendable outer drawing satisfying our conditions. 

  Now suppose there exists an extendable outer drawing. Consider the generalized configuration obtained as follows. First, we extend the edges between the points of $P$ and $V$ into pseudolines, which is possible from the extendability of the drawing. Then we connect the points of $P$ by the pseudoline at infinity. Finally, by iteratively applying Levi's enlargement Lemma, we can also connect every pair of points in $V$ by a new pseudoline. The order in which each point of $P$ is connected to all points of $V$ realizes one permutation in the allowable sequence defined by the generalized configuration. The collection of such permutations therefore forms a suballowable sequence.
\end{proof}

Note that since suballowable sequences can be recognized efficiently (see~\cite{A08}), this directly yields a polynomial-time algorithm for deciding whether there exists an extendable outer drawing realizing a given rotation system.

\subsection{Straight-line Outer Drawings}

We have a direct equivalent statement for straight-line outer drawings.
\begin{theorem}
  Given a sequence $\pi_1, \pi_2,\ldots ,\pi_k$ of permutations of $n$ elements, there exists an straight-line outer drawing of $K_{k,n}$ with this rotation system if and only if $\pi_1, \pi_2,\ldots ,\pi_k$ is a suballowable sequence that can be realized by an actual point arrangement.
\end{theorem}

This does not yield an efficient algorithm, since deciding whether a given allowable sequence can be realized by an arrangement of straight lines is a hard problem~\cite{s-splNP-91,m-ut-88,hoffmann2016thesis}.
The problem, however, is tractable in the special case where $k=3$. This is due to the fact that we can fix the orientation of the three sets of edges.
There are two cases to consider. If the suballowable sequence $\pi_1,\pi_2,\pi_3$ can be extended into an allowable sequence such that the three permutations $\pi_i$ appear in the first $m$ permutations of the whole sequence, then there exists a straight-line outer drawing with $p_1=(-\infty, \infty), p_2=(0,\infty)$, and $p_3=(\infty, \infty)$. This corresponds to the case where the three projection directions span less than a half-circle. Otherwise there exists a drawing with $p_1=(\infty,\infty)$, $p_2 = (0,-\infty)$ and $p_3 = (-\infty,0)$.

%, we can fix the orientations of the three sets of edges, and the problem is tractable.

We first give an example showing that there are rotation systems for $k=3$ that are suballowable sequences but do not admit a straight line outer drawing.
Up to renaming, this example is the same as Asinowski's example of nonrealizable suballowable sequence (\cite{A08}, Proposition 8).
We give the proof for completeness. Furthermore, our proof directly indicates how to solve the straight-line realizability problem in the case $k=3$.

\begin{theorem}\label{thm:nonstretchable}
The triple of rotations given by:
\begin{eqnarray*}
\pi_1 &= (q_1,q_6,q_5,q_2,q_3,q_4)\\
\pi_2 &= (q_3,q_2,q_1,q_4,q_5,q_6)\\
\pi_3 &= (q_5,q_4,q_3,q_6,q_1,q_2)
\end{eqnarray*}
is a suballowable sequence, but is not realizable with straight lines.
\end{theorem}

%%%%%%%%%%%%%%%%%%%%%%%%%%%%%%%%%%%%%%%%%%%%%%%%%%%
%%
% in einem figure environment mit caption
   \calc_figscale{45}
    \begin{figure}[htb]
    \centerline{\input{\path/non-stretch-small.pstex_t}}
    \caption{\label{fig:non-stretch-small}}
    \end{figure}
    VC
{An outer drawing of the system from Theorem~\ref{thm:nonstretchable}.}
%%
%%%%%%%%%%%%%%%%%%%%%%%%%%%%%%%%%%%%%%%%%%%%%%%%%%%

\begin{proof}
Figure~\ref{fig:non-stretch-small} shows an outer drawing of the
system. Suppose that the system is stretchable, then there is a
drawing with $p_1=(\infty,\infty)$, $p_2 = (0,-\infty)$ 
and $p_3 = (-\infty,0)$, i.e., with horizontal rays, vertical rays,
and rays of slope 1. The coordinates of a vertex $v$ are denoted
$(x(v),y(v))$.

From $\pi_1$ we deduce the inequality $y(q_6) - x(q_6) < y(q_1) - x(q_1)$.
Similarly $y(q_2) - x(q_2) < y(q_5) - x(q_5)$ and $y(q_4) - x(q_4) < y(q_3) - x(q_3)$.
We will refer to these inequalities as I1, I2, and I3 in this order.

From $\pi_2$ we get $x(q_6) < x(q_5)$, $x(q_4) < x(q_1)$, and $x(q_2) < x(q_3)$.
Finally, from $\pi_3$ we get $y(q_1) < y(q_2)$, $y(q_3) < y(q_6)$, and $y(q_5) < y(q_4)$.
These six inequalities are the {\em helper inequalities}.

Adding I1 and I2 and rearranging terms we obtain 
$$ y(q_6) + y(q_4) + x(q_1) + x(q_3) < y(q_1) + y(q_3) + x(q_6) + x(q_4)$$
Replacing $y(q_6)$, $y(q_4)$, and $x(q_3)$ on the left side and
$y(q_3)$, $x(q_6)$, and $x(q_4)$ on the right side using the six
helper inequalities we obtain $ y(q_5) + x(q_2)  < y(q_2) + x(q_5)$.
This contradicts I3.
\end{proof}

A straight line outer drawing of such a rotation system
can be regarded as a {\em tropical arrangement of lines}.
We refer to~\cite{ms-itrop-15} for an
introduction to tropical geometry.  Such outer drawings can consequently
be regarded as {\em tropical arrangements of pseudolines}.
Restated as a result in tropical geometry
Theorem~\ref{thm:nonstretchable} tells us that there are
non-stretchable tropical arrangements of pseudolines.

The classical construction of a simple non-stretchable example of
pseudolines due to Ringel is based on the 9 lines of a Pappus
configuration. In~\cite{gst-fstrop-05} tropical versions of Pappus are
discussed. They note that a configuration on 9 tropical lines which
can be obtained from the configuration of
Theorem~\ref{thm:nonstretchable} by replacing points $q_\ell$ by
triple incidences of lines $a_i,b_j,c_k$ with $\{i,j,k\} = \{1,2,3\}$
is a valid tropical version of Pappus' Theorem.

The stretchability problem for a tropical arrangements of pseudolines is polynomial-time solvable, as it boils down to the feasability of a linear program.
The proof of Theorem~\ref{thm:nonstretchable} indicates how to set up such a linear program.

%%%%%%%%%%%%%%%%%%%%%%%%%%%%%%%%%%%%%%%%%%%%%%%%%%%%%%%%%%%%%%%%%%%%%%%%%%%%%%
\section{Open problems}

A problem that remains for $k=2$ is to find a good description of all outer drawings 
for a given pair $(\pi_1,\pi_2)$ of rotations. We state the same problem differently:
\begin{problem}\label{prob:2all-cons-ass}
  For a given pair $(\pi_1,\pi_2)$ of permutations characterize all
  consistent assignments of types with the property that if $a<b$ in
  both permutations, then $\type(a,b)\in \{A,B\}$, and if $a<_1b$ and
  $b<_2 a$, then $\type(a,b)=N$.
\end{problem} 
The simpler case where $\pi_1 = \pi_2$ (hence all types are in $\{A,B\}$)
is covered by Corollary~\ref{cor:countunif}, which provides a bijection
with combinatorial structures counted by Schr\"oder numbers.\\

The following natural question also remains unsolved:
\begin{problem}
For a given triple $(\pi_1,\pi_2,\pi_3)$ of rotations, can we decide
in polynomial time whether there exists a consistent assignment of types?
\end{problem}
The problem could be generalized to arbitrary $k$: how hard is it to decide
whether there exists an outer drawing of $K_{k,n}$ for a given rotation system?
We answered this question partially, by showing that it is polynomial-time
decidable when the drawing is further required to be extendable.\\

We also note that outer drawings are not really enlightening with respect to crossing numbers, as minimally crossing outer drawings of $K_{k,n}$ are easy to describe. It can be checked that by drawing the $n$ vertices of $V$ on a line so that $\lfloor k/2\rfloor$ vertices of $P$ see them in some order, and the other $\lceil k/2\rceil$ vertices see them in the reversed order, we achieve the minimum number of crossings. This minimum has value exactly ${n\choose 2} \left( {k\choose 2} - \lfloor k/2\rfloor \lceil k/2\rceil \right)$.

Further insight into crossing numbers could be gained by considering the obvious remaining problem of generalizing our analysis to general simple topological drawings, dropping the outer property.

\section*{Acknowledgments}
This work was started at the {\it Workshop on Order Types, Rotation
  Systems, and Good Drawings} in Strobl am Wolfgangsee (Austria) in 2015. We thank the organizers
from TU Graz as well as the participants for fruitful discussions.
Special thanks go to Pedro Ramos who suggested to look at complete bipartite graphs.

S.~Felsner also acknowledges support from DFG grant Fe 340/11-1.

% *************************************************************

\bibliography{pseudodraw}
\bibliographystyle{my-siam}

%
%%%%%%%%%%%%%%%%%%%%%%%%%%%%%%%%%%%%%%%%%%%%%%%%%%%%%%%%%%%%%%%%%%%%%%%%%%%%%%%%%%%%
\end{document}

%% file: Figures/example-1.pstex_t
\begin{picture}(0,0)%
\includegraphics[scale=\figscale]{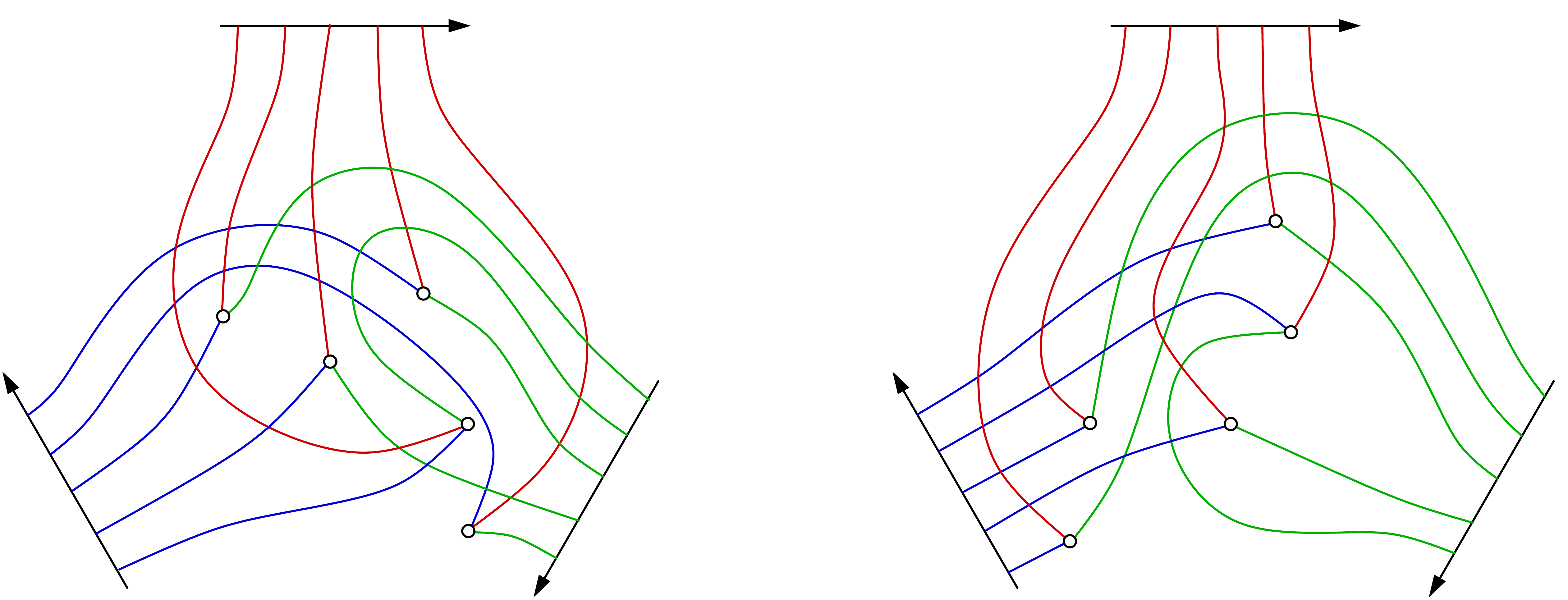}%
\end{picture}%
\setlength{\unitlength}{4144sp}%

\multiply\unitlength by \magproz
\divide\unitlength by 100

\begingroup\makeatletter\ifx\SetFigFont\undefined%
\gdef\SetFigFont#1#2#3#4#5{%
  \reset@font\fontsize{#1}{#2pt}%
  \fontfamily{#3}\fontseries{#4}\fontshape{#5}%
  \selectfont}%
\fi\endgroup%
\begin{picture}(23056,8819)(538,-10308)
\put(2086,-10021){\rotatebox{120.0}{\makebox(0,0)[b]{\smash{{\SetFigFont{12}{14.4}{\rmdefault}{\mddefault}{\updefault}{\color[rgb]{0,0,0}$1$}%
}}}}}
\put(1749,-9437){\rotatebox{120.0}{\makebox(0,0)[b]{\smash{{\SetFigFont{12}{14.4}{\rmdefault}{\mddefault}{\updefault}{\color[rgb]{0,0,0}$3$}%
}}}}}
\put(1411,-8852){\rotatebox{120.0}{\makebox(0,0)[b]{\smash{{\SetFigFont{12}{14.4}{\rmdefault}{\mddefault}{\updefault}{\color[rgb]{0,0,0}$2$}%
}}}}}
\put(1074,-8268){\rotatebox{120.0}{\makebox(0,0)[b]{\smash{{\SetFigFont{12}{14.4}{\rmdefault}{\mddefault}{\updefault}{\color[rgb]{0,0,0}$5$}%
}}}}}
\put(736,-7683){\rotatebox{120.0}{\makebox(0,0)[b]{\smash{{\SetFigFont{12}{14.4}{\rmdefault}{\mddefault}{\updefault}{\color[rgb]{0,0,0}$4$}%
}}}}}
\put(10291,-7426){\rotatebox{240.0}{\makebox(0,0)[b]{\smash{{\SetFigFont{12}{14.4}{\rmdefault}{\mddefault}{\updefault}{\color[rgb]{0,0,0}$2$}%
}}}}}
\put(9953,-8010){\rotatebox{240.0}{\makebox(0,0)[b]{\smash{{\SetFigFont{12}{14.4}{\rmdefault}{\mddefault}{\updefault}{\color[rgb]{0,0,0}$1$}%
}}}}}
\put(9616,-8595){\rotatebox{240.0}{\makebox(0,0)[b]{\smash{{\SetFigFont{12}{14.4}{\rmdefault}{\mddefault}{\updefault}{\color[rgb]{0,0,0}$4$}%
}}}}}
\put(9278,-9179){\rotatebox{240.0}{\makebox(0,0)[b]{\smash{{\SetFigFont{12}{14.4}{\rmdefault}{\mddefault}{\updefault}{\color[rgb]{0,0,0}$3$}%
}}}}}
\put(8941,-9764){\rotatebox{240.0}{\makebox(0,0)[b]{\smash{{\SetFigFont{12}{14.4}{\rmdefault}{\mddefault}{\updefault}{\color[rgb]{0,0,0}$5$}%
}}}}}
\put(15183,-10021){\rotatebox{120.0}{\makebox(0,0)[b]{\smash{{\SetFigFont{12}{14.4}{\rmdefault}{\mddefault}{\updefault}{\color[rgb]{0,0,0}$1$}%
}}}}}
\put(14846,-9437){\rotatebox{120.0}{\makebox(0,0)[b]{\smash{{\SetFigFont{12}{14.4}{\rmdefault}{\mddefault}{\updefault}{\color[rgb]{0,0,0}$3$}%
}}}}}
\put(14508,-8852){\rotatebox{120.0}{\makebox(0,0)[b]{\smash{{\SetFigFont{12}{14.4}{\rmdefault}{\mddefault}{\updefault}{\color[rgb]{0,0,0}$2$}%
}}}}}
\put(14171,-8268){\rotatebox{120.0}{\makebox(0,0)[b]{\smash{{\SetFigFont{12}{14.4}{\rmdefault}{\mddefault}{\updefault}{\color[rgb]{0,0,0}$5$}%
}}}}}
\put(13833,-7683){\rotatebox{120.0}{\makebox(0,0)[b]{\smash{{\SetFigFont{12}{14.4}{\rmdefault}{\mddefault}{\updefault}{\color[rgb]{0,0,0}$4$}%
}}}}}
\put(23465,-7426){\rotatebox{240.0}{\makebox(0,0)[b]{\smash{{\SetFigFont{12}{14.4}{\rmdefault}{\mddefault}{\updefault}{\color[rgb]{0,0,0}$2$}%
}}}}}
\put(23127,-8010){\rotatebox{240.0}{\makebox(0,0)[b]{\smash{{\SetFigFont{12}{14.4}{\rmdefault}{\mddefault}{\updefault}{\color[rgb]{0,0,0}$1$}%
}}}}}
\put(22790,-8595){\rotatebox{240.0}{\makebox(0,0)[b]{\smash{{\SetFigFont{12}{14.4}{\rmdefault}{\mddefault}{\updefault}{\color[rgb]{0,0,0}$4$}%
}}}}}
\put(22452,-9179){\rotatebox{240.0}{\makebox(0,0)[b]{\smash{{\SetFigFont{12}{14.4}{\rmdefault}{\mddefault}{\updefault}{\color[rgb]{0,0,0}$3$}%
}}}}}
\put(22115,-9764){\rotatebox{240.0}{\makebox(0,0)[b]{\smash{{\SetFigFont{12}{14.4}{\rmdefault}{\mddefault}{\updefault}{\color[rgb]{0,0,0}$5$}%
}}}}}
\put(4051,-1636){\makebox(0,0)[b]{\smash{{\SetFigFont{12}{14.4}{\rmdefault}{\mddefault}{\updefault}{\color[rgb]{0,0,0}$1$}%
}}}}
\put(4726,-1636){\makebox(0,0)[b]{\smash{{\SetFigFont{12}{14.4}{\rmdefault}{\mddefault}{\updefault}{\color[rgb]{0,0,0}$2$}%
}}}}
\put(5401,-1636){\makebox(0,0)[b]{\smash{{\SetFigFont{12}{14.4}{\rmdefault}{\mddefault}{\updefault}{\color[rgb]{0,0,0}$3$}%
}}}}
\put(6076,-1636){\makebox(0,0)[b]{\smash{{\SetFigFont{12}{14.4}{\rmdefault}{\mddefault}{\updefault}{\color[rgb]{0,0,0}$4$}%
}}}}
\put(6751,-1636){\makebox(0,0)[b]{\smash{{\SetFigFont{12}{14.4}{\rmdefault}{\mddefault}{\updefault}{\color[rgb]{0,0,0}$5$}%
}}}}
\put(17148,-1636){\makebox(0,0)[b]{\smash{{\SetFigFont{12}{14.4}{\rmdefault}{\mddefault}{\updefault}{\color[rgb]{0,0,0}$1$}%
}}}}
\put(17823,-1636){\makebox(0,0)[b]{\smash{{\SetFigFont{12}{14.4}{\rmdefault}{\mddefault}{\updefault}{\color[rgb]{0,0,0}$2$}%
}}}}
\put(18498,-1636){\makebox(0,0)[b]{\smash{{\SetFigFont{12}{14.4}{\rmdefault}{\mddefault}{\updefault}{\color[rgb]{0,0,0}$3$}%
}}}}
\put(19173,-1636){\makebox(0,0)[b]{\smash{{\SetFigFont{12}{14.4}{\rmdefault}{\mddefault}{\updefault}{\color[rgb]{0,0,0}$4$}%
}}}}
\put(19848,-1636){\makebox(0,0)[b]{\smash{{\SetFigFont{12}{14.4}{\rmdefault}{\mddefault}{\updefault}{\color[rgb]{0,0,0}$5$}%
}}}}
\end{picture}%

%% file: Figures/quadrants.pstex_t
\begin{picture}(0,0)%
\includegraphics[scale=\figscale]{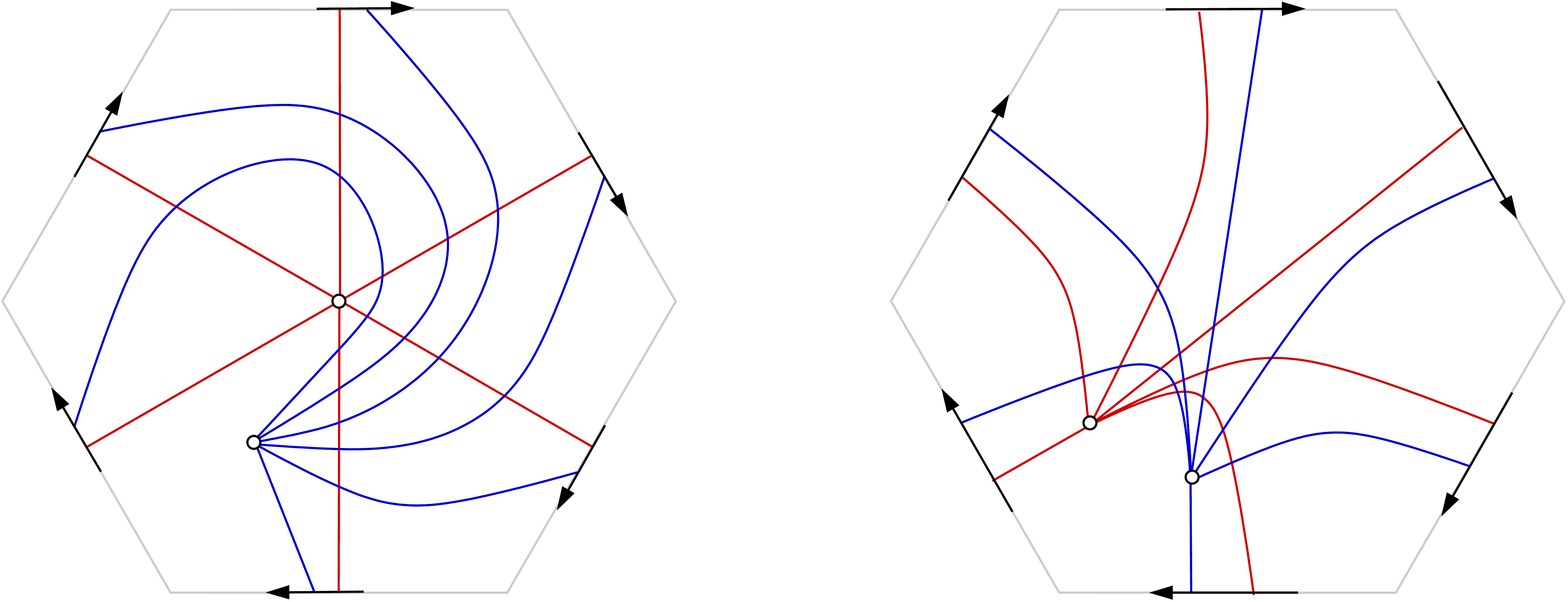}%
\end{picture}%
\setlength{\unitlength}{4144sp}%

\multiply\unitlength by \magproz
\divide\unitlength by 100

\begingroup\makeatletter\ifx\SetFigFont\undefined%
\gdef\SetFigFont#1#2#3#4#5{%
  \reset@font\fontsize{#1}{#2pt}%
  \fontfamily{#3}\fontseries{#4}\fontshape{#5}%
  \selectfont}%
\fi\endgroup%
\begin{picture}(21765,8353)(656,-35803)
\put(5607,-31269){\makebox(0,0)[b]{\smash{{\SetFigFont{12}{14.4}{\rmdefault}{\mddefault}{\updefault}{\color[rgb]{0,0,0}$a$}%
}}}}
\put(3836,-33481){\makebox(0,0)[b]{\smash{{\SetFigFont{12}{14.4}{\rmdefault}{\mddefault}{\updefault}{\color[rgb]{0,0,0}$b$}%
}}}}
\put(3424,-34613){\makebox(0,0)[b]{\smash{{\SetFigFont{12}{14.4}{\rmdefault}{\mddefault}{\updefault}{\color[rgb]{0,0,0}$Q_4(a)$}%
}}}}
\put(3186,-32021){\makebox(0,0)[b]{\smash{{\SetFigFont{12}{14.4}{\rmdefault}{\mddefault}{\updefault}{\color[rgb]{0,0,0}$Q_5(a)$}%
}}}}
\put(8115,-31862){\makebox(0,0)[b]{\smash{{\SetFigFont{12}{14.4}{\rmdefault}{\mddefault}{\updefault}{\color[rgb]{0,0,0}$Q_2(a)$}%
}}}}
\put(6727,-34187){\makebox(0,0)[b]{\smash{{\SetFigFont{12}{14.4}{\rmdefault}{\mddefault}{\updefault}{\color[rgb]{0,0,0}$Q_3(a)$}%
}}}}
\put(4495,-30530){\makebox(0,0)[b]{\smash{{\SetFigFont{12}{14.4}{\rmdefault}{\mddefault}{\updefault}{\color[rgb]{0,0,0}$Q_6(a)$}%
}}}}
\put(6786,-29618){\makebox(0,0)[b]{\smash{{\SetFigFont{12}{14.4}{\rmdefault}{\mddefault}{\updefault}{\color[rgb]{0,0,0}$Q_1(a)$}%
}}}}
\put(4678,-27621){\makebox(0,0)[b]{\smash{{\SetFigFont{12}{14.4}{\rmdefault}{\mddefault}{\updefault}{\color[rgb]{0,0,0}$p_1$}%
}}}}
\put(6155,-35730){\makebox(0,0)[b]{\smash{{\SetFigFont{12}{14.4}{\rmdefault}{\mddefault}{\updefault}{\color[rgb]{0,0,0}$p_4$}%
}}}}
\put(9290,-29550){\makebox(0,0)[b]{\smash{{\SetFigFont{12}{14.4}{\rmdefault}{\mddefault}{\updefault}{\color[rgb]{0,0,0}$p_2$}%
}}}}
\put(9347,-33914){\makebox(0,0)[b]{\smash{{\SetFigFont{12}{14.4}{\rmdefault}{\mddefault}{\updefault}{\color[rgb]{0,0,0}$p_3$}%
}}}}
\put(1293,-33891){\makebox(0,0)[b]{\smash{{\SetFigFont{12}{14.4}{\rmdefault}{\mddefault}{\updefault}{\color[rgb]{0,0,0}$p_5$}%
}}}}
\put(1361,-29448){\makebox(0,0)[b]{\smash{{\SetFigFont{12}{14.4}{\rmdefault}{\mddefault}{\updefault}{\color[rgb]{0,0,0}$p_6$}%
}}}}
\end{picture}%

%% file: Figures/pseudolines-1.pstex_t
\begin{picture}(0,0)%
\includegraphics[scale=\figscale]{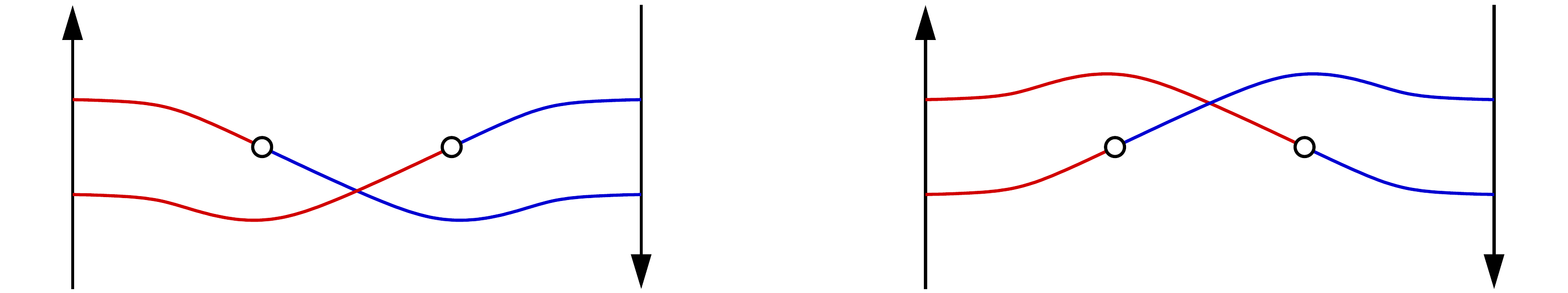}%
\end{picture}%
\setlength{\unitlength}{4144sp}%

\multiply\unitlength by \magproz
\divide\unitlength by 100

\begingroup\makeatletter\ifx\SetFigFont\undefined%
\gdef\SetFigFont#1#2#3#4#5{%
  \reset@font\fontsize{#1}{#2pt}%
  \fontfamily{#3}\fontseries{#4}\fontshape{#5}%
  \selectfont}%
\fi\endgroup%
\begin{picture}(14880,2766)(-239,-16294)
\put(10351,-15586){\makebox(0,0)[b]{\smash{{\SetFigFont{12}{14.4}{\rmdefault}{\mddefault}{\updefault}{\color[rgb]{0,0,0}$a$}%
}}}}
\put(12151,-15586){\makebox(0,0)[b]{\smash{{\SetFigFont{12}{14.4}{\rmdefault}{\mddefault}{\updefault}{\color[rgb]{0,0,0}$b$}%
}}}}
\put(2251,-14461){\makebox(0,0)[b]{\smash{{\SetFigFont{12}{14.4}{\rmdefault}{\mddefault}{\updefault}{\color[rgb]{0,0,0}$b$}%
}}}}
\put(4051,-14461){\makebox(0,0)[b]{\smash{{\SetFigFont{12}{14.4}{\rmdefault}{\mddefault}{\updefault}{\color[rgb]{0,0,0}$a$}%
}}}}
\put(-224,-14911){\makebox(0,0)[b]{\smash{{\SetFigFont{12}{14.4}{\rmdefault}{\mddefault}{\updefault}{\color[rgb]{0,0,0}$A$}%
}}}}
\put(14626,-14911){\makebox(0,0)[b]{\smash{{\SetFigFont{12}{14.4}{\rmdefault}{\mddefault}{\updefault}{\color[rgb]{0,0,0}$B$}%
}}}}
\end{picture}%

%% file: Figures/pseudolines-2.pstex_t
\begin{picture}(0,0)%
\includegraphics[scale=\figscale]{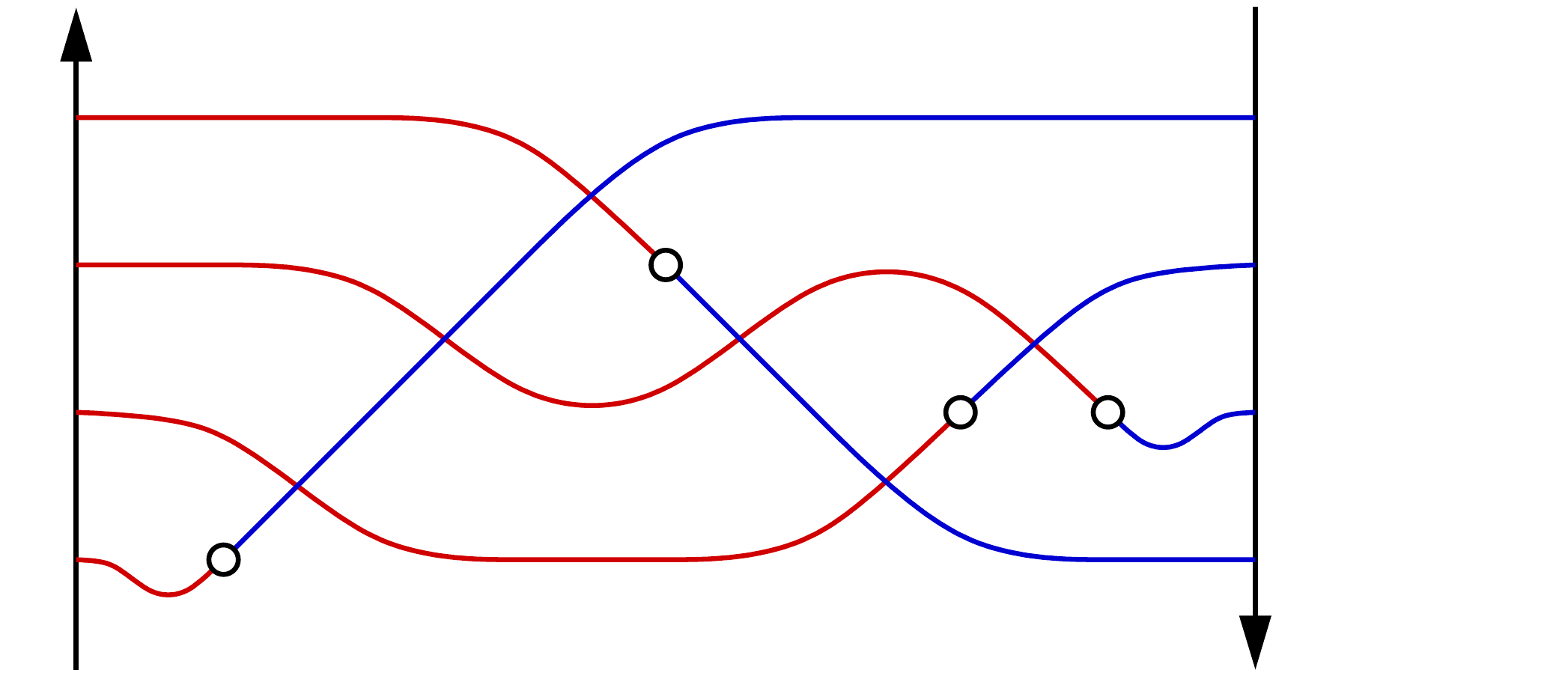}%
\end{picture}%
\setlength{\unitlength}{4144sp}%

\multiply\unitlength by \magproz
\divide\unitlength by 100

\begingroup\makeatletter\ifx\SetFigFont\undefined%
\gdef\SetFigFont#1#2#3#4#5{%
  \reset@font\fontsize{#1}{#2pt}%
  \fontfamily{#3}\fontseries{#4}\fontshape{#5}%
  \selectfont}%
\fi\endgroup%
\begin{picture}(9570,4116)(-464,-22819)
\put(7651,-19411){\makebox(0,0)[b]{\smash{{\SetFigFont{12}{14.4}{\rmdefault}{\mddefault}{\updefault}{\color[rgb]{0,0,0}$1$}%
}}}}
\put(7651,-20311){\makebox(0,0)[b]{\smash{{\SetFigFont{12}{14.4}{\rmdefault}{\mddefault}{\updefault}{\color[rgb]{0,0,0}$2$}%
}}}}
\put(7651,-21211){\makebox(0,0)[b]{\smash{{\SetFigFont{12}{14.4}{\rmdefault}{\mddefault}{\updefault}{\color[rgb]{0,0,0}$3$}%
}}}}
\put(7651,-22111){\makebox(0,0)[b]{\smash{{\SetFigFont{12}{14.4}{\rmdefault}{\mddefault}{\updefault}{\color[rgb]{0,0,0}$4$}%
}}}}
\put(-449,-19411){\makebox(0,0)[b]{\smash{{\SetFigFont{12}{14.4}{\rmdefault}{\mddefault}{\updefault}{\color[rgb]{0,0,0}$4$}%
}}}}
\put(-449,-20311){\makebox(0,0)[b]{\smash{{\SetFigFont{12}{14.4}{\rmdefault}{\mddefault}{\updefault}{\color[rgb]{0,0,0}$3$}%
}}}}
\put(-449,-21211){\makebox(0,0)[b]{\smash{{\SetFigFont{12}{14.4}{\rmdefault}{\mddefault}{\updefault}{\color[rgb]{0,0,0}$2$}%
}}}}
\put(-449,-22111){\makebox(0,0)[b]{\smash{{\SetFigFont{12}{14.4}{\rmdefault}{\mddefault}{\updefault}{\color[rgb]{0,0,0}$1$}%
}}}}
\put(9091,-20851){\makebox(0,0)[lb]{\smash{{\SetFigFont{12}{14.4}{\rmdefault}{\mddefault}{\updefault}{\color[rgb]{0,0,0}\QTableOne}%
}}}}
\end{picture}%

%% file: Figures/triple-both.pstex_t
\begin{picture}(0,0)%
\includegraphics[scale=\figscale]{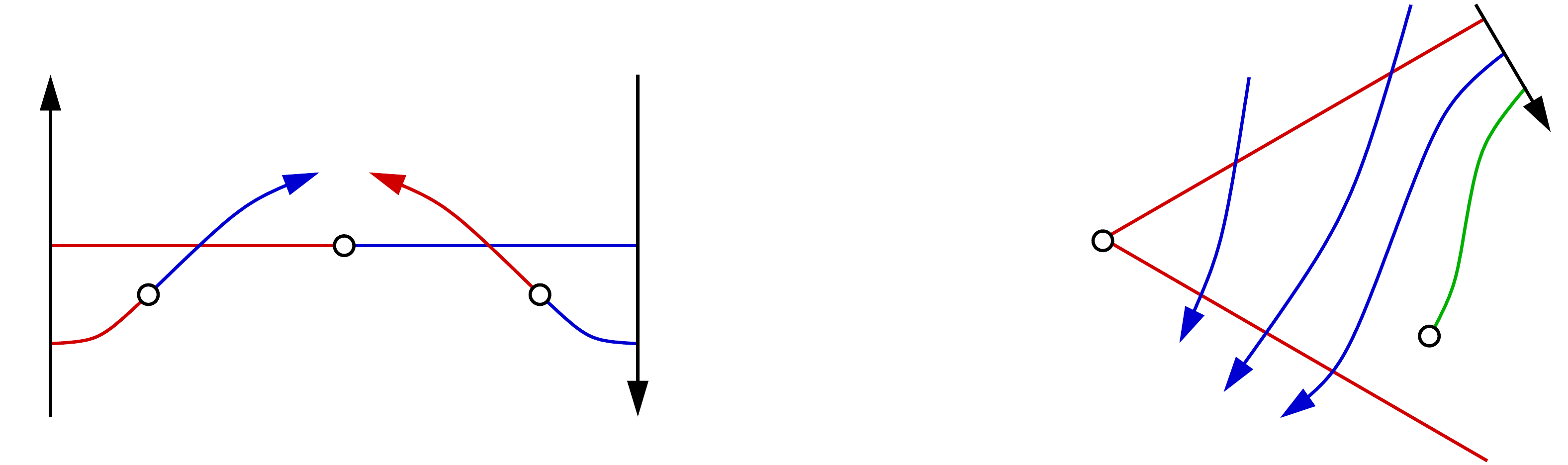}%
\end{picture}%
\setlength{\unitlength}{4144sp}%

\multiply\unitlength by \magproz
\divide\unitlength by 100

\begingroup\makeatletter\ifx\SetFigFont\undefined%
\gdef\SetFigFont#1#2#3#4#5{%
  \reset@font\fontsize{#1}{#2pt}%
  \fontfamily{#3}\fontseries{#4}\fontshape{#5}%
  \selectfont}%
\fi\endgroup%
\begin{picture}(14401,4262)(-14,-27722)
\put(10125,-25329){\makebox(0,0)[b]{\smash{{\SetFigFont{12}{14.4}{\rmdefault}{\mddefault}{\updefault}{\color[rgb]{0,0,0}$a$}%
}}}}
\put(13610,-26604){\makebox(0,0)[b]{\smash{{\SetFigFont{12}{14.4}{\rmdefault}{\mddefault}{\updefault}{\color[rgb]{0,0,0}$c$}%
}}}}
\put(  1,-25801){\makebox(0,0)[b]{\smash{{\SetFigFont{12}{14.4}{\rmdefault}{\mddefault}{\updefault}{\color[rgb]{0,0,0}$b$}%
}}}}
\put(  1,-24901){\makebox(0,0)[b]{\smash{{\SetFigFont{12}{14.4}{\rmdefault}{\mddefault}{\updefault}{\color[rgb]{0,0,0}$c$}%
}}}}
\put(6301,-24901){\makebox(0,0)[b]{\smash{{\SetFigFont{12}{14.4}{\rmdefault}{\mddefault}{\updefault}{\color[rgb]{0,0,0}$a$}%
}}}}
\put(6301,-25801){\makebox(0,0)[b]{\smash{{\SetFigFont{12}{14.4}{\rmdefault}{\mddefault}{\updefault}{\color[rgb]{0,0,0}$b$}%
}}}}
\put(6301,-26701){\makebox(0,0)[b]{\smash{{\SetFigFont{12}{14.4}{\rmdefault}{\mddefault}{\updefault}{\color[rgb]{0,0,0}$c$}%
}}}}
\put(  1,-26701){\makebox(0,0)[b]{\smash{{\SetFigFont{12}{14.4}{\rmdefault}{\mddefault}{\updefault}{\color[rgb]{0,0,0}$a$}%
}}}}
\put(14372,-23960){\makebox(0,0)[b]{\smash{{\SetFigFont{12}{14.4}{\rmdefault}{\mddefault}{\updefault}{\color[rgb]{0,0,0}$p_j$}%
}}}}
\end{picture}%

%% file: Figures/quadruple.pstex_t
\begin{picture}(0,0)%
\includegraphics[scale=\figscale]{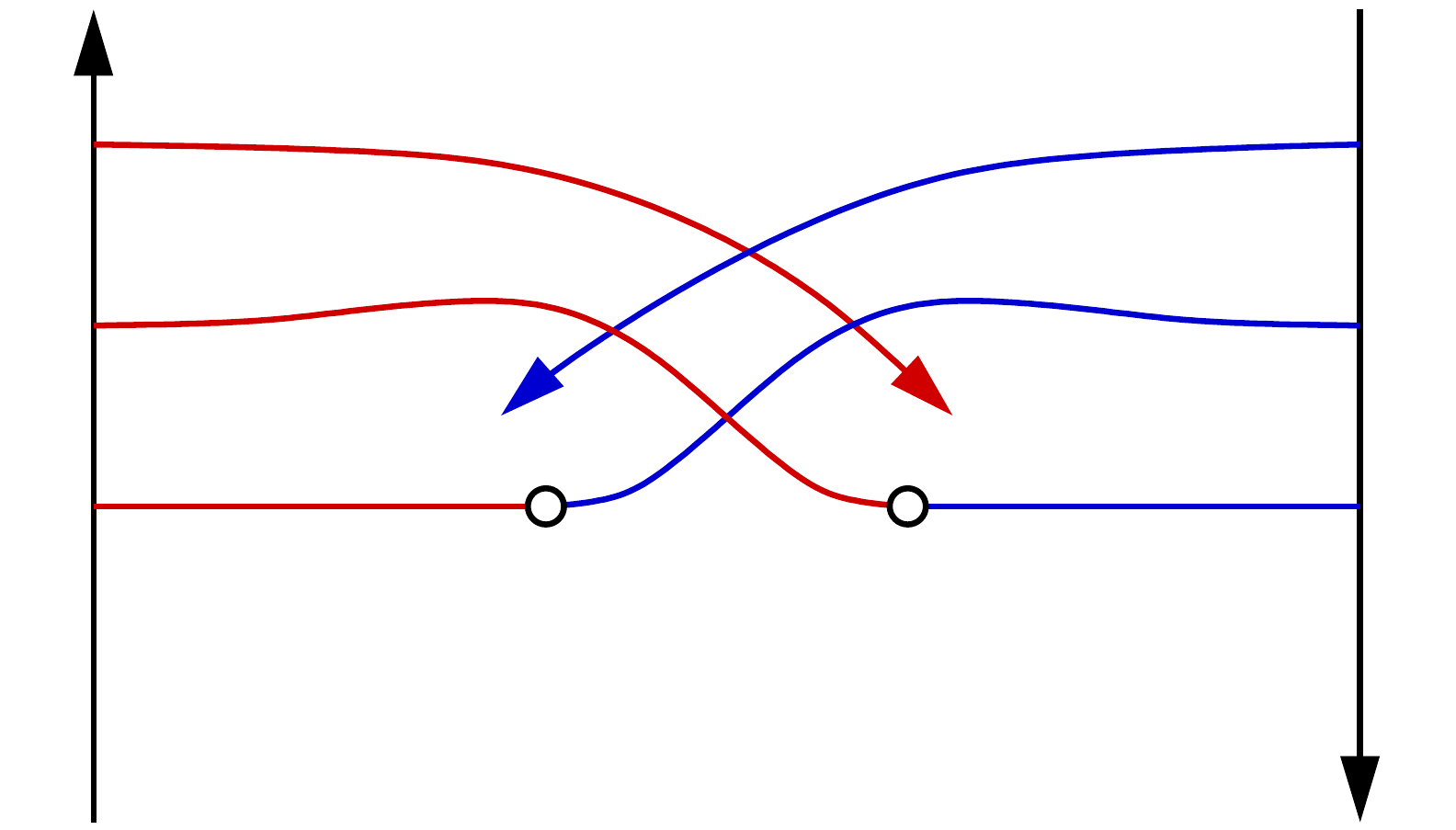}%
\end{picture}%
\setlength{\unitlength}{4144sp}%

\multiply\unitlength by \magproz
\divide\unitlength by 100

\begingroup\makeatletter\ifx\SetFigFont\undefined%
\gdef\SetFigFont#1#2#3#4#5{%
  \reset@font\fontsize{#1}{#2pt}%
  \fontfamily{#3}\fontseries{#4}\fontshape{#5}%
  \selectfont}%
\fi\endgroup%
\begin{picture}(7230,4116)(7636,-29119)
\put(7651,-28501){\makebox(0,0)[b]{\smash{{\SetFigFont{12}{14.4}{\rmdefault}{\mddefault}{\updefault}{\color[rgb]{0,0,0}$a$}%
}}}}
\put(7651,-27601){\makebox(0,0)[b]{\smash{{\SetFigFont{12}{14.4}{\rmdefault}{\mddefault}{\updefault}{\color[rgb]{0,0,0}$b$}%
}}}}
\put(7651,-26701){\makebox(0,0)[b]{\smash{{\SetFigFont{12}{14.4}{\rmdefault}{\mddefault}{\updefault}{\color[rgb]{0,0,0}$c$}%
}}}}
\put(7651,-25801){\makebox(0,0)[b]{\smash{{\SetFigFont{12}{14.4}{\rmdefault}{\mddefault}{\updefault}{\color[rgb]{0,0,0}$d$}%
}}}}
\put(14851,-25801){\makebox(0,0)[b]{\smash{{\SetFigFont{12}{14.4}{\rmdefault}{\mddefault}{\updefault}{\color[rgb]{0,0,0}$a$}%
}}}}
\put(14851,-26701){\makebox(0,0)[b]{\smash{{\SetFigFont{12}{14.4}{\rmdefault}{\mddefault}{\updefault}{\color[rgb]{0,0,0}$b$}%
}}}}
\put(14851,-27601){\makebox(0,0)[b]{\smash{{\SetFigFont{12}{14.4}{\rmdefault}{\mddefault}{\updefault}{\color[rgb]{0,0,0}$c$}%
}}}}
\put(14851,-28501){\makebox(0,0)[b]{\smash{{\SetFigFont{12}{14.4}{\rmdefault}{\mddefault}{\updefault}{\color[rgb]{0,0,0}$d$}%
}}}}
\end{picture}%

%% file: Figures/threetypes-2.pstex_t
\begin{picture}(0,0)%
\includegraphics[scale=\figscale]{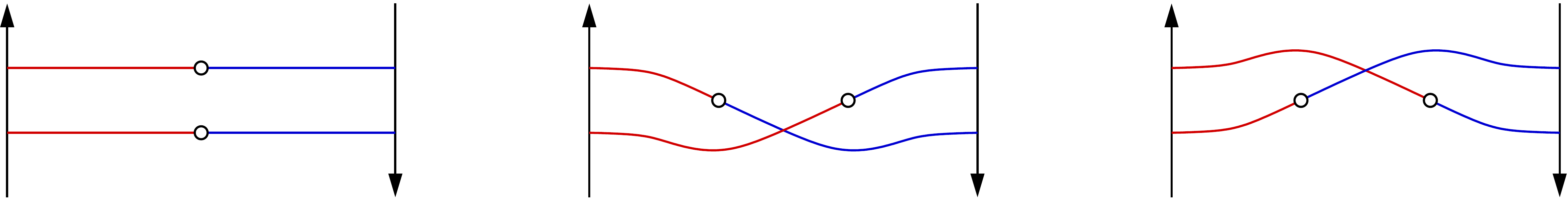}%
\end{picture}%
\setlength{\unitlength}{4144sp}%

\multiply\unitlength by \magproz
\divide\unitlength by 100

\begingroup\makeatletter\ifx\SetFigFont\undefined%
\gdef\SetFigFont#1#2#3#4#5{%
  \reset@font\fontsize{#1}{#2pt}%
  \fontfamily{#3}\fontseries{#4}\fontshape{#5}%
  \selectfont}%
\fi\endgroup%
\begin{picture}(21799,2766)(-7748,-16294)
\put(10351,-15586){\makebox(0,0)[b]{\smash{{\SetFigFont{12}{14.4}{\rmdefault}{\mddefault}{\updefault}{\color[rgb]{0,0,0}$a$}%
}}}}
\put(12151,-15586){\makebox(0,0)[b]{\smash{{\SetFigFont{12}{14.4}{\rmdefault}{\mddefault}{\updefault}{\color[rgb]{0,0,0}$b$}%
}}}}
\put(2251,-14461){\makebox(0,0)[b]{\smash{{\SetFigFont{12}{14.4}{\rmdefault}{\mddefault}{\updefault}{\color[rgb]{0,0,0}$b$}%
}}}}
\put(4051,-14461){\makebox(0,0)[b]{\smash{{\SetFigFont{12}{14.4}{\rmdefault}{\mddefault}{\updefault}{\color[rgb]{0,0,0}$a$}%
}}}}
\put(-4949,-15811){\makebox(0,0)[b]{\smash{{\SetFigFont{12}{14.4}{\rmdefault}{\mddefault}{\updefault}{\color[rgb]{0,0,0}$a$}%
}}}}
\put(-4949,-14011){\makebox(0,0)[b]{\smash{{\SetFigFont{12}{14.4}{\rmdefault}{\mddefault}{\updefault}{\color[rgb]{0,0,0}$b$}%
}}}}
\put(-7019,-14011){\makebox(0,0)[b]{\smash{{\SetFigFont{12}{14.4}{\rmdefault}{\mddefault}{\updefault}{\color[rgb]{0,0,0}$N$}%
}}}}
\put(1081,-14011){\makebox(0,0)[b]{\smash{{\SetFigFont{12}{14.4}{\rmdefault}{\mddefault}{\updefault}{\color[rgb]{0,0,0}$A$}%
}}}}
\put(9181,-14011){\makebox(0,0)[b]{\smash{{\SetFigFont{12}{14.4}{\rmdefault}{\mddefault}{\updefault}{\color[rgb]{0,0,0}$B$}%
}}}}
\end{picture}%

%% file: Figures/feasible.pstex_t
\begin{picture}(0,0)%
\includegraphics[scale=\figscale]{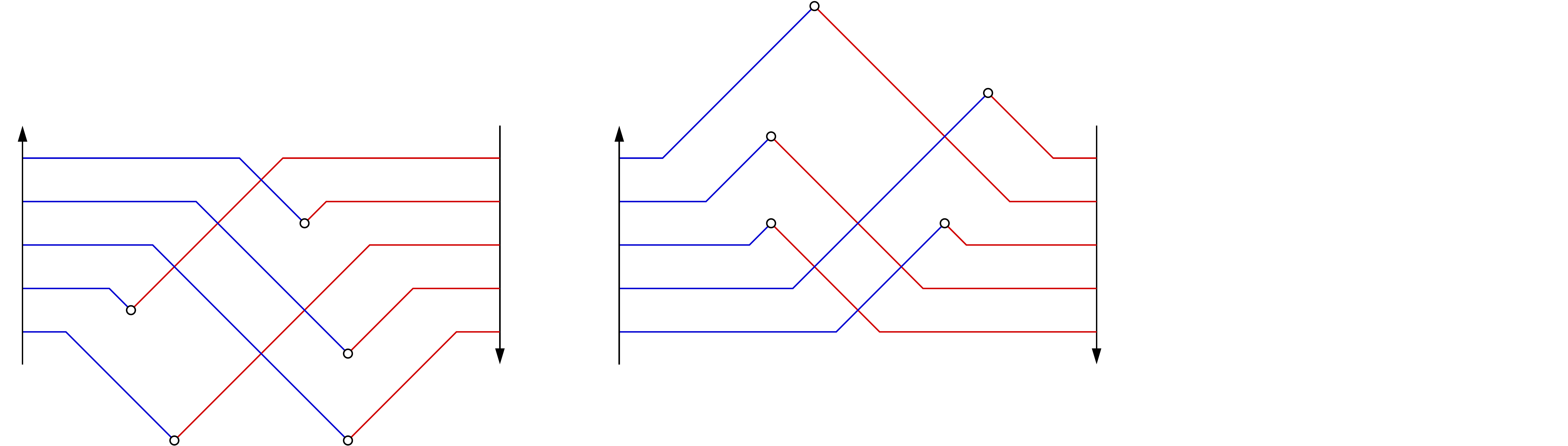}%
\end{picture}%
\setlength{\unitlength}{4144sp}%

\multiply\unitlength by \magproz
\divide\unitlength by 100

\begingroup\makeatletter\ifx\SetFigFont\undefined%
\gdef\SetFigFont#1#2#3#4#5{%
  \reset@font\fontsize{#1}{#2pt}%
  \fontfamily{#3}\fontseries{#4}\fontshape{#5}%
  \selectfont}%
\fi\endgroup%
\begin{picture}(32505,9224)(-16664,-24473)
\put(6526,-22321){\makebox(0,0)[b]{\smash{{\SetFigFont{12}{14.4}{\rmdefault}{\mddefault}{\updefault}{\color[rgb]{0,0,0}$3$}%
}}}}
\put(6526,-21421){\makebox(0,0)[b]{\smash{{\SetFigFont{12}{14.4}{\rmdefault}{\mddefault}{\updefault}{\color[rgb]{0,0,0}$4$}%
}}}}
\put(6526,-20521){\makebox(0,0)[b]{\smash{{\SetFigFont{12}{14.4}{\rmdefault}{\mddefault}{\updefault}{\color[rgb]{0,0,0}$1$}%
}}}}
\put(6526,-19621){\makebox(0,0)[b]{\smash{{\SetFigFont{12}{14.4}{\rmdefault}{\mddefault}{\updefault}{\color[rgb]{0,0,0}$5$}%
}}}}
\put(6526,-18721){\makebox(0,0)[b]{\smash{{\SetFigFont{12}{14.4}{\rmdefault}{\mddefault}{\updefault}{\color[rgb]{0,0,0}$2$}%
}}}}
\put(-5849,-22336){\makebox(0,0)[b]{\smash{{\SetFigFont{12}{14.4}{\rmdefault}{\mddefault}{\updefault}{\color[rgb]{0,0,0}$3$}%
}}}}
\put(-5849,-21436){\makebox(0,0)[b]{\smash{{\SetFigFont{12}{14.4}{\rmdefault}{\mddefault}{\updefault}{\color[rgb]{0,0,0}$4$}%
}}}}
\put(-5849,-20536){\makebox(0,0)[b]{\smash{{\SetFigFont{12}{14.4}{\rmdefault}{\mddefault}{\updefault}{\color[rgb]{0,0,0}$1$}%
}}}}
\put(-5849,-19636){\makebox(0,0)[b]{\smash{{\SetFigFont{12}{14.4}{\rmdefault}{\mddefault}{\updefault}{\color[rgb]{0,0,0}$5$}%
}}}}
\put(-5849,-18736){\makebox(0,0)[b]{\smash{{\SetFigFont{12}{14.4}{\rmdefault}{\mddefault}{\updefault}{\color[rgb]{0,0,0}$2$}%
}}}}
\put(-4274,-19636){\makebox(0,0)[b]{\smash{{\SetFigFont{12}{14.4}{\rmdefault}{\mddefault}{\updefault}{\color[rgb]{0,0,0}$4$}%
}}}}
\put(-4274,-20536){\makebox(0,0)[b]{\smash{{\SetFigFont{12}{14.4}{\rmdefault}{\mddefault}{\updefault}{\color[rgb]{0,0,0}$3$}%
}}}}
\put(-4274,-21436){\makebox(0,0)[b]{\smash{{\SetFigFont{12}{14.4}{\rmdefault}{\mddefault}{\updefault}{\color[rgb]{0,0,0}$2$}%
}}}}
\put(-4274,-22336){\makebox(0,0)[b]{\smash{{\SetFigFont{12}{14.4}{\rmdefault}{\mddefault}{\updefault}{\color[rgb]{0,0,0}$1$}%
}}}}
\put(-4274,-18736){\makebox(0,0)[b]{\smash{{\SetFigFont{12}{14.4}{\rmdefault}{\mddefault}{\updefault}{\color[rgb]{0,0,0}$5$}%
}}}}
\put(-16649,-19636){\makebox(0,0)[b]{\smash{{\SetFigFont{12}{14.4}{\rmdefault}{\mddefault}{\updefault}{\color[rgb]{0,0,0}$4$}%
}}}}
\put(-16649,-20536){\makebox(0,0)[b]{\smash{{\SetFigFont{12}{14.4}{\rmdefault}{\mddefault}{\updefault}{\color[rgb]{0,0,0}$3$}%
}}}}
\put(-16649,-21436){\makebox(0,0)[b]{\smash{{\SetFigFont{12}{14.4}{\rmdefault}{\mddefault}{\updefault}{\color[rgb]{0,0,0}$2$}%
}}}}
\put(-16649,-22336){\makebox(0,0)[b]{\smash{{\SetFigFont{12}{14.4}{\rmdefault}{\mddefault}{\updefault}{\color[rgb]{0,0,0}$1$}%
}}}}
\put(-16649,-18736){\makebox(0,0)[b]{\smash{{\SetFigFont{12}{14.4}{\rmdefault}{\mddefault}{\updefault}{\color[rgb]{0,0,0}$5$}%
}}}}
\put(7743,-20309){\makebox(0,0)[lb]{\smash{{\SetFigFont{12}{14.4}{\rmdefault}{\mddefault}{\updefault}{\color[rgb]{0,0,0}\QTableAB}%
}}}}
\end{picture}%

%% file: Figures/qruleFail.pstex_t
\begin{picture}(0,0)%
\includegraphics[scale=\figscale]{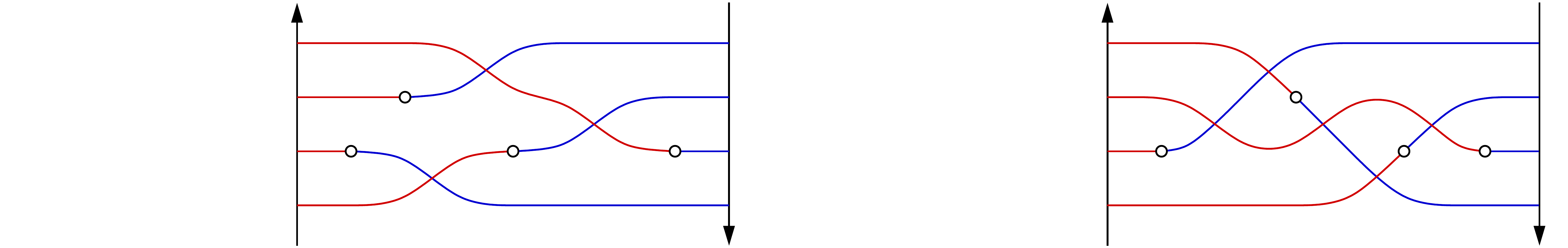}%
\end{picture}%
\setlength{\unitlength}{4144sp}%

\multiply\unitlength by \magproz
\divide\unitlength by 100

\begingroup\makeatletter\ifx\SetFigFont\undefined%
\gdef\SetFigFont#1#2#3#4#5{%
  \reset@font\fontsize{#1}{#2pt}%
  \fontfamily{#3}\fontseries{#4}\fontshape{#5}%
  \selectfont}%
\fi\endgroup%
\begin{picture}(26124,4116)(-3150,-22819)
\put(14851,-19411){\makebox(0,0)[b]{\smash{{\SetFigFont{12}{14.4}{\rmdefault}{\mddefault}{\updefault}{\color[rgb]{0,0,0}$_4$}%
}}}}
\put(14851,-20311){\makebox(0,0)[b]{\smash{{\SetFigFont{12}{14.4}{\rmdefault}{\mddefault}{\updefault}{\color[rgb]{0,0,0}$_3$}%
}}}}
\put(14851,-21211){\makebox(0,0)[b]{\smash{{\SetFigFont{12}{14.4}{\rmdefault}{\mddefault}{\updefault}{\color[rgb]{0,0,0}$_2$}%
}}}}
\put(14851,-22111){\makebox(0,0)[b]{\smash{{\SetFigFont{12}{14.4}{\rmdefault}{\mddefault}{\updefault}{\color[rgb]{0,0,0}$_1$}%
}}}}
\put(1351,-19411){\makebox(0,0)[b]{\smash{{\SetFigFont{12}{14.4}{\rmdefault}{\mddefault}{\updefault}{\color[rgb]{0,0,0}$_4$}%
}}}}
\put(1351,-20311){\makebox(0,0)[b]{\smash{{\SetFigFont{12}{14.4}{\rmdefault}{\mddefault}{\updefault}{\color[rgb]{0,0,0}$_3$}%
}}}}
\put(1351,-21211){\makebox(0,0)[b]{\smash{{\SetFigFont{12}{14.4}{\rmdefault}{\mddefault}{\updefault}{\color[rgb]{0,0,0}$_2$}%
}}}}
\put(1351,-22111){\makebox(0,0)[b]{\smash{{\SetFigFont{12}{14.4}{\rmdefault}{\mddefault}{\updefault}{\color[rgb]{0,0,0}$_1$}%
}}}}
\put(22959,-19411){\makebox(0,0)[b]{\smash{{\SetFigFont{12}{14.4}{\rmdefault}{\mddefault}{\updefault}{\color[rgb]{0,0,0}$_2$}%
}}}}
\put(22959,-21211){\makebox(0,0)[b]{\smash{{\SetFigFont{12}{14.4}{\rmdefault}{\mddefault}{\updefault}{\color[rgb]{0,0,0}$_3$}%
}}}}
\put(22959,-20311){\makebox(0,0)[b]{\smash{{\SetFigFont{12}{14.4}{\rmdefault}{\mddefault}{\updefault}{\color[rgb]{0,0,0}$_1$}%
}}}}
\put(22959,-22111){\makebox(0,0)[b]{\smash{{\SetFigFont{12}{14.4}{\rmdefault}{\mddefault}{\updefault}{\color[rgb]{0,0,0}$_4$}%
}}}}
\put(9451,-19411){\makebox(0,0)[b]{\smash{{\SetFigFont{12}{14.4}{\rmdefault}{\mddefault}{\updefault}{\color[rgb]{0,0,0}$_3$}%
}}}}
\put(9451,-21211){\makebox(0,0)[b]{\smash{{\SetFigFont{12}{14.4}{\rmdefault}{\mddefault}{\updefault}{\color[rgb]{0,0,0}$_4$}%
}}}}
\put(9451,-20311){\makebox(0,0)[b]{\smash{{\SetFigFont{12}{14.4}{\rmdefault}{\mddefault}{\updefault}{\color[rgb]{0,0,0}$_1$}%
}}}}
\put(9451,-22111){\makebox(0,0)[b]{\smash{{\SetFigFont{12}{14.4}{\rmdefault}{\mddefault}{\updefault}{\color[rgb]{0,0,0}$_2$}%
}}}}
\put(901,-21211){\makebox(0,0)[rb]{\smash{{\SetFigFont{12}{14.4}{\rmdefault}{\mddefault}{\updefault}{\color[rgb]{0,0,0}\QTableTwoN}%
}}}}
\put(14401,-21211){\makebox(0,0)[rb]{\smash{{\SetFigFont{12}{14.4}{\rmdefault}{\mddefault}{\updefault}{\color[rgb]{0,0,0}\QTableTwoB}%
}}}}
\end{picture}%

%% file: Figures/local.pstex_t
\begin{picture}(0,0)%
\includegraphics[scale=\figscale]{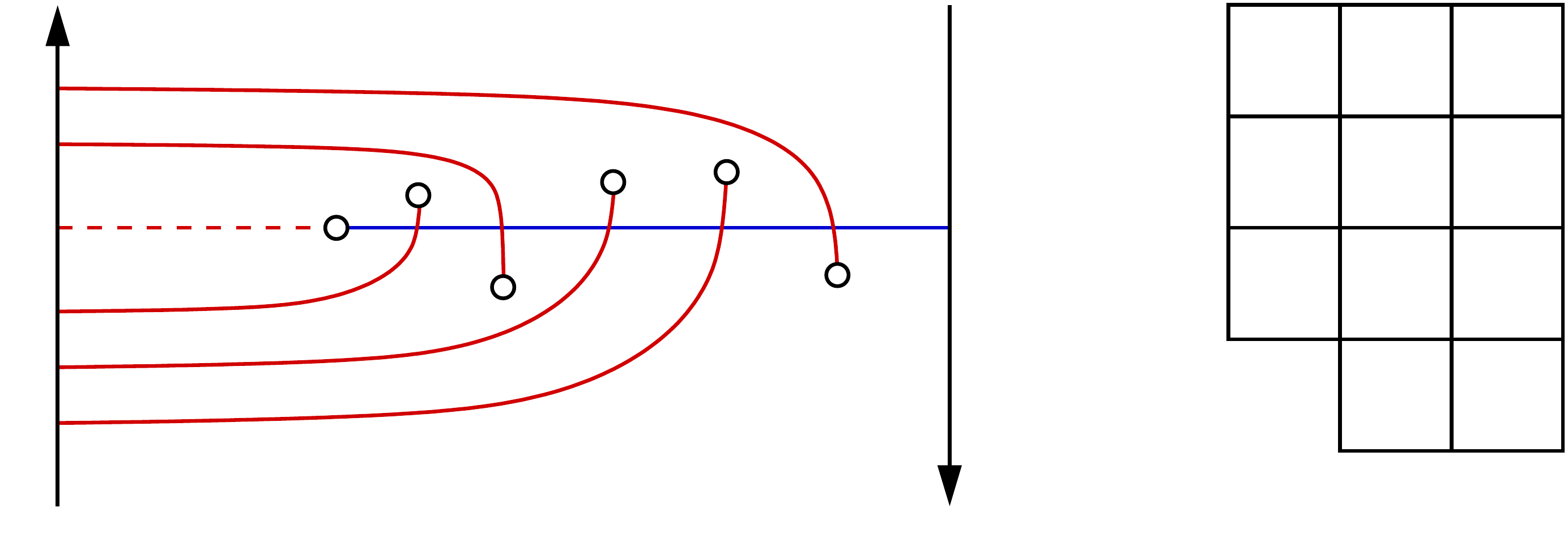}%
\end{picture}%
\setlength{\unitlength}{4144sp}%

\multiply\unitlength by \magproz
\divide\unitlength by 100

\begingroup\makeatletter\ifx\SetFigFont\undefined%
\gdef\SetFigFont#1#2#3#4#5{%
  \reset@font\fontsize{#1}{#2pt}%
  \fontfamily{#3}\fontseries{#4}\fontshape{#5}%
  \selectfont}%
\fi\endgroup%
\begin{picture}(12648,4327)(-464,-27755)
\put(9001,-24091){\makebox(0,0)[b]{\smash{{\SetFigFont{12}{14.4}{\rmdefault}{\mddefault}{\updefault}{\color[rgb]{0,0,0}$a$}%
}}}}
\put(9001,-24991){\makebox(0,0)[b]{\smash{{\SetFigFont{12}{14.4}{\rmdefault}{\mddefault}{\updefault}{\color[rgb]{0,0,0}$b$}%
}}}}
\put(9001,-25891){\makebox(0,0)[b]{\smash{{\SetFigFont{12}{14.4}{\rmdefault}{\mddefault}{\updefault}{\color[rgb]{0,0,0}$c$}%
}}}}
\put(10801,-27691){\makebox(0,0)[b]{\smash{{\SetFigFont{12}{14.4}{\rmdefault}{\mddefault}{\updefault}{\color[rgb]{0,0,0}$d$}%
}}}}
\put(11701,-27691){\makebox(0,0)[b]{\smash{{\SetFigFont{12}{14.4}{\rmdefault}{\mddefault}{\updefault}{\color[rgb]{0,0,0}$e$}%
}}}}
\put(9901,-26791){\makebox(0,0)[b]{\smash{{\SetFigFont{12}{14.4}{\rmdefault}{\mddefault}{\updefault}{\color[rgb]{0,0,0}$i$}%
}}}}
\put(10801,-26791){\makebox(0,0)[b]{\smash{{\SetFigFont{12}{14.4}{\rmdefault}{\mddefault}{\updefault}{\color[rgb]{0,0,0}$B$}%
}}}}
\put(11701,-26791){\makebox(0,0)[b]{\smash{{\SetFigFont{12}{14.4}{\rmdefault}{\mddefault}{\updefault}{\color[rgb]{0,0,0}$B$}%
}}}}
\put(11701,-25891){\makebox(0,0)[b]{\smash{{\SetFigFont{12}{14.4}{\rmdefault}{\mddefault}{\updefault}{\color[rgb]{0,0,0}$B$}%
}}}}
\put(10801,-25891){\makebox(0,0)[b]{\smash{{\SetFigFont{12}{14.4}{\rmdefault}{\mddefault}{\updefault}{\color[rgb]{0,0,0}$B$}%
}}}}
\put(9901,-25891){\makebox(0,0)[b]{\smash{{\SetFigFont{12}{14.4}{\rmdefault}{\mddefault}{\updefault}{\color[rgb]{0,0,0}$A$}%
}}}}
\put(9901,-24991){\makebox(0,0)[b]{\smash{{\SetFigFont{12}{14.4}{\rmdefault}{\mddefault}{\updefault}{\color[rgb]{0,0,0}$A$}%
}}}}
\put(10801,-24991){\makebox(0,0)[b]{\smash{{\SetFigFont{12}{14.4}{\rmdefault}{\mddefault}{\updefault}{\color[rgb]{0,0,0}$A$}%
}}}}
\put(9901,-24091){\makebox(0,0)[b]{\smash{{\SetFigFont{12}{14.4}{\rmdefault}{\mddefault}{\updefault}{\color[rgb]{0,0,0}$A$}%
}}}}
\put(10801,-24091){\makebox(0,0)[b]{\smash{{\SetFigFont{12}{14.4}{\rmdefault}{\mddefault}{\updefault}{\color[rgb]{0,0,0}$A$}%
}}}}
\put(11701,-24091){\makebox(0,0)[b]{\smash{{\SetFigFont{12}{14.4}{\rmdefault}{\mddefault}{\updefault}{\color[rgb]{0,0,0}$B$}%
}}}}
\put(11701,-24991){\makebox(0,0)[b]{\smash{{\SetFigFont{12}{14.4}{\rmdefault}{\mddefault}{\updefault}{\color[rgb]{0,0,0}$B$}%
}}}}
\put(-449,-26971){\makebox(0,0)[b]{\smash{{\SetFigFont{12}{14.4}{\rmdefault}{\mddefault}{\updefault}{\color[rgb]{0,0,0}$a$}%
}}}}
\put(-449,-24721){\makebox(0,0)[b]{\smash{{\SetFigFont{12}{14.4}{\rmdefault}{\mddefault}{\updefault}{\color[rgb]{0,0,0}$d$}%
}}}}
\put(-449,-25981){\makebox(0,0)[b]{\smash{{\SetFigFont{12}{14.4}{\rmdefault}{\mddefault}{\updefault}{\color[rgb]{0,0,0}$c$}%
}}}}
\put(-449,-24181){\makebox(0,0)[b]{\smash{{\SetFigFont{12}{14.4}{\rmdefault}{\mddefault}{\updefault}{\color[rgb]{0,0,0}$e$}%
}}}}
\put(-449,-26540){\makebox(0,0)[b]{\smash{{\SetFigFont{12}{14.4}{\rmdefault}{\mddefault}{\updefault}{\color[rgb]{0,0,0}$b$}%
}}}}
\put(7651,-25308){\makebox(0,0)[b]{\smash{{\SetFigFont{12}{14.4}{\rmdefault}{\mddefault}{\updefault}{\color[rgb]{0,0,0}$i$}%
}}}}
\end{picture}%

%% file: Figures/sixtypes.pstex_t
\begin{picture}(0,0)%
\includegraphics[scale=\figscale]{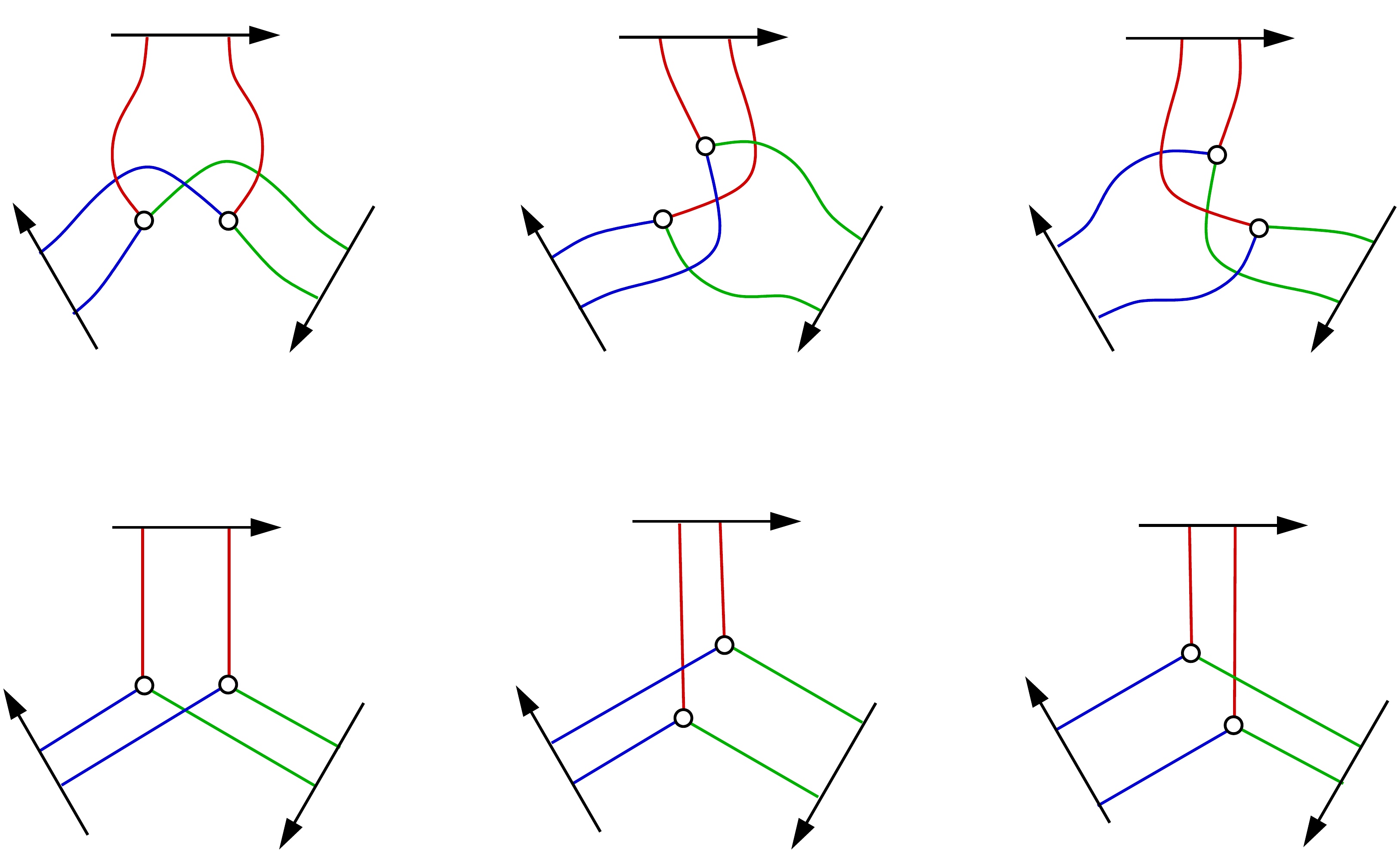}%
\end{picture}%
\setlength{\unitlength}{4144sp}%

\multiply\unitlength by \magproz
\divide\unitlength by 100

\begingroup\makeatletter\ifx\SetFigFont\undefined%
\gdef\SetFigFont#1#2#3#4#5{%
  \reset@font\fontsize{#1}{#2pt}%
  \fontfamily{#3}\fontseries{#4}\fontshape{#5}%
  \selectfont}%
\fi\endgroup%
\begin{picture}(14861,9063)(-3145,-27703)
\put(6199,-21333){\rotatebox{240.0}{\makebox(0,0)[b]{\smash{{\SetFigFont{12}{14.4}{\rmdefault}{\mddefault}{\updefault}{\color[rgb]{0,0,0}$1$}%
}}}}}
\put(5862,-21917){\rotatebox{240.0}{\makebox(0,0)[b]{\smash{{\SetFigFont{12}{14.4}{\rmdefault}{\mddefault}{\updefault}{\color[rgb]{0,0,0}$2$}%
}}}}}
\put(11577,-21333){\rotatebox{240.0}{\makebox(0,0)[b]{\smash{{\SetFigFont{12}{14.4}{\rmdefault}{\mddefault}{\updefault}{\color[rgb]{0,0,0}$1$}%
}}}}}
\put(11240,-21917){\rotatebox{240.0}{\makebox(0,0)[b]{\smash{{\SetFigFont{12}{14.4}{\rmdefault}{\mddefault}{\updefault}{\color[rgb]{0,0,0}$2$}%
}}}}}
\put(689,-26613){\rotatebox{240.0}{\makebox(0,0)[b]{\smash{{\SetFigFont{12}{14.4}{\rmdefault}{\mddefault}{\updefault}{\color[rgb]{0,0,0}$2$}%
}}}}}
\put(351,-27198){\rotatebox{240.0}{\makebox(0,0)[b]{\smash{{\SetFigFont{12}{14.4}{\rmdefault}{\mddefault}{\updefault}{\color[rgb]{0,0,0}$1$}%
}}}}}
\put(-2927,-26703){\rotatebox{120.0}{\makebox(0,0)[b]{\smash{{\SetFigFont{12}{14.4}{\rmdefault}{\mddefault}{\updefault}{\color[rgb]{0,0,0}$1$}%
}}}}}
\put(-2589,-27278){\rotatebox{120.0}{\makebox(0,0)[b]{\smash{{\SetFigFont{12}{14.4}{\rmdefault}{\mddefault}{\updefault}{\color[rgb]{0,0,0}$2$}%
}}}}}
\put(2817,-27211){\rotatebox{120.0}{\makebox(0,0)[b]{\smash{{\SetFigFont{12}{14.4}{\rmdefault}{\mddefault}{\updefault}{\color[rgb]{0,0,0}$1$}%
}}}}}
\put(2490,-26617){\rotatebox{120.0}{\makebox(0,0)[b]{\smash{{\SetFigFont{12}{14.4}{\rmdefault}{\mddefault}{\updefault}{\color[rgb]{0,0,0}$2$}%
}}}}}
\put(11219,-27235){\rotatebox{240.0}{\makebox(0,0)[b]{\smash{{\SetFigFont{12}{14.4}{\rmdefault}{\mddefault}{\updefault}{\color[rgb]{0,0,0}$2$}%
}}}}}
\put(11576,-26581){\rotatebox{240.0}{\makebox(0,0)[b]{\smash{{\SetFigFont{12}{14.4}{\rmdefault}{\mddefault}{\updefault}{\color[rgb]{0,0,0}$1$}%
}}}}}
\put(10044,-23988){\makebox(0,0)[b]{\smash{{\SetFigFont{12}{14.4}{\rmdefault}{\mddefault}{\updefault}{\color[rgb]{0,0,0}$2$}%
}}}}
\put(9359,-23988){\makebox(0,0)[b]{\smash{{\SetFigFont{12}{14.4}{\rmdefault}{\mddefault}{\updefault}{\color[rgb]{0,0,0}$1$}%
}}}}
\put(4669,-23955){\makebox(0,0)[b]{\smash{{\SetFigFont{12}{14.4}{\rmdefault}{\mddefault}{\updefault}{\color[rgb]{0,0,0}$2$}%
}}}}
\put(4044,-23955){\makebox(0,0)[b]{\smash{{\SetFigFont{12}{14.4}{\rmdefault}{\mddefault}{\updefault}{\color[rgb]{0,0,0}$1$}%
}}}}
\put(-765,-24020){\makebox(0,0)[b]{\smash{{\SetFigFont{12}{14.4}{\rmdefault}{\mddefault}{\updefault}{\color[rgb]{0,0,0}$2$}%
}}}}
\put(-1540,-24020){\makebox(0,0)[b]{\smash{{\SetFigFont{12}{14.4}{\rmdefault}{\mddefault}{\updefault}{\color[rgb]{0,0,0}$1$}%
}}}}
\put(6192,-26483){\rotatebox{240.0}{\makebox(0,0)[b]{\smash{{\SetFigFont{12}{14.4}{\rmdefault}{\mddefault}{\updefault}{\color[rgb]{0,0,0}$2$}%
}}}}}
\put(5764,-27218){\rotatebox{240.0}{\makebox(0,0)[b]{\smash{{\SetFigFont{12}{14.4}{\rmdefault}{\mddefault}{\updefault}{\color[rgb]{0,0,0}$1$}%
}}}}}
\put(7985,-26663){\rotatebox{120.0}{\makebox(0,0)[b]{\smash{{\SetFigFont{12}{14.4}{\rmdefault}{\mddefault}{\updefault}{\color[rgb]{0,0,0}$1$}%
}}}}}
\put(8333,-27318){\rotatebox{120.0}{\makebox(0,0)[b]{\smash{{\SetFigFont{12}{14.4}{\rmdefault}{\mddefault}{\updefault}{\color[rgb]{0,0,0}$2$}%
}}}}}
\put(8266,-22093){\rotatebox{120.0}{\makebox(0,0)[b]{\smash{{\SetFigFont{12}{14.4}{\rmdefault}{\mddefault}{\updefault}{\color[rgb]{0,0,0}$1$}%
}}}}}
\put(7929,-21509){\rotatebox{120.0}{\makebox(0,0)[b]{\smash{{\SetFigFont{12}{14.4}{\rmdefault}{\mddefault}{\updefault}{\color[rgb]{0,0,0}$2$}%
}}}}}
\put(9268,-18820){\makebox(0,0)[b]{\smash{{\SetFigFont{12}{14.4}{\rmdefault}{\mddefault}{\updefault}{\color[rgb]{0,0,0}$1$}%
}}}}
\put(9943,-18820){\makebox(0,0)[b]{\smash{{\SetFigFont{12}{14.4}{\rmdefault}{\mddefault}{\updefault}{\color[rgb]{0,0,0}$2$}%
}}}}
\put(-1514,-18787){\makebox(0,0)[b]{\smash{{\SetFigFont{12}{14.4}{\rmdefault}{\mddefault}{\updefault}{\color[rgb]{0,0,0}$1$}%
}}}}
\put(-839,-18787){\makebox(0,0)[b]{\smash{{\SetFigFont{12}{14.4}{\rmdefault}{\mddefault}{\updefault}{\color[rgb]{0,0,0}$2$}%
}}}}
\put(-2534,-22072){\rotatebox{120.0}{\makebox(0,0)[b]{\smash{{\SetFigFont{12}{14.4}{\rmdefault}{\mddefault}{\updefault}{\color[rgb]{0,0,0}$1$}%
}}}}}
\put(-2871,-21488){\rotatebox{120.0}{\makebox(0,0)[b]{\smash{{\SetFigFont{12}{14.4}{\rmdefault}{\mddefault}{\updefault}{\color[rgb]{0,0,0}$2$}%
}}}}}
\put(799,-21333){\rotatebox{240.0}{\makebox(0,0)[b]{\smash{{\SetFigFont{12}{14.4}{\rmdefault}{\mddefault}{\updefault}{\color[rgb]{0,0,0}$1$}%
}}}}}
\put(462,-21917){\rotatebox{240.0}{\makebox(0,0)[b]{\smash{{\SetFigFont{12}{14.4}{\rmdefault}{\mddefault}{\updefault}{\color[rgb]{0,0,0}$2$}%
}}}}}
\put(2866,-22093){\rotatebox{120.0}{\makebox(0,0)[b]{\smash{{\SetFigFont{12}{14.4}{\rmdefault}{\mddefault}{\updefault}{\color[rgb]{0,0,0}$1$}%
}}}}}
\put(2529,-21509){\rotatebox{120.0}{\makebox(0,0)[b]{\smash{{\SetFigFont{12}{14.4}{\rmdefault}{\mddefault}{\updefault}{\color[rgb]{0,0,0}$2$}%
}}}}}
\put(3886,-18809){\makebox(0,0)[b]{\smash{{\SetFigFont{12}{14.4}{\rmdefault}{\mddefault}{\updefault}{\color[rgb]{0,0,0}$1$}%
}}}}
\put(4561,-18809){\makebox(0,0)[b]{\smash{{\SetFigFont{12}{14.4}{\rmdefault}{\mddefault}{\updefault}{\color[rgb]{0,0,0}$2$}%
}}}}
\put(-1182,-21834){\makebox(0,0)[b]{\smash{{\SetFigFont{12}{14.4}{\rmdefault}{\mddefault}{\updefault}{\color[rgb]{0,0,0}$\U{1}$}%
}}}}
\put(-1182,-27254){\makebox(0,0)[b]{\smash{{\SetFigFont{12}{14.4}{\rmdefault}{\mddefault}{\updefault}{\color[rgb]{0,0,0}$\S{1}$}%
}}}}
\put(3368,-20014){\makebox(0,0)[b]{\smash{{\SetFigFont{12}{14.4}{\rmdefault}{\mddefault}{\updefault}{\color[rgb]{0,0,0}$\U{2}$}%
}}}}
\put(11028,-20014){\makebox(0,0)[b]{\smash{{\SetFigFont{12}{14.4}{\rmdefault}{\mddefault}{\updefault}{\color[rgb]{0,0,0}$\U{3}$}%
}}}}
\put(3368,-25344){\makebox(0,0)[b]{\smash{{\SetFigFont{12}{14.4}{\rmdefault}{\mddefault}{\updefault}{\color[rgb]{0,0,0}$\S{2}$}%
}}}}
\put(11028,-25374){\makebox(0,0)[b]{\smash{{\SetFigFont{12}{14.4}{\rmdefault}{\mddefault}{\updefault}{\color[rgb]{0,0,0}$\S{3}$}%
}}}}
\end{picture}%

%% file: Figures/meander.pstex_t
\begin{picture}(0,0)%
\includegraphics[scale=\figscale]{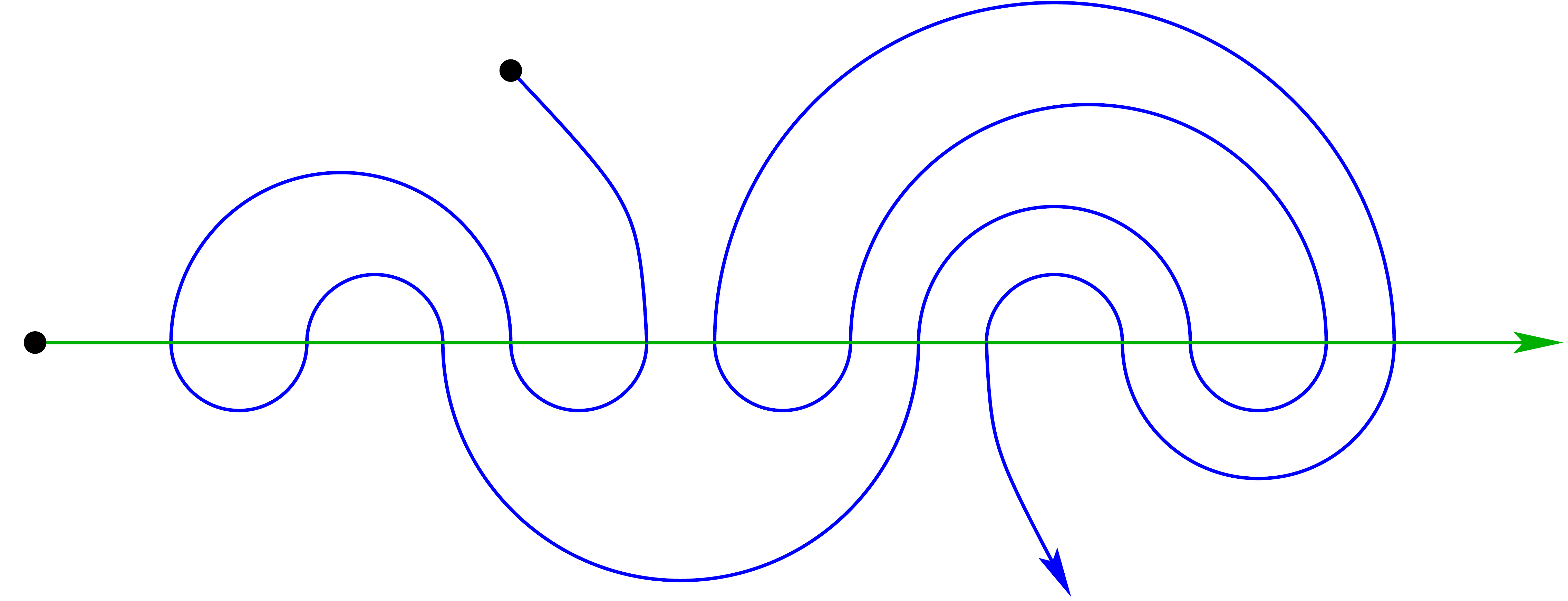}%
\end{picture}%
\setlength{\unitlength}{4144sp}%

\multiply\unitlength by \magproz
\divide\unitlength by 100

\begingroup\makeatletter\ifx\SetFigFont\undefined%
\gdef\SetFigFont#1#2#3#4#5{%
  \reset@font\fontsize{#1}{#2pt}%
  \fontfamily{#3}\fontseries{#4}\fontshape{#5}%
  \selectfont}%
\fi\endgroup%
\begin{picture}(20759,7950)(2236,-10680)
\put(2251,-6811){\makebox(0,0)[b]{\smash{{\SetFigFont{12}{14.4}{\rmdefault}{\mddefault}{\updefault}{\color[rgb]{0,0,0}$v$}%
}}}}
\put(8551,-3211){\makebox(0,0)[b]{\smash{{\SetFigFont{12}{14.4}{\rmdefault}{\mddefault}{\updefault}{\color[rgb]{0,0,0}$w$}%
}}}}
\end{picture}%

%% file: Figures/Fall2+3.pstex_t
\begin{picture}(0,0)%
\includegraphics[scale=\figscale]{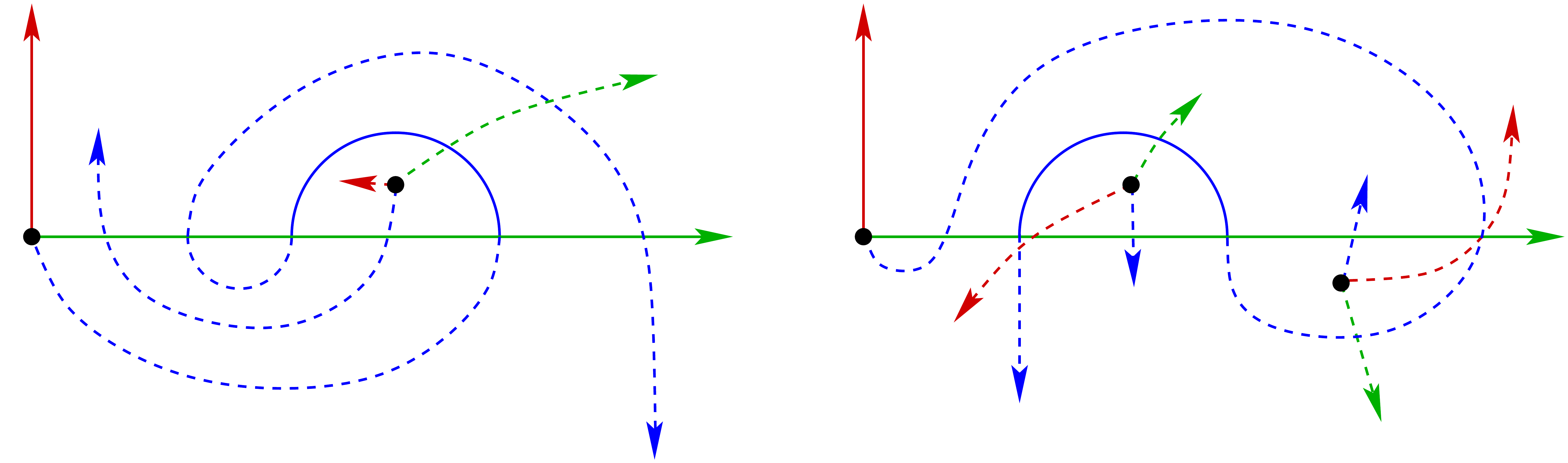}%
\end{picture}%
\setlength{\unitlength}{4144sp}%

\multiply\unitlength by \magproz
\divide\unitlength by 100

\begingroup\makeatletter\ifx\SetFigFont\undefined%
\gdef\SetFigFont#1#2#3#4#5{%
  \reset@font\fontsize{#1}{#2pt}%
  \fontfamily{#3}\fontseries{#4}\fontshape{#5}%
  \selectfont}%
\fi\endgroup%
\begin{picture}(27144,8011)(3951,-11178)
\put(6864,-7014){\makebox(0,0)[b]{\smash{{\SetFigFont{12}{14.4}{\rmdefault}{\mddefault}{\updefault}{\color[rgb]{0,0,0}$j$}%
}}}}
\put(12838,-7090){\makebox(0,0)[b]{\smash{{\SetFigFont{12}{14.4}{\rmdefault}{\mddefault}{\updefault}{\color[rgb]{0,0,0}$i$}%
}}}}
\put(15349,-7104){\makebox(0,0)[b]{\smash{{\SetFigFont{12}{14.4}{\rmdefault}{\mddefault}{\updefault}{\color[rgb]{0,0,0}$k$}%
}}}}
\put(3966,-6953){\makebox(0,0)[b]{\smash{{\SetFigFont{12}{14.4}{\rmdefault}{\mddefault}{\updefault}{\color[rgb]{0,0,0}$v$}%
}}}}
\put(18386,-6953){\makebox(0,0)[b]{\smash{{\SetFigFont{12}{14.4}{\rmdefault}{\mddefault}{\updefault}{\color[rgb]{0,0,0}$v$}%
}}}}
\put(25460,-7061){\makebox(0,0)[b]{\smash{{\SetFigFont{12}{14.4}{\rmdefault}{\mddefault}{\updefault}{\color[rgb]{0,0,0}$i$}%
}}}}
\put(26455,-7887){\makebox(0,0)[b]{\smash{{\SetFigFont{12}{14.4}{\rmdefault}{\mddefault}{\updefault}{\color[rgb]{0,0,0}$w$}%
}}}}
\put(11398,-6854){\makebox(0,0)[b]{\smash{{\SetFigFont{12}{14.4}{\rmdefault}{\mddefault}{\updefault}{\color[rgb]{0,0,0}$u$}%
}}}}
\put(24092,-6740){\makebox(0,0)[b]{\smash{{\SetFigFont{12}{14.4}{\rmdefault}{\mddefault}{\updefault}{\color[rgb]{0,0,0}$u$}%
}}}}
\put(29953,-7812){\makebox(0,0)[b]{\smash{{\SetFigFont{12}{14.4}{\rmdefault}{\mddefault}{\updefault}{\color[rgb]{0,0,0}$j$}%
}}}}
\end{picture}%

%% file: Figures/eight-configs.pstex_t
\begin{picture}(0,0)%
\includegraphics[scale=\figscale]{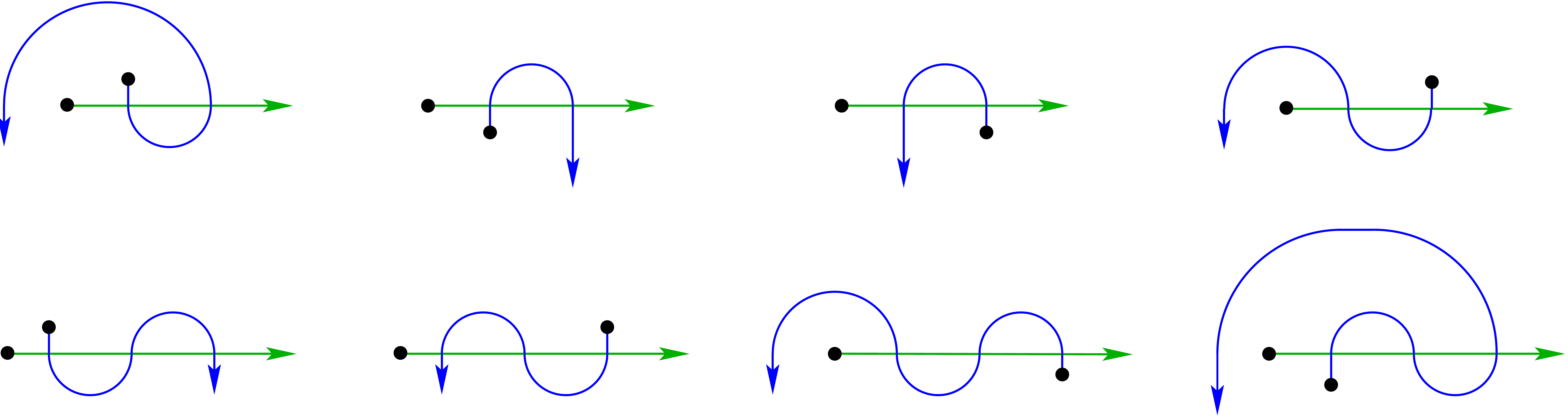}%
\end{picture}%
\setlength{\unitlength}{4144sp}%

\multiply\unitlength by \magproz
\divide\unitlength by 100

\begingroup\makeatletter\ifx\SetFigFont\undefined%
\gdef\SetFigFont#1#2#3#4#5{%
  \reset@font\fontsize{#1}{#2pt}%
  \fontfamily{#3}\fontseries{#4}\fontshape{#5}%
  \selectfont}%
\fi\endgroup%
\begin{picture}(34108,9089)(5087,-31155)
\put(23851,-22561){\makebox(0,0)[b]{\smash{{\SetFigFont{34}{40.8}{\rmdefault}{\bfdefault}{\updefault}{\color[rgb]{0,0,0}III}%
}}}}
\put(5401,-27961){\makebox(0,0)[b]{\smash{{\SetFigFont{34}{40.8}{\rmdefault}{\bfdefault}{\updefault}{\color[rgb]{0,0,0}V}%
}}}}
\put(34201,-28411){\makebox(0,0)[b]{\smash{{\SetFigFont{34}{40.8}{\rmdefault}{\bfdefault}{\updefault}{\color[rgb]{0,0,0}VIII}%
}}}}
\put(5401,-22561){\makebox(0,0)[b]{\smash{{\SetFigFont{34}{40.8}{\rmdefault}{\bfdefault}{\updefault}{\color[rgb]{0,0,0}I}%
}}}}
\put(14851,-22561){\makebox(0,0)[b]{\smash{{\SetFigFont{34}{40.8}{\rmdefault}{\bfdefault}{\updefault}{\color[rgb]{0,0,0}II}%
}}}}
\put(14851,-27961){\makebox(0,0)[b]{\smash{{\SetFigFont{34}{40.8}{\rmdefault}{\bfdefault}{\updefault}{\color[rgb]{0,0,0}VI}%
}}}}
\put(23851,-27961){\makebox(0,0)[b]{\smash{{\SetFigFont{34}{40.8}{\rmdefault}{\bfdefault}{\updefault}{\color[rgb]{0,0,0}VII}%
}}}}
\put(34201,-22561){\makebox(0,0)[b]{\smash{{\SetFigFont{34}{40.8}{\rmdefault}{\bfdefault}{\updefault}{\color[rgb]{0,0,0}IV}%
}}}}
\end{picture}%

%% file: Figures/back-configs.pstex_t
\begin{picture}(0,0)%
\includegraphics[scale=\figscale]{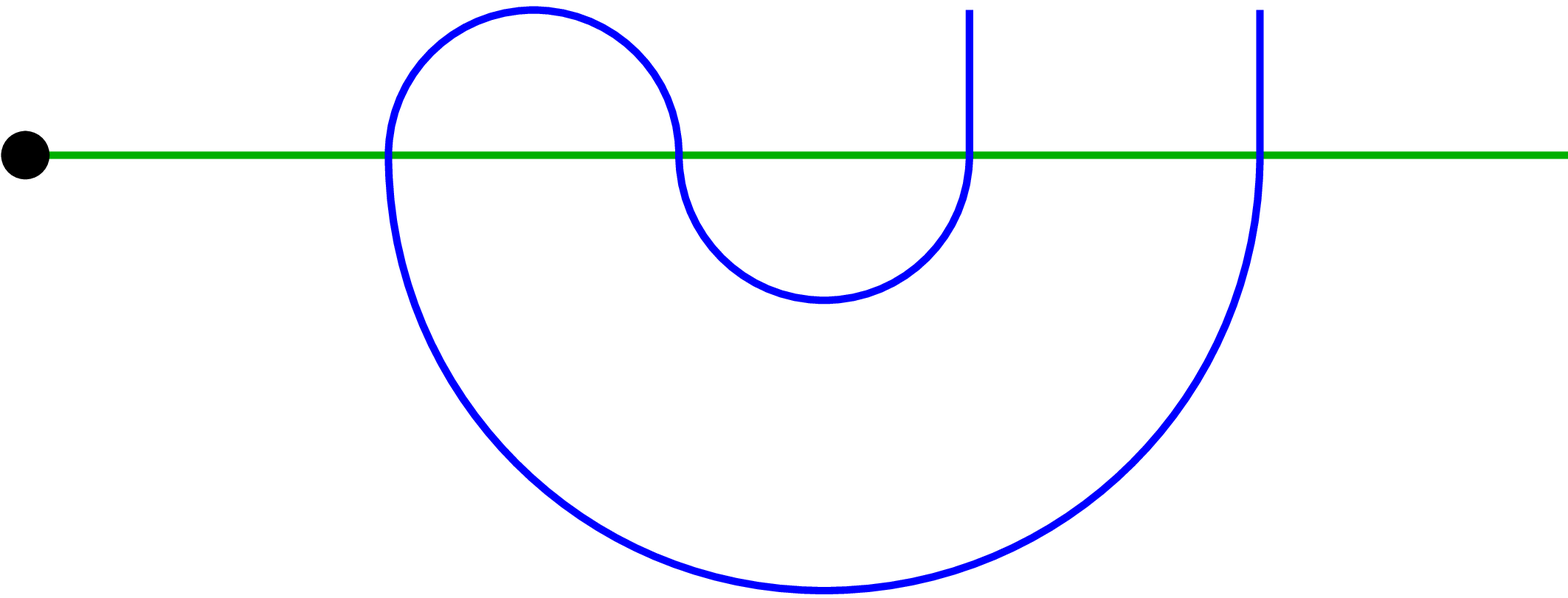}%
\end{picture}%
\setlength{\unitlength}{4144sp}%

\multiply\unitlength by \magproz
\divide\unitlength by 100

\begingroup\makeatletter\ifx\SetFigFont\undefined%
\gdef\SetFigFont#1#2#3#4#5{%
  \reset@font\fontsize{#1}{#2pt}%
  \fontfamily{#3}\fontseries{#4}\fontshape{#5}%
  \selectfont}%
\fi\endgroup%
\begin{picture}(47451,4589)(5694,-32955)
\put(7201,-28861){\makebox(0,0)[b]{\smash{{\SetFigFont{34}{40.8}{\rmdefault}{\bfdefault}{\updefault}{\color[rgb]{0,0,0}TB 1}%
}}}}
\put(19351,-30211){\makebox(0,0)[b]{\smash{{\SetFigFont{34}{40.8}{\rmdefault}{\bfdefault}{\updefault}{\color[rgb]{0,0,0}TB 2}%
}}}}
\put(43201,-30211){\makebox(0,0)[b]{\smash{{\SetFigFont{34}{40.8}{\rmdefault}{\bfdefault}{\updefault}{\color[rgb]{0,0,0}TB 4}%
}}}}
\put(31501,-32911){\makebox(0,0)[b]{\smash{{\SetFigFont{34}{40.8}{\rmdefault}{\bfdefault}{\updefault}{\color[rgb]{0,0,0}TB 3}%
}}}}
\end{picture}%

%% file: Figures/caseTB1.pstex_t
\begin{picture}(0,0)%
\includegraphics[scale=\figscale]{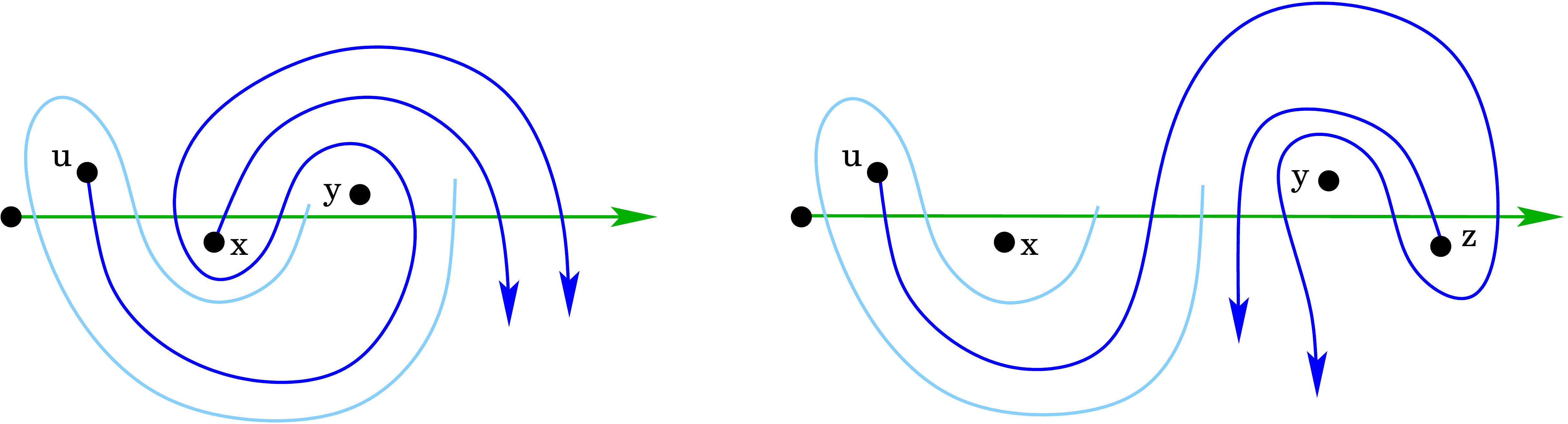}%
\end{picture}%
\setlength{\unitlength}{4144sp}%

\multiply\unitlength by \magproz
\divide\unitlength by 100

\begingroup\makeatletter\ifx\SetFigFont\undefined%
\gdef\SetFigFont#1#2#3#4#5{%
  \reset@font\fontsize{#1}{#2pt}%
  \fontfamily{#3}\fontseries{#4}\fontshape{#5}%
  \selectfont}%
\fi\endgroup%
\begin{picture}(22238,6014)(-156,-38099)
\end{picture}%

%% file: Figures/caseTB1second.pstex_t
\begin{picture}(0,0)%
\includegraphics[scale=\figscale]{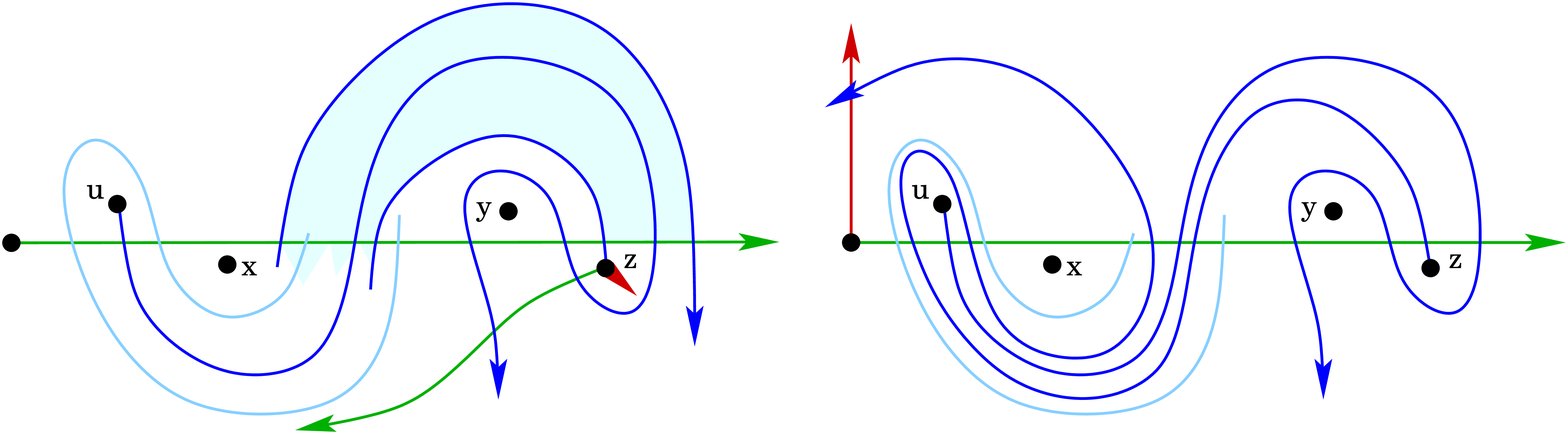}%
\end{picture}%
\setlength{\unitlength}{4144sp}%

\multiply\unitlength by \magproz
\divide\unitlength by 100

\begingroup\makeatletter\ifx\SetFigFont\undefined%
\gdef\SetFigFont#1#2#3#4#5{%
  \reset@font\fontsize{#1}{#2pt}%
  \fontfamily{#3}\fontseries{#4}\fontshape{#5}%
  \selectfont}%
\fi\endgroup%
\begin{picture}(25643,7070)(-3098,-38278)
\end{picture}%

%% file: Figures/caseTB2.pstex_t
\begin{picture}(0,0)%
\includegraphics[scale=\figscale]{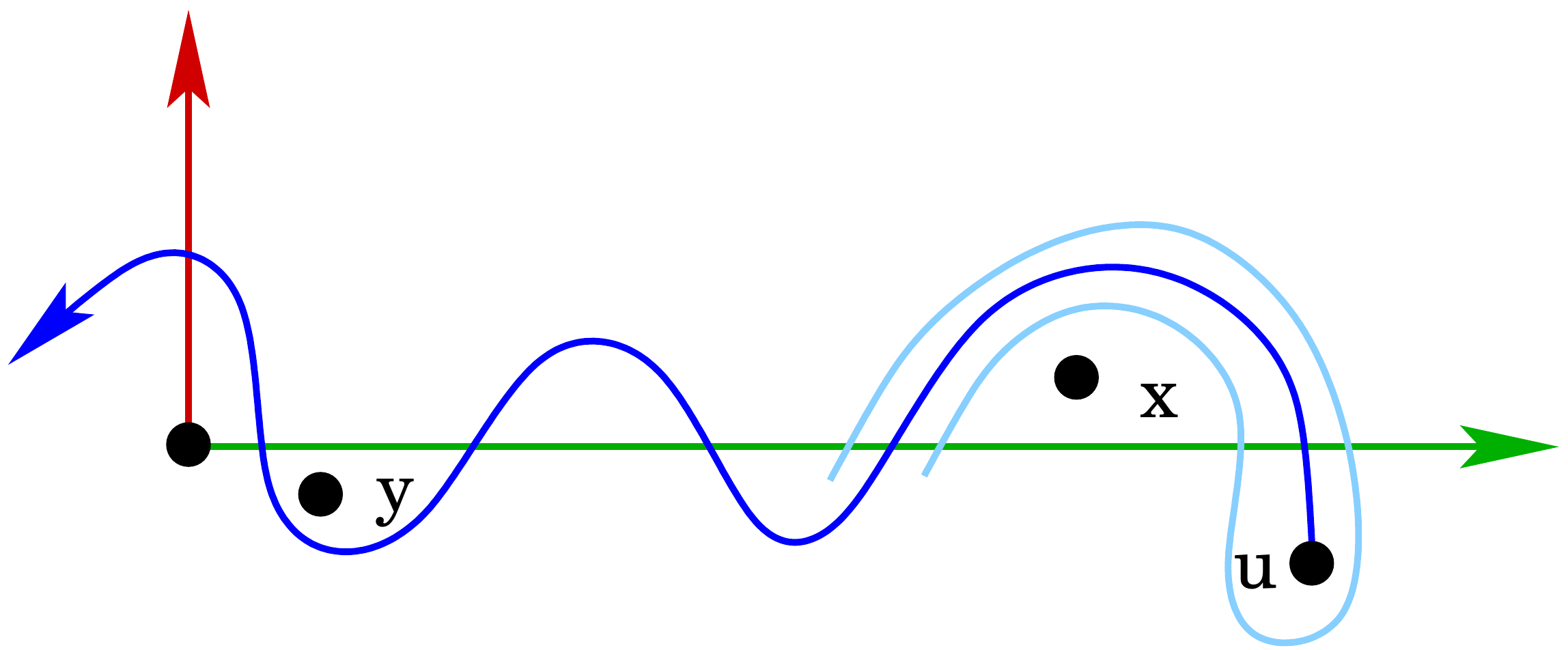}%
\end{picture}%
\setlength{\unitlength}{4144sp}%

\multiply\unitlength by \magproz
\divide\unitlength by 100

\begingroup\makeatletter\ifx\SetFigFont\undefined%
\gdef\SetFigFont#1#2#3#4#5{%
  \reset@font\fontsize{#1}{#2pt}%
  \fontfamily{#3}\fontseries{#4}\fontshape{#5}%
  \selectfont}%
\fi\endgroup%
\begin{picture}(10491,4339)(-134,-44045)
\end{picture}%

%% file: Figures/caseTB4.pstex_t
\begin{picture}(0,0)%
\includegraphics[scale=\figscale]{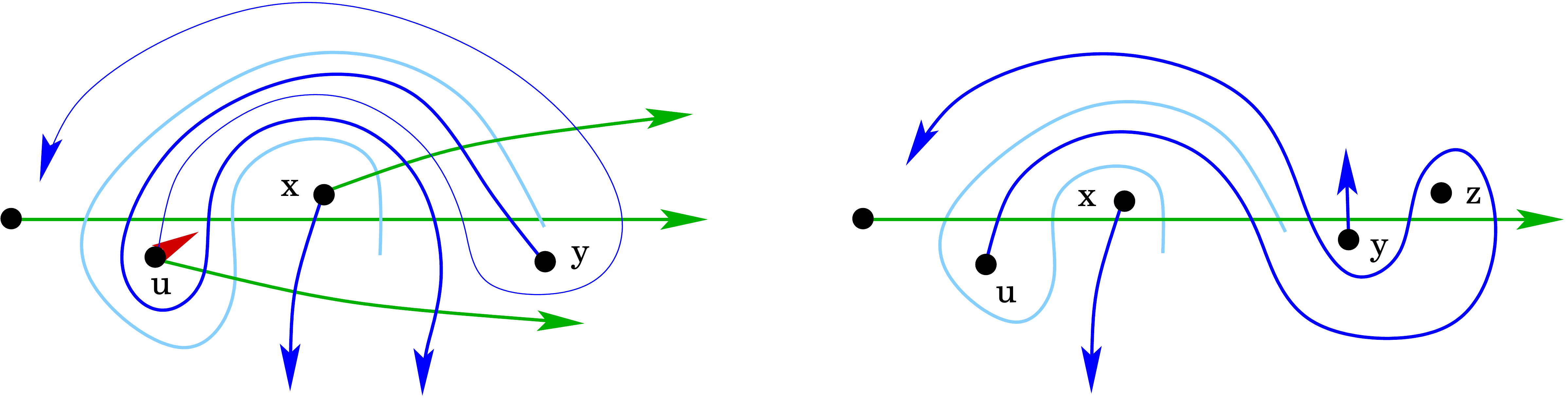}%
\end{picture}%
\setlength{\unitlength}{4144sp}%

\multiply\unitlength by \magproz
\divide\unitlength by 100

\begingroup\makeatletter\ifx\SetFigFont\undefined%
\gdef\SetFigFont#1#2#3#4#5{%
  \reset@font\fontsize{#1}{#2pt}%
  \fontfamily{#3}\fontseries{#4}\fontshape{#5}%
  \selectfont}%
\fi\endgroup%
\begin{picture}(21907,5574)(8206,-52338)
\end{picture}%

%% file: Figures/non-stretch-small.pstex_t
\begin{picture}(0,0)%
\includegraphics[scale=\figscale]{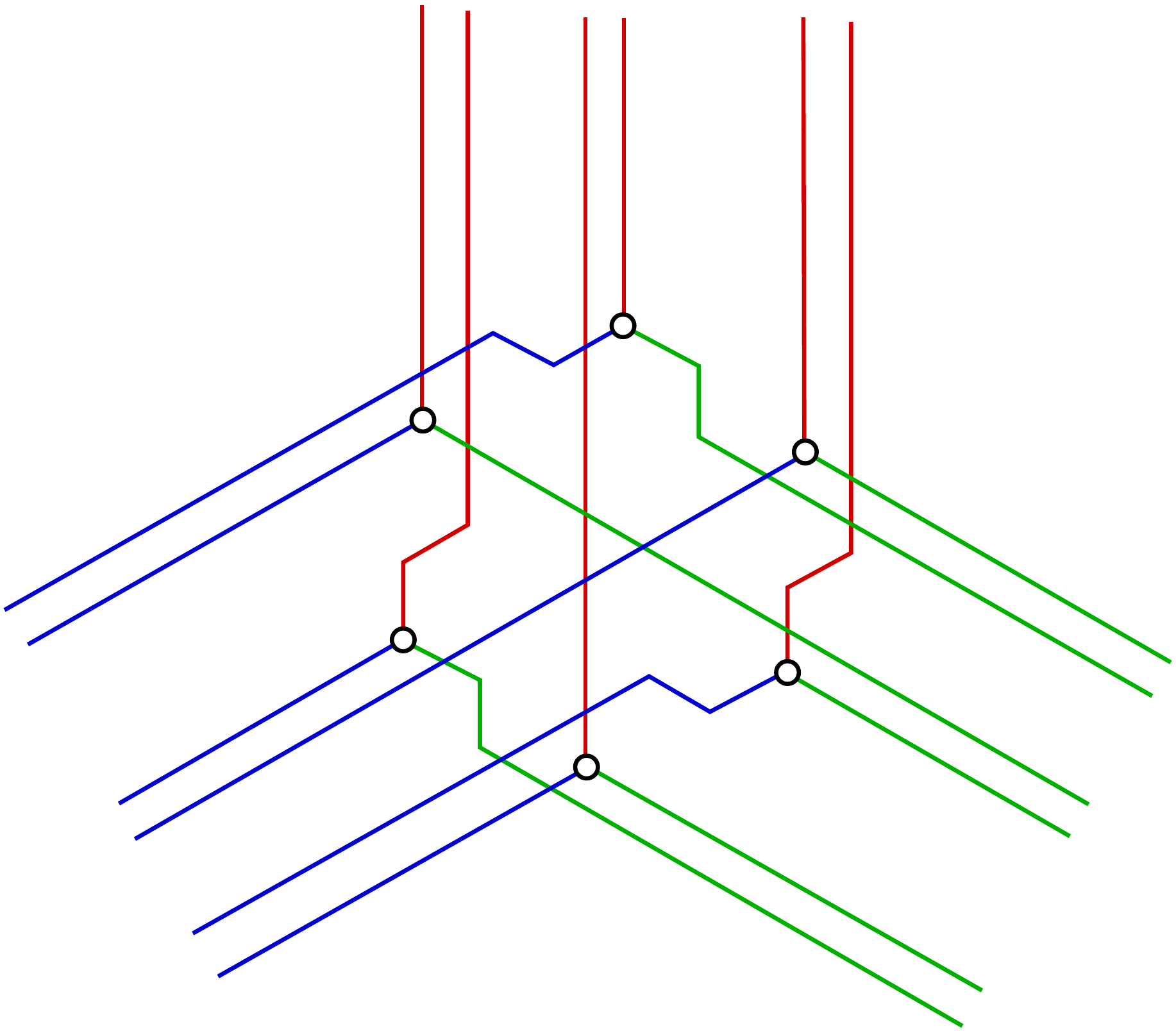}%
\end{picture}%
\setlength{\unitlength}{4144sp}%

\multiply\unitlength by \magproz
\divide\unitlength by 100

\begingroup\makeatletter\ifx\SetFigFont\undefined%
\gdef\SetFigFont#1#2#3#4#5{%
  \reset@font\fontsize{#1}{#2pt}%
  \fontfamily{#3}\fontseries{#4}\fontshape{#5}%
  \selectfont}%
\fi\endgroup%
\begin{picture}(8382,7346)(34663,-46267)
\put(37550,-42353){\makebox(0,0)[b]{\smash{{\SetFigFont{12}{14.4}{\rmdefault}{\mddefault}{\updefault}{\color[rgb]{0,0,0}$q_1$}%
}}}}
\put(39927,-43669){\makebox(0,0)[b]{\smash{{\SetFigFont{12}{14.4}{\rmdefault}{\mddefault}{\updefault}{\color[rgb]{0,0,0}$q_4$}%
}}}}
\put(39166,-44384){\makebox(0,0)[b]{\smash{{\SetFigFont{12}{14.4}{\rmdefault}{\mddefault}{\updefault}{\color[rgb]{0,0,0}$q_5$}%
}}}}
\put(37181,-43461){\makebox(0,0)[b]{\smash{{\SetFigFont{12}{14.4}{\rmdefault}{\mddefault}{\updefault}{\color[rgb]{0,0,0}$q_6$}%
}}}}
\put(40167,-41915){\makebox(0,0)[b]{\smash{{\SetFigFont{12}{14.4}{\rmdefault}{\mddefault}{\updefault}{\color[rgb]{0,0,0}$q_3$}%
}}}}
\put(39270,-41628){\makebox(0,0)[b]{\smash{{\SetFigFont{12}{14.4}{\rmdefault}{\mddefault}{\updefault}{\color[rgb]{0,0,0}$q_2$}%
}}}}
\end{picture}%